\newif\iflong
\longtrue 

\documentclass[a4paper,USenglish,thm-restate]{lipics-v2021}

\usepackage[vlined,linesnumbered,noresetcount]{algorithm2e}
\usepackage{float}
\usepackage{tikz}

\usepackage{cite}

\usetikzlibrary{shapes.geometric}
\usetikzlibrary{decorations.pathreplacing,calligraphy}
\usetikzlibrary{math}

\SetKwBlock{SubAlgoBlock}{}{end}
\SetAlFnt{\small}

\SetKwRepeat{PUntil}{do}{periodically until}
\SetKwRepeat{DoUntil}{do}{until}
\SetKwBlock{uTask}{start periodic task}{}

\SetKwComment{tcp}{ \%}{}

\newcommand{\tr}[2]{\iflong{}\S#1\else{}\cite[\S#2]{ext}\fi}

\let\oldnl\nl

\newcommand{\nonl}{\renewcommand{\nl}{\let\nl\oldnl}}

\newcommand{\stale}{{\sf stale}}

\newcommand{\valu}{{\sf value}}

\newcommand{\writ}{{\tt write}}
\newcommand{\rd}{{\tt read}}
\newcommand{\cnt}{\ensuremath{\mathit{cnt}}}

\newcommand{\rdint}{{\tt read\_quorum}}

\newcommand{\wrint}{{\tt write\_quorum}}

\newcommand{\newid}{\mathtt{get\_unique\_id}}

\newcommand{\inc}{{\sf inc}}

\newcommand{\incvar}{\ensuremath{\mathit{inc}}}

\newcommand{\opval}{{\sf value}}

\newcommand{\Wmin}{\wrq_{\mathit{min}}}
\newcommand{\Winc}{\wrq_{\mathit{inc}}}

\newcommand{\pmin}{p_{\mathit{min}}}

\newcommand{\tmin}{t_{\mathit{min}}}

\newcommand{\ReadAMS}{{\tt ReadAMS}}
\newcommand{\WriteAMS}{{\tt WriteAMS}}
\newcommand{\amnesic}{{\sf amnesic}}

\newcommand{\finish}{\mathit{finish}}
\newcommand{\rdcv}{\mathit{read\_cv}}
\newcommand{\acks}{\mathit{acks}}
\newcommand{\resptime}{\mathit{resp\_time}}
\newcommand{\var}{\mathit{var}}
\newcommand{\twinc}{t_\mathit{write\_inc}}
\newcommand{\trstate}{t_\mathit{read\_state}}
\newcommand{\tlrstate}{t_\mathit{last\_read\_state}}

\newcommand{\tcrash}{t_\mathit{crash}}
\newcommand{\tstale}{t_\mathit{clear\_stale}}
\newcommand{\tresp}{t_\mathit{resp}}
\newcommand{\tlc}{t_\mathit{last\_crash}}

\newcommand{\wrq}{\mathit{WQ}}
\newcommand{\rdq}{\mathit{RQ}}

\newcommand{\wrr}{W}
\newcommand{\rdr}{R}
\newcommand{\allrdr}{\mathsf{Reads}}
\newcommand{\allwrr}{\mathsf{Writes}}
\newcommand{\allreps}{\mathcal{R}}
\newcommand{\ups}{U}

\newcommand{\allprocs}{\mathcal{P}}
\newcommand{\world}{\mathcal{E}}

\newcommand{\ndown}{d}

\newcommand{\ncrash}{c}
\newcommand{\incnum}{g}
\newcommand{\tsset}{\mathcal{T}}
\newcommand{\incset}{\mathcal{I}}

\newcommand{\cv}{\ensuremath{\mathsf{CV}}}
\newcommand{\cvvar}{\ensuremath{\mathsf{cv}}}

\newcommand{\precv}{\ensuremath{\mathsf{preCV}}}
\newcommand{\precvvar}{\mathsf{pre\_cv}}

\newcommand{\TS}{\ensuremath{\mathsf{TS}}}
\newcommand{\Tsvar}{\ensuremath{\mathsf{ts}}}

\newcommand{\Valvar}{\ensuremath{\mathsf{val}}}

\newcommand{\id}{\ensuremath{\mathsf{id}}}
\newcommand{\idvar}{\ensuremath{\mathit{id}}}

\newcommand{\proc}{\ensuremath{\mathsf{proc}}}
\newcommand{\crv}{\ensuremath{\mathsf{cv}}}
\newcommand{\wacks}{\ensuremath{\mathsf{w}\text{-}\mathsf{acks}}}
\newcommand{\racks}{\ensuremath{\mathsf{r}\text{-}\mathsf{acks}}}

\newcommand{\nextinc}{\ensuremath{\mathit{next\_inc}}}

\newcommand{\tsval}{\ensuremath{\mathsf{TSVal}}}

\newcommand{\ST}{\ensuremath{\mathsf{State}}}

\newcommand{\req}{\ensuremath{\mathit{req}}}
\newcommand{\reqq}{\ensuremath{\mathsf{req}}}
\newcommand{\resp}{\ensuremath{\mathit{r}}}

\newcommand{\tmpts}{\ensuremath{\mathit{ts}}}
\newcommand{\tmpval}{\ensuremath{v}}
\newcommand{\tmpcv}{\ensuremath{\mathit{cv}}}

\newcommand{\WRITE}{{\tt WRITE}}
\newcommand{\WRITEACK}{{\tt WRITE\_ACK}}
\newcommand{\READ}{{\tt READ}}
\newcommand{\READACK}{{\tt READ\_ACK}}

\newcommand{\READREPLICA}{{\tt READ\_REPLICA}}
\newcommand{\REPLY}{{\tt REPLY}}
\newcommand{\WRITEREPLICA}{{\tt WRITE\_REPLICA}}
\newcommand{\SUCCESS}{{\tt SUCCESS}}

\newcommand{\State}{\ensuremath{\mathsf{state}}}

\newcommand{\TRUE}{\text{\sc true}}
\newcommand{\FALSE}{\text{\sc false}}

\newcommand{\return}{\text{\bf return}}
\newcommand{\from}{\text{\bf from}}
\newcommand{\pre}{\text{\bf pre:}}

\newcommand{\onrestart}{\text{\bf on restart}}

\newcommand{\fn}{\text{\bf function}}

\newcommand{\whenrc}{\text{\bf when received}}

\newcommand{\rc}{\textbf{received}}

\newcommand{\send}{\text{\bf send}}

\newcommand{\alg}[2]{\mathcal{A}_{#1}^{#2}}
\newcommand{\algname}{\textsc{TEE-Rex}}
\newcommand{\newresps}{K}

\newcommand{\CRR}{\mathsf{CRR}}

\newcommand{\FCRR}[1]{\mathsf{CRR}(#1)}

\newcommand{\Fdsym}{\mathsf{CRR}\text{-}\mathsf{D}}
\newcommand{\Fd}[1]{\Fdsym(#1)}

\newcommand{\rdqsize}{q_r}
\newcommand{\wrqsize}{q_w}
\newcommand{\vis}{\mathsf{vis}}

\newcommand{\newcode}[1]{{\color{blue}#1}}

\newcommand{\CC}{\mathsf{CC}}

\renewcommand{\_}{\texttt{\textunderscore}}

\newcommand{\WR}{{\sf wr}}
\newcommand{\WW}{{\sf ww}}
\newcommand{\RW}{{\sf rw}}
\newcommand{\RT}{{\sf rt}}

\EventEditors{Dariusz R. Kowalski}
\EventNoEds{1}
\EventLongTitle{39th International Symposium on Distributed Computing (DISC 2025)}
\EventShortTitle{DISC 2025}
\EventAcronym{DISC}
\EventYear{2025}
\EventDate{October 27--31, 2025}
\EventLocation{Berlin, Germany}
\EventLogo{}
\SeriesVolume{356}
\ArticleNo{18}

\iflong
\title{TEE Is Not a Healer: Rollback-Resistant Reliable Storage (Extended Version)}
\else
\title{TEE Is Not a Healer: Rollback-Resistant Reliable Storage}
\fi

\author{Sadegh Keshavarzi}{University of Surrey, Guildford, UK}{}{0009-0000-9386-7401}{}
\author{Gregory Chockler}{University of Surrey, Guildford, UK}{}{0000-0001-6700-9235}{}
\author{Alexey Gotsman}{IMDEA Software Institute, Madrid, Spain}{}{}{}

\ccsdesc{Theory of computation~Distributed computing models}
\keywords{Trusted execution environments, fault tolerance, crash recovery}

\authorrunning{S. Keshavarzi, G. Chockler, and A. Gotsman}
\titlerunning{TEE Is Not a Healer: Rollback-Resistant Reliable Storage}

\Copyright{Sadegh Keshavarzi, Gregory Chockler, and Alexey Gotsman}

\iflong
{}
\else
\relatedversiondetails[cite={ext}]{Extended 
Version}{https://arxiv.org/abs/2505.18648}
\fi

\funding{This work was partially supported by the projects BYZANTIUM and DECO
  funded by MCIN/AEI, and REDONDA funded by the CHIST-ERA network.}

\begin{document}

\nolinenumbers

\maketitle

\begin{abstract} 
  Recent advances in secure hardware technologies, such as Intel SGX or ARM
  TrustZone, offer an opportunity to substantially reduce the costs of Byzantine
  fault-tolerance by placing the program code and state within a secure enclave
  known as a \emph{Trusted Execution Environment (TEE)}. However, the protection
  offered by a TEE only applies during program execution. Once power is switched
  off, the non-volatile portion of the program state becomes vulnerable to
  \emph{rollback} attacks wherein it is undetectably reverted to an older
  version. In this paper we consider the problem of implementing reliable
  read/write registers out of failure-prone replicas subject to state
  rollbacks. To this end, we introduce a new unified model that captures
  multiple failure types that can affect a TEE-based system and establish tight
  bounds on the fault-tolerance of register constructions in this model. We
  consider both the static case, where failure thresholds hold throughout the
  entire execution, and the dynamic case, where any number of replicas can roll
  back, provided these failures do not occur too often. Our dynamic register
  emulation algorithm, \algname\/, provides the first correct implementation of
  a distributed state recovery procedure that requires neither durable storage
  nor specialized hardware, such as trusted monotonic counters.
\end{abstract}

\section{Introduction}
\label{sec:intro}

Tolerating Byzantine failures plays a major role in designing modern reliable
distributed systems, and in particular blockchains. Unfortunately, tolerating
Byzantine failures is expensive in terms of message complexity~\cite{dolev1985bounds},
latency~\cite{abraham2021good,guo2023tetrabft},
and the number of replicas 
required~\cite{pease1980reaching,lynch1986easy,dwork1988consensus}. 
Recent advances in secure hardware
technologies, such as Intel SGX or ARM TrustZone, offer an opportunity to
substantially reduce these costs by placing the program's code
and a portion of its state in a secure enclave known as a \emph{Trusted Execution
	Environments (TEE)}.
Applications can furthermore detect any runtime alterations to the state stored
outside the enclave using integrity metadata stored inside it. These mechanisms
together can then be used to convert Byzantine failures to simple crash or
omission failures, which are less expensive to
tolerate~\cite{veronese2013efficient,liu2019fastbft,behl2017hybrids,DBLP:conf/ccs/0005DNRZ22}.

However, a TEE is a not a panacea, because
the protection it offers only applies during program
execution. Once power is switched off, the non-volatile portion of the
program state becomes vulnerable to tampering by an attacker.
To address this, hardware manufacturers provide a \emph{sealing} 
mechanism~\cite{costan2016intel}
that allows programs to encrypt a portion of their state using a
device-specific key before storing it in non-volatile storage. When the
machine restarts, the TEE can check whether the sealed state
was indeed written by it in the past. 
However, it cannot determine whether this state is the most recent version. 
This limitation makes TEE-based systems vulnerable to \emph{rollback} attacks where the sealed 
state 
is undetectably replaced with an older 
version~\cite{203712,strackx2016ariadne,brandenburger2017rollback,DBLP:conf/ccs/0005DNRZ22,DBLP:conf/ndss/DinisD023}.

In this paper we study theoretical foundations for building reliable
distributed systems in the presence of rollbacks.
We start by introducing \emph{crash-restart-rollback} ($\CRR$), a new unified
failure model that formalizes rollbacks alongside other types of failures that
can occur in a TEE-based system. In this model, processes can use local
non-volatile storage to persist their states and are assumed to behave correctly
while they are executing. A $\CRR$-faulty process can permanently crash, crash and
restart infinitely often, or crash and restart with the content of its
non-volatile storage being undetectably reverted to an older version.



We then study the problem of implementing reliable read/write registers in an
asynchronous message-passing system where the register's state
is stored at a collection of $n$ $\CRR$ failure-prone \emph{replicas} accessed
by arbitrarily many crash-prone \emph{clients}. 
We consider two sub-classes of the above $\CRR$ failure model for replicas -- \emph{static} and
\emph{dynamic}. Static failure models assume the existence
of a priori fixed thresholds restricting the number of replicas
that can experience failures of various types
in every execution; this is similar to the classical crash or Byzantine
failure models. The dynamic failure models generalize the static ones by 
permitting any number of replicas to roll back in exchange for requiring that
these failures do not occur too often.



\subparagraph*{Static failure models.}  For the static failure models, we give a
fine-grained characterization of failure resilience, time complexity, and
non-volatile storage requirements of a register implementation in terms of
several thresholds: $k$ on the total number of any $\CRR$ failures; $r\le k$ on
the number of rollback failures; and $b\le n$ on the number of correct replicas
that can crash at least once, but eventually stay up. Our main result
establishes the following: a wait-free atomic
multi-writer/multi-reader (MWMR) register can be implemented in the static model
if and only if $n \ge 2k + \min(b, r) + 1$.  Our algorithm matching this bound
has the latency of $4$ message delays for both
writes and reads, which is the same as for crash fault-tolerant register
implementations~\cite{rambo,grishas-abd} and is thus optimal. This stands in a
sharp contrast with the Byzantine-resilient register constructions, where 
write and read latencies of $4$ and $8$ message delays are inherent even for
implementing a single-writer/multi-reader (SWMR) wait-free atomic
register~\cite{dgmsv11}.

Our lower bound implies that, when $r \le b$, rollback failures are no different
from Byzantine failures in terms of failure resilience: the number of replicas
$n \ge 2k + r + 1$ required in this case is the same as that needed to implement
Byzantine-resilient registers when at most $r$ out of $k$ faulty replicas can
behave arbitrarily~\cite{gv06}. On the other hand, the two models are separated
when $b < r$: in this case, a register can be implemented in $\CRR$ with
$2k + b + 1$ replicas, which is strictly fewer than $2k + r + 1$ required by
Byzantine fault-tolerant implementations.

We also establish a lower bound on non-volatile storage consumption: when
$2k + r + 1 \le n < 2k + b + 1$, at least $2k + r + 1$ replicas must store their
states persistently on their local non-volatile storage. Since this storage is
not protected by TEEs, in practice applications must rely on expensive
mechanisms to monitor its integrity at runtime, e.g., using TEE-protected
integrity metadata, such as Merkle trees. Our results show that for some
resilience levels these overheads are unavoidable.

\subparagraph*{Dynamic failure models.}
For the dynamic failure models, we first show that no 
implementation of a single-writer/single-reader (SWSR) safe register can  use fewer than 
$\ndown + \ncrash + 1$ replicas, where $\ndown$ bounds 
the number of replicas that either crash permanently or never stop crashing and
recovering, and $\ncrash\ge \ndown$ bounds the number of replicas that
crash at least once. While it is well-known that 
no register implementation can exist without a majority of 
replicas being eventually up ($n \ge 2\ndown + 1$), our result
further refines this lower bound for the case when $\ncrash > \ndown$.

For a matching upper bound, we propose an algorithm, called $\algname$, that implements
an MWMR atomic register using $n \ge \ndown + \ncrash + 1$ replicas. The algorithm
is always safe and is wait-free in executions where: $n-\ndown$
replicas eventually stop crashing; and either replicas do not
crash too often or $n-\ncrash$ replicas never crash at all.
Unlike in the static model, these conditions place no restrictions on the number
of replicas that can roll back.
The algorithm is
parameterized by $\ndown$, but is unaware of $\ncrash$. It achieves the optimal
time complexity of $4$ message delays for both writes and reads in failure-free
executions.


%

The most interesting ingredient of $\algname$ is a novel distributed
recovery protocol that enables crashed replicas to rebuild their state and
become fully operational upon restart.
As we explain in \S\ref{sec:algo2plus}, implementing recovery correctly is
nontrivial: a naive approach of simply querying a read quorum
and adopting the state of the most up-to-date replica is fundamentally
unsafe~\cite{michael2017recovering}. In fact, as we discovered in this work,
some of the prior recoverable register implementations suffer from nontrivial
safety bugs~\cite{DBLP:journals/talg/GuerraouiLPP08,DBLP:conf/ndss/DinisD023}
(\S\ref{sec:related}). Others do not guarantee
liveness~\cite{203712,DBLP:conf/ccs/0005DNRZ22}, or require trusted
hardware-based primitives such as persistent monotonic
counters~\cite{replacement,michael2017recovering}.  In contrast, our $\algname$ algorithm
provides what we believe is the first self-contained and
rollback-resilient construction of a read/write register under dynamic failure
models.

\section{System Model}
\label{sec:model}

We consider an asynchronous message-passing system consisting of a collection of
failure-prone processes $\allprocs=\{p_1, p_2, \ldots\}$ implementing a
high-level object abstraction. The set of processes is partitioned into a set
$\allreps = \{p_1, \ldots, p_n\}$ of $n>1$ \emph{replicas}, and a set
$\{p_{n+1}, p_ {n+2}, \ldots\}$ of possibly infinitely many
\emph{clients}.  Similarly to the classical faulty shared memory model
of~\cite{jct98,agmt95}, the replicas are responsible for storing the object
state, and the clients interact with the replicas to handle the requests
supported by the implemented object type. To study tradeoffs between
failure resilience and non-volatile storage consumption, we assume
that $s\ge 0$ replicas have access to non-volatile storage.

Formally, an object \emph{implementation} is a composition of client and
replica automata. An execution of an implementation is a sequence of states
interleaved with atomic send and receive actions starting from an initial
state. The action atomicity ensures that all non-volatile storage modifications
performed as part of the received message handling are \emph{failure
	atomic}~\cite{atlas14}, i.e., their effects are either persisted in their
entirety or not at all.

\subparagraph*{Failure models.}
A failure model is a predicate over executions.
An execution $\alpha$ is \emph{valid under a failure model} $\mathcal{F}$
(or simply \emph{$\mathcal{F}$-valid}) if $\mathcal{F}(\alpha)$ holds. 
For failure models $\mathcal{F}$ and $\mathcal{F'}$, $\mathcal{F} \preceq \mathcal{F'}$
iff $\forall \alpha.\, \mathcal{F}(\alpha) \implies \mathcal{F'}(\alpha)$.

\subparagraph*{Client failures.}
Clients can experience permanent \emph{crash} failures, but otherwise do not
deviate from their prescribed protocols. A client is \emph{correct} in an execution
if it never crashes, and is \emph{faulty}, otherwise.

\subparagraph*{Replica failures.}
Our baseline failure model for replicas, to which we refer as 
\emph{crash-restart-rollback} ($\CRR$), formalizes the types of
failures that can occur in a TEE-based system (see \S\ref{sec:intro}).
Under $\CRR$, a replica can experience a \emph{crash} that 
interrupts its execution either permanently or temporarily. 
In the former case, the replica stops taking steps; in the latter case, it 
\emph{restarts}, by executing a specified \onrestart\ block and then 
continues its execution from the state reached thereupon.
For simplicity, we assume that when a replica is restarted,
the restart and its immediately preceding crash occur \emph{simultaneously}
as a single atomic event.
We refer to the restarted replicas that have 
completed their \onrestart\ blocks as \emph{active}.
Upon restart, a replica with non-volatile storage may 
additionally experience a \emph{rollback} failure,
which causes the content of its non-volatile storage
to revert to a version older than it had before the latest crash.

A replica is \emph{up} in an interval 
$[t, t']$ if it does not crash without restart before $t$, and does not
crash (either with or without restart) at all times in $[t, t']$;
a replica is \emph{up after} $t$, if it is up in $[t, \infty)$; a replica $p_i$
is \emph{eventually up} if there exists $t$ such that
$p_i$ is up after $t$. Replicas in $\CRR$ fall into 
four disjoint classes:
\begin{itemize}
	\item a replica is \emph{perfect} if it never crashes;
	\item a replica is \emph{benign} if it crashes and restarts at least once, but is 
	eventually up, and furthermore, it never rolls back;
	\item a replica is \emph{crash-faulty} if it crashes without restarting or
	crashes and restarts infinitely often, and furthermore, it never rolls back; and
	\item a replica is \emph{rollback-faulty} if it rolls back at least once.
\end{itemize}
Perfect and benign replicas are together called \emph{correct}; crash-faulty and
rollback-faulty replicas are together called \emph{faulty}.

Note that $\CRR$ generalizes the classical crash-recovery failure 
model~\cite{hurfin1998consensus} by extending its set of faulty  behaviors with rollback
failures. It is also strictly weaker than the Byzantine failure model since
faulty replicas are required to follow their prescribed protocols. The
relationship of our model to other prior failure models is further discussed
in~\S\ref{sec:related}.

\subparagraph*{Channel reliability assumptions.}
We assume that every pair of processes
$p_i$ and $p_j$ are connected via point-to-point authenticated links such that
if both $p_i$ and $p_j$ are up after some $t$,
then every message sent by $p_i$ to $p_j$ after $t$ is eventually delivered 
by $p_j$. Note that this assumption is strictly weaker than the standard notion of a reliable channel. 
Specifically, it does not require reliability for messages sent by a correct process 
to another correct process before both processes stop crashing.

\subparagraph*{Specifications.}
We use standard safety and liveness notions to specify correctness conditions for 
register implementations: atomicity~\cite{lamport1986interprocess} 
(linearizability~\cite{DBLP:journals/toplas/HerlihyW90}), 
safeness~\cite{lamport1986interprocess}, wait-freedom~\cite{wait-free}, 
and obstruction-freedom~\cite{herlihy2003obstruction}. 
We consider both \emph{single-writer/single-reader (SWSR)} and 
\emph{multiple-writer/multiple-reader (MWMR)} register implementations.

\subparagraph*{Fault-tolerant implementations.}  An object implementation is
\emph{safe under a failure model $\mathcal{F}$} (or simply
$\mathcal{F}$\emph{-safe}) if it satisfies safety in all $\mathcal{F}$-valid
executions. We define $\mathcal{F}$\emph{-liveness} similarly. An implementation
is $\mathcal{F}$\emph{-tolerant} if it is both safe and live under
$\mathcal{F}$.

We study the implementability of registers under failure models
where, in every execution, arbitrarily many clients can crash and replicas
are subject to a restricted variant of the baseline $\CRR$ model above. 
Specifically, in~\S\ref{sec:static} we consider \emph{static} failure models 
$\FCRR{k,r,b}$ for replicas where the number of 
faulty, rollback-faulty, and benign replicas in every execution
is bounded by $k\le n$, $r\le k$, and $b\le n$, respectively.
Note that $\forall k, r, b.\, \FCRR{k,r,b} \preceq \CRR$.
In~\S\ref{sec:dynamic} we consider more flexible \emph{dynamic} failure models.



\section{Register Implementations in Static Failure Models}
\label{sec:static}

In this section we give a full characterization of the costs of implementing a
register under static failure models in terms of its failure resilience, time
complexity, and non-volatile storage requirements.  For resilience,
the following theorem establishes a tight bound on the total number of replicas
$n$ as a function of thresholds on the number of replicas of different types.
\begin{theorem}
  Assume that all replicas have access to non-volatile storage.  Then for all $k$, $r$, $b$, 
  there exists a $\FCRR{k,r,b}$-tolerant implementation of a wait-free atomic MWMR register if
  and only if $n \ge 2k + \min(b, r) + 1$.
\label{thm:main-tight}
\end{theorem}

When $r=0$, the bound in the theorem specializes to $n \ge 2k + 1$,
the same as for crash fault-tolerant register implementations. When $r \le b$,
the bound specializes to $n \ge 2k + r + 1$. This matches the known
generalization~\cite{gv06} of the $n \ge 3k+1$ Byzantine implementability bound
for registers~\cite{mad02,dgmsv11} to the case when at most $r$ out of $k$
faulty replicas can behave arbitrarily and the rest can crash but not deviate
from the prescribed protocol. On the other hand, if $b < r$, then the theorem
shows that we can implement a $\FCRR{k,r,b}$-tolerant wait-free atomic register 
with
$n = 2k + b + 1 < 2k + r + 1$ replicas -- strictly fewer than in the Byzantine
case.
We next refine Theorem~\ref{thm:main-tight} to additionally give a tight bound
on the number of replicas that must be equipped with non-volatile storage.
\begin{restatable}{theorem}{nonvolmain}
\label{thm:non-vol-main}
  Let $s$ be the number of replicas with non-volatile storage. Then for all 
  $k$, $r$, $b$, there exists a $\FCRR{k,r,b}$-tolerant implementation of a wait-free
  atomic MWMR register if and only if 
  $$
  (2k + r + 1 \le n < 2k + b + 1 \wedge s \ge 2k + r + 1) \vee
  (n \ge 2k + b + 1).
  $$
\end{restatable}
Note that the constraint on $n$ in Theorem~\ref{thm:non-vol-main} implies that
in Theorem~\ref{thm:main-tight}; we illustrate the constraint in
Figure~\ref{fig:upper-bound}. The above theorem shows that, when
$n \ge 2k + b + 1$, atomic registers can be implemented without any stable
storage at all. In practice, avoiding non-volatile storage offers substantial
performance gains when the register state fits within the TEE-protected RAM, as
it avoids the overheads of monitoring storage integrity at runtime.

\begin{figure}[t]
	\centering
	
	\scalebox{0.9}{
	\begin{tikzpicture}[yscale=1]
		
		\draw (0,0) rectangle (2,2);
		\draw (0,2) rectangle (2,4);
		\draw (2,0) rectangle (4,2);
		\draw (2,2) rectangle (4,4);

		\node[font=\scriptsize] at (-1, 1) {$n \ge 2k + r + 1$};
		\node[font=\scriptsize] at (-1, 3) {$n < 2k + r + 1$};
		\node[font=\scriptsize] at (1, 4.25) {$n < 2k + b + 1$};
		\node[font=\scriptsize] at (3, 4.25) {$n \ge 2k + b + 1$};
		
		\fill [color=brown] (2, 0) rectangle (4, 4);
		\fill [color=green] (0, 0) rectangle (2, 2);

		\node [font=\scriptsize, align=center] at (1, 3) {Atomic\\$s = 0$};
		\node [font=\scriptsize, align=center] at (1, 1) {Atomic\\Wait-free\\$s = 2k + r + 1$};
		\node [font=\scriptsize, align=center] at (3, 2) {Atomic\\Wait-free\\$s = 0$};
		
	\end{tikzpicture}
	}
	
	\caption{Resilience ($k$, $r$, $b$) and non-volatile storage usage ($s$) of the upper bound. Each square
	corresponds to the chunk of the problem space where the conjunction of the conditions on the axes holds.}
	\label{fig:upper-bound}
\end{figure}

Finally, the result in Theorem~\ref{thm:non-vol-main} can be strengthened along
several dimensions:
\begin{itemize}
\item The time complexity of the upper bound we present in \S\ref{sec:upper} is
  always $4$ message delays for both writes and reads. This matches the time
  complexity of crash fault-tolerant implementations, which is optimal. This
  property stands in a sharp contrast with Byzantine-resilient register
  constructions, where write and read latencies of $4$ and $8$ message delays,
  respectively, are inherent even for single-writer/multi-reader (SWMR)
  wait-free atomic registers~\cite{dgmsv11}.
\item The lower bound in Theorem~\ref{thm:non-vol-main} is established for
  obstruction-free safe implementations of SWSR registers.
\item When either of $b$ or $r$ is unknown, a tight bound is obtained from
  Theorem~\ref{thm:non-vol-main} by letting $b \triangleq n$ or
  $r \triangleq k$.
\item If the number of replicas is less than the resilience bound
  ($n < 2k + \min(b, r) + 1$), the upper bound does not sacrifice safety, but
  only liveness (the top-left box in Figure~\ref{fig:upper-bound}). In
  \S\ref{sec:dynamic} we leverage this property to develop a
  register construction under dynamic failure models that do not require
  all failure thresholds to hold for the entire execution. 
\end{itemize}

Our findings collectively demonstrate that, while existing Byzantine
fault-tolerant register constructions~\cite{mad02,ackm06,gv06,dgmsv11} can be
used for implementing $\FCRR{k,r,b}$-tolerant atomic registers, they are not well-suited
for this purpose due to their sub-optimal failure resilience and latency. We
next present an algorithm that gives a witness for our upper bound; we defer the
proofs of the other results to \tr{\ref{app:ktol}}{A}.


\subsection{Upper Bound}
\label{sec:upper}

An algorithm $\alg{}{}$ that shows our upper bound for $\FCRR{k,r,b}$ is presented in
Figure~\ref{fig:pseudo-a}. It is based on the MWMR variant of
ABD~\cite{DBLP:conf/podc/AttiyaBD90,rambo,grishas-abd}, where read and write
quorums are sets of replicas of size $q_r$ and $q_w$, respectively.
The algorithm uses the truth value of the predicate
$P(n, k, r, b) = (2k + r + 1 \le n < 2k + b + 1)$ to determine the
values of $q_r$ and $q_w$ (line~\ref{line:quorum}), the restart handler logic
(lines~\ref{line:onrestart}--\ref{alg:line:setstale-a2}), and whether
non-volatile storage should be used to store the replica's state
(line~\ref{line:non-vol}). In the following we refer to the variant of
$\alg{}{}$ executed if $P$ holds as algorithm $\alg{1}{}$, and as algorithm
$\alg{2}{}$ otherwise. The former corresponds to the bottom-left box in
Figure~\ref{fig:upper-bound}, while the latter corresponds to the remaining
portions of the figure.

\begin{figure}[t]
	\scalebox{0.96}{
\begin{minipage}[t]{\textwidth}
\begin{algorithm}[H]
  \DontPrintSemicolon
\nonl \textbf{Predicates and constants:}
\SubAlgoBlock
{
	 $P(n, k, r, b) \triangleq 2k + r + 1 \le n < 2k + b + 1$\;
	$(\wrqsize, \rdqsize) \triangleq (n-k, n-k)$ if $P(n, k, r, b)$; and $(n-k, k+1)$ otherwise 
	\label{line:quorum}
}
\end{algorithm}
\end{minipage}
}

\scalebox{0.96}{
\begin{algorithm}[H]
\DontPrintSemicolon
\nonl \textbf{State:}
\SubAlgoBlock
{
	$\State$, initially $((0, 0), v_0)$, with selectors $\Tsvar$ and $\Valvar$\; \label{alg:line:statedef}
	$\stale \in \{\TRUE, \FALSE\}$, initially $\FALSE$\;
	$\State$, $\stale$: \textbf{non-volatile} if $P(n, k, r, b) \wedge (i \le 2k + r + 1)$ 
	\label{line:non-vol}
}
\end{algorithm}
}

\scalebox{0.96}{
\begin{minipage}[t]{0.6\textwidth}
\begin{algorithm}[H]
  \DontPrintSemicolon
  \onrestart \label{line:onrestart}
\SubAlgoBlock
{
	\lIf{$\neg P(n, k, r, b)$}{$\stale \leftarrow \TRUE$\label{alg:line:setstale-a2}} 
}
\smallskip\smallskip
$\fn\ \writ(\tmpval)$ \label{alg:line:wrfn}
\SubAlgoBlock
{
	$S \leftarrow \rdint(\TS)$\; \label{alg:line:writereadtag}
	$\cnt \leftarrow \max\{\cnt'\ |\ (\cnt', \_) \in S \}$\; 
	\label{alg:line:writeselectmax}
	$\tmpts \leftarrow (\cnt + 1, i)$\; \label{alg:line:writeinctag}
	$\wrint(\tsval(\tmpts, \tmpval))$\; \label{alg:line:recwritewriteint}
	\KwRet \textsf{ack}\; \label{alg:line:wrfnret}
}
\smallskip\smallskip
$\fn\ \wrint(\req)$ \label{alg:line:wrint-a1}
\SubAlgoBlock
{
	$\idvar \leftarrow \newid()$\; \label{line:newid-w}
	\PUntil{\upshape \rc\ $\{\WRITEACK(\idvar, j)\ |\ p_j \in Q\}$ \qquad\qquad\qquad\qquad {\bf for some\
			$Q$ such that $|Q| \ge \wrqsize$} 
		\label{alg:line:rcvwriteack-a1}}
	{$\send\ \WRITE(\idvar, \req)\ \text{\bf to}\ \allreps$\; \label{alg:line:wbcast-a1}}
	$\return$\; \label{alg:line:endwrint-a1}
}
%
%
\smallskip\smallskip
$\whenrc\ \WRITE(\idvar, \tsval(\tmpts, \tmpval))$ \from{} $p_j$\label{alg:line:write-msg}
\SubAlgoBlock
{
	\If{$\tmpts > \State.\Tsvar$}
	{$(\State.\Tsvar, \State.\Valvar) \leftarrow (\tmpts, \tmpval)$}		\label{alg:line:tsvalset-a1}
	$\send\ \WRITEACK(\idvar, i)$ \text{\bf to} 
	$p_j$
		\label{alg:line:sendwrack-a1} 
}
\end{algorithm}
\end{minipage}
\hspace{-0.4cm}
\begin{minipage}[t]{0.49\textwidth}
\begin{algorithm}[H]
\DontPrintSemicolon
$\fn\ \rd()$ \label{alg:line:rdfn}
\SubAlgoBlock
{
	$S \leftarrow \rdint(\tsval)$\; \label{alg:line:recreadreadint}
	\textbf{let} $(\tmpts, \tmpval)$ \textbf{be such that} 
		$(\tmpts, \tmpval) \in S \wedge {}$ $\tmpts = \max\{\tmpts'\ |\ (\tmpts', 
		\_) \in S\}$\; \label{alg:read:selectmax}
	$\wrint(\tsval(\tmpts, \tmpval))$\; \label{alg:line:readwriteback}
	\KwRet $v$\; \label{alg:line:rdfnret}
}
\smallskip\smallskip
$\fn\ \rdint(\req)$ \label{alg:line:rdint-a1}
\SubAlgoBlock
{
	$\idvar \leftarrow \newid()$\; \label{line:newid-r}
	\PUntil{\upshape \rc\ $\{\READACK(\idvar, x_j, j)\ |\ p_j \in Q\}$ \qquad\qquad\qquad{\bf for
			some $Q$ such that $|Q| \ge \rdqsize$}
		\label{alg:line:rcvreadack-a1}}
	{$\send\ \READ(\idvar, \req)\ \text{\bf to}\ \allreps$\; \label{alg:line:rbcast-a1}}
	$\return\ \{x_j\ |\ p_j \in Q\}$ \label{alg:line:endread-a1}
}

\smallskip\smallskip

$\whenrc\ \READ(\idvar, \req)\ \from\ p_j$\label{alg:line:read-msg}
\SubAlgoBlock
{
	$\pre\ \stale = \FALSE$\\ \label{alg:line:notstale-a2}
	$\resp \leftarrow$\ \Case {$\req$} { \label{alg:line:rproc-a2}
		\Indp
		$\TS$: $\State.\Tsvar$\;
		$\tsval$: $(\State.\Tsvar, \State.\Valvar)$\;
	}
	$\send\ \READACK(\idvar, \resp, i)\ \text{\bf to}\ p_j$ \label{alg:line:readack-a2}
}
\end{algorithm}
\end{minipage}
}
\caption{Pseudocode for algorithm $\alg{}{}$ at process $p_i \in \allprocs$.}
\label{fig:pseudo-a}
\end{figure}


\subparagraph*{Replica states.}
Each replica stores a copy of the register state in 
a $\State$ variable (line~\ref{alg:line:statedef}), which 
is a tuple consisting of the register value $\State.\Valvar$ and 
timestamp $\State.\Tsvar$. Timestamps are pairs of a counter 
and the client identifier, ordered lexicographically:
$(\cnt_1, j_1) \le (\cnt_2, j_2)$ iff $\cnt_1 < \cnt_2 \vee (\cnt_1=\cnt_2 \wedge j_1 \le j_2)$.
Each replica also maintains a Boolean flag $\stale$: if this flag is
true, the replica is not allowed to respond to read requests.

\subparagraph*{Quorum-access functions.}
The clients use auxiliary functions $\rdint$ (lines~\ref{alg:line:wrint-a1}--\ref{alg:line:endwrint-a1})
and $\wrint$ (lines~\ref{alg:line:rdint-a1}--\ref{alg:line:endread-a1})
to respectively query and update various
components of replica states at a quorum. For $\rdint$, the $\req$ argument
specifies the state elements to query: e.g., $\tsval$ corresponds to
$(\State.\Tsvar, \State.\Valvar)$.  For $\wrint$, $\req$ specifies which state
elements to modify and their new values: e.g., $\tsval(\tmpts, \tmpval)$ to set
$(\State.\Tsvar, \State.\Valvar)$ to $(\tmpts, \tmpval)$. 

The $\rdint$ (respectively, $\wrint$) function starts by generating
a globally unique request identifier $\idvar$ (lines~\ref{line:newid-w} 
and~\ref{line:newid-r}); in practice this can be implemented using
cryptographic nonces. 
The function then broadcasts a $\READ(\idvar, \req)$ (respectively, $\WRITE(\idvar, \req)$)
message to all replicas, and awaits $\READACK$ (respectively, $\WRITEACK$) 
messages from $q_r$ (respectively, $q_w$) replicas.
Since channels may fail 
to deliver messages sent before their respective destinations stop crashing (\S\ref{sec:model}),
to ensure liveness, both $\rdint$ and $\wrint$ continue retransmitting requests until
receiving the desired quorums of responses.
The $\rdint$ function then returns the set of payloads
received in $\READACK$ messages from a read quorum of replicas.

Whenever a replica with $\stale = \FALSE$ receives a $\READ(\idvar, \req)$
message (line~\ref{alg:line:read-msg}), it responds with a $\READACK$ message,
carrying $\idvar$ and the current value of the state variable associated with
$\req$. Whenever a replica receives a $\WRITE(\idvar, \req)$ message
(line~\ref{alg:line:write-msg}), it applies the update encoded in $\req$ to its
state and acknowledges this fact with a $\WRITEACK$ message. In particular, if
$\req=\tsval(\tmpts, \tmpval)$ and $\State.\Tsvar < \tmpts$, then
$(\State.\Tsvar, \State.\Valvar)$ is set to $(\tmpts, \tmpval)$; otherwise, the
state is left unchanged.

\subparagraph*{Read and write protocols.}
To write a value $v$, a process $p_i$ first calls $\rdint$ 
to retrieve a set of timestamps from a read quorum
(line~\ref{alg:line:writereadtag}). It then
selects the highest timestamp counter $\cnt$ and calls $\wrint$ to 
store $\tmpval$ with timestamp $(\cnt+1, i)$ at a write quorum.
To read a value, a process $p_i$ first 
calls $\rdint$ to retrieve a set of 
timestamp-value pairs from a read quorum (line~\ref{alg:line:recreadreadint}).
It then selects the pair $(\tmpts, \tmpval)$ with the highest 
timestamp, calls $\wrint$ to store it at a write quorum, and 
returns $\tmpval$.

\subparagraph*{Algorithm $\alg{1}{}$.}
Algorithm $\alg{1}{}$ handles the case when the predicate $P(n, k, r, b)$ holds,
i.e., $2k + r + 1 \le n < 2k + b + 1$ (the bottom-left box in
Figure~\ref{fig:upper-bound} and the first disjunct in
Theorem~\ref{thm:non-vol-main}). The algorithm uses read and write quorums of
size $q_w=q_r=n-k$ (line~\ref{line:quorum}) and keeps the state of $2k+r+1$
replicas in non-volatile storage (line~\ref{line:non-vol}), which is the minimum
required by the lower bound of Theorem~\ref{thm:non-vol-main}. Its restart
handler returns immediately
(lines~\ref{line:onrestart}--\ref{alg:line:setstale-a2}), so the $\stale$
flag is always $\FALSE$ at all replicas, which can thus respond to $\READ$
messages.

Since in every execution at least $n - k > 0$ replicas are
eventually up, both write and read quorums are always available, and, 
as a result, algorithm $\alg{1}{}$ is wait-free.  We defer the full proof of its
linearizability to \tr{\ref{app:safety-high}}{A} and only prove the following 
key property the proof relies on:
\begin{restatable}[Real-Time Order]{property}{rtp}
	\label{lem:rtp}
        Let $\wrq$ be a completed call to $\wrint(\tsval(x, \_))$, and $\rdq$ be
        a completed call to $\rdint(\TS)$ or $\rdint(\tsval)$. If $\rdq$ is
        invoked after the completion of $\wrq$, then it returns a set containing
        a value $\ge x$ for $\State.\Tsvar$.
\end{restatable}
\begin{proof}[Proof of Property~\ref{lem:rtp} for $\alg{1}{}$.]
  Since $\wrq$ returns, it has executed the $\tsval(x, \_)$ request at $q_w=n-k$
  replicas. Since $\rdq$ returns, it has retrieved the value of $\State.\Tsvar$
  from a $q_r=k+1$ replicas. The intersection of the two quorums has a
  cardinality $\ge (n - k) + (n - k) - n = n - 2k$. Furthermore, the number of
  replicas in this intersection that store their states on non-volatile storage
  is $\ge (n - 2k) + (2k + r + 1) - n = r + 1$. Since at most $r$ replicas
  experience rollback failures, there exists at least one replica $p$ in the
  intersection of the quorums used by $\wrq$ and $\rdq$ that stores its state on
  non-volatile storage and never rolls back. Our protocol only allows replicas
  to increase the value of $\State.\Tsvar$ during normal execution
  (line~\ref{alg:line:tsvalset-a1}), so replica $p$ must have
  $\State.\Tsvar \ge x$ after it responds to $\wrq$. Since $\rdq$ starts after
  $\wrq$ completes, $p$ must respond to $\rdq$ after $\wrq$ completes. Thus, $p$
  must have $\State.\Tsvar \ge x$ when it responds to $\rdq$, so $\rdq$ receives
  a response containing $\State.\Tsvar \ge x$.
\end{proof}

\subparagraph*{Algorithm $\alg{2}{}$.}
Algorithm $\alg{2}{}$ handles the case when the predicate $P(n, k, r, b)$ does
not hold (the right and top-left boxes in Figure~\ref{fig:upper-bound} and the
second disjunct in Theorem~\ref{thm:non-vol-main}).
The algorithm does not use non-volatile storage, and has write quorums of size
$q_w=n-k$ and read quorums of size $q_r=k+1$. Additionally, only replicas that
have not crashed before being queried can participate in read quorums.  To
ensure this, each replica sets its $\stale$ flag to $\TRUE$ before returning
from the recovery procedure (line~\ref{alg:line:setstale-a2}). Recall that a
replica contacted by $\rdint$ checks this flag and responds only if it is
$\FALSE$ (line~\ref{alg:line:notstale-a2}).

When $n \ge 2k + b + 1$, in every execution at least $n - k > 0$ replicas are
eventually up, so some write quorum is eventually available. Also, at 
least $n-(k+b) \ge k+1$ replicas never crash, so some read quorum is always
available. Hence, in this case $\alg{2}{}$ is wait-free. Furthermore,
$\alg{2}{}$ is always safe, even when $n < 2k + r + 1$. Similarly to
$\alg{1}{}$, the safety of $\alg{2}{}$ this follows from the Real-Time Order
Property, proved below; the rest of the proof is given in
\tr{\ref{app:safety-high}}{A}.
\begin{proof}[Proof of Property~\ref{lem:rtp} for $\alg{2}{}$.]
  The intersection of any read and write quorums has a cardinality
  $\ge (n-k)+(k+1)-n \ge 1$.  Hence there exists a replica $p$ that both
  executes $\tsval(x, \_)$ and responds to $\rdq$. Since $p$ is a member of a
  read quorum, it could not have crashed. Since the value of $\State.\Tsvar$ is
  non-decreasing in the absence of crashes (line~\ref{alg:line:tsvalset-a1}),
  $p$ must have $\State.\Tsvar \ge x$ after it responds to $\wrq$. Since $\rdq$
  starts after $\wrq$ completes, $p$ must respond to $\rdq$ after $\wrq$
  completes. Thus, $p$ must have $\State.\Tsvar \ge x$ when it responds to
  $\rdq$, so $\rdq$ receives a response containing $\State.\Tsvar \ge x$.
\end{proof}


\section{Register Implementations in Dynamic Failure Models}
\label{sec:dynamic}

In this section we study the implementability of registers under failure models
where any number of replicas can experience $\CRR$ failures 
throughout an execution. By Theorem~\ref{thm:main-tight}, no register
implementation can be simultaneously safe and live in all executions
if $n < 2k+\min(b,r)+1$. To circumvent this impossibility, we introduce
a family of \emph{dynamic} failure models $\Fdsym$, defined as follows.



We assume a constant $\Delta$ such that, in any execution, every message sent by
a process $p_i$ to a process $p_j$ at time $t$ is guaranteed to be received by
$p_j$ by $t+\Delta$, provided $p_i$ and $p_j$ are up in $[t, t+\Delta]$. Note that this
assumption does not make the model stronger than asynchronous: since we do not
assume a lower bound on message delays or processing times, processes do not
have a means to measure time passage and thus take advantage of the existence
of $\Delta$~\cite{hagit-churn}.

\begin{definition}
\label{ass:what}
  Given $\ndown$, $\ncrash$, and $M$, an execution $\alpha$ is
  $\Fd{\ndown,\ncrash,M}$\textbf{-valid} if
  there exists $\ups \subseteq \allreps$ such that $|\ups| \ge n - \ndown$, all replicas in $\ups$
  are eventually up in $\alpha$ and either:
\begin{enumerate}
\item \label{ass:dyn}
  $c > d$ and for any two crash events occurring at times $t$ and
  $t'\ge t$, we have $t'-t > M\Delta$; or
\item \label{ass:stat} no replica crashes in some set $S \subseteq U$ such that
  $|S| \ge n - \ncrash$.
\end{enumerate}
\end{definition}

Intuitively, the parameter $\ndown$ captures the upper bound on the number of
replicas that are not eventually up in $\alpha$, and must be known to any
register implementation. Conditions~\ref{ass:dyn} and~\ref{ass:stat} further
restrict failure scenarios: either failures must be separated by at least $M$
message delays, or at most $\ncrash \ge \ndown$ replicas can crash
in $\alpha$. The former is similar to the churn-limiting assumptions used to
model process participation in dynamic and reconfigurable
systems~\cite{hagit-churn}. The parameter $M$ is implementation-specific and
captures the minimum time required for a replica to recover its state before the
next crash occurs.

We show that every static failure model $\FCRR{k,r,b}$ from~\S\ref{sec:static}
is a special case of a dynamic model $\Fd{k, k+b, \_}$. Thus, any 
$\Fd{k, k+b, \_}$-tolerant register implementation is also $\FCRR{k,r,b}$-tolerant.
\begin{proposition}
	$\forall k, r, b, M.\, \FCRR{k,r,b} \preceq \Fd{k, k+b, M}$.
	\label{prop:dyn-stat}
\end{proposition}
\begin{proof}
  Consider any $\FCRR{k,r,b}$-valid execution $\alpha$. Then in this execution
  up to $k$ replicas are faulty and up to $b$ replicas are benign. Let $U$ be
  the set of correct replicas, and $S \subseteq U$ be the set of perfect
  replicas. Then $|U| \ge n-k$ and $|S| \ge n-k-b$. Furthermore, the replicas in
  $U$ are eventually up and those in $S$ never crash, so the condition in case 2
  of Definition~\ref{ass:what} is satisfied. Thus, $\alpha$ is
  $\Fd{k, k+b, M}$-valid, as needed.
\end{proof}

We prove the following upper bound in the dynamic failure model:
\begin{theorem}
For all $\ndown$, $\ncrash$, if $n \ge \ndown+\ncrash+1$ and $M \ge 12$, then there exists an 
implementation of an atomic wait-free MWMR register that is
$\Fd{\ndown,\ncrash,M}$-live and always safe.
\label{thm:upper-dynamic}
\end{theorem}

For the lower bound, we prove a stronger result that holds for
$\Fd{\ndown, \ncrash, M}$-tolerant implementations, and not just those that are
$\Fd{\ndown, \ncrash, M}$-live and always safe:
\begin{theorem}
	For all $\ndown$, $\ncrash$, $M$, if there exists a $\Fd{\ndown, \ncrash, M}$-tolerant 
	implementation of a
	safe obstruction-free SWSR register, then $n \ge \ndown+\ncrash+1$.
	\label{thm:lower-for-dynamic}
\end{theorem}

The two theorems imply a tight resilience bound:
\begin{theorem}
  For all $\ndown$, $\ncrash$, 
  $\Fd{\ndown, \ncrash, M}$-live implementation of an atomic wait-free MWMR
  register exists for some $M$ if and only if $n \ge \ndown+\ncrash+1$.
\label{thm:tight-dynamic}
\end{theorem}

The lower bound in Theorem~\ref{thm:lower-for-dynamic} reveals an inherent
trade-off: while by Proposition~\ref{prop:dyn-stat} any
$\Fd{k, k+b, \_}$-tolerant register implementation is also
$\FCRR{k,r,b}$-tolerant, any such implementation requires $n \ge 2k+b+1$,
thus sacrificing the optimal resilience under $\FCRR{k,r,b}$
(cf. Theorem~\ref{thm:non-vol-main}).


In the rest of this section we present an algorithm that validates
Theorem~\ref{thm:upper-dynamic}, which we call $\algname$; we defer the proofs
of the other results to \tr{\ref{app:dynamic}}{B}.  The algorithm's latency in
crash-free executions is $4$ message delays for both reads and writes, which is
optimal. Since under $\Fd{\ndown, \ncrash, M}$, any number of replicas can
suffer rollbacks in some executions, the algorithm does not rely on
non-volatile storage.  It also does not require the knowledge of $\ncrash$.



\subsection{Crash-Consistency Basics}
\label{sec:algo2plus}


\begin{figure}[!ht]
  \scalebox{0.96}{%
\!\begin{algorithm}[H]
  \DontPrintSemicolon
              \setcounter{AlgoLine}{0}
		\nonl \textbf{Constants:}
		\SubAlgoBlock
		{
			$(\wrqsize, \rdqsize) \triangleq (n - \ndown, \ndown + 1)$
		}
	\end{algorithm}
}
\scalebox{0.96}{%
\!\begin{algorithm}[H]
  \DontPrintSemicolon
%
\nonl \textbf{State:}
\SubAlgoBlock
{
	$\State$, initially $((0, 0), v_0, [0..0]\newcode{, [0..0]})$, with 
	selectors $\Tsvar$, $\Valvar$, $\cvvar$\newcode{, and 
	$\precvvar$}\;
	$\stale \in \{\TRUE, \FALSE\}$, initially $\FALSE$\;
}
\end{algorithm}
}

\scalebox{0.96}{
\!\begin{minipage}[t]{0.563\textwidth}
	\begin{algorithm}[H]
          \DontPrintSemicolon
		\onrestart \label{alg:line:dyn-recov}
		\SubAlgoBlock
		{
			$\stale \leftarrow \TRUE$\; \label{alg:line:setstale-a2+}
			\newcode{
				$S \leftarrow \rdint(\precv(i))$\; 
				\label{alg:line:recoverymyincread}
				$\nextinc \leftarrow \max(S) + 1$\; \label{alg:line:choosenewinc}
				$\wrint(\precv(i, \nextinc))$\; 
				\label{alg:line:recoverymyincwrite}
			}
			$\State.\cvvar[i] \leftarrow \nextinc$\; 
			\label{alg:line:recoverymyincincrement}
			$\wrint(\cv(i, \State.\cvvar[i]))$\;\label{alg:line:cvwrite}
			$S\leftarrow \rdint(\ST(\State.\cvvar[i]))$\; \label{alg:line:rdbegin}
		
 		\ForAll {$z = 1..n$\label{alg:line:for-k-max}} {
	\newcode{$\State.\precvvar[z] \leftarrow \max 
		\{\State.\precvvar[z],$ $\max_{s \in S}\{s.\precvvar[z]\}\}$\;}
	\label{alg:line:updatepre}
	$\State.\cvvar[z] \leftarrow \max \{\State.\cvvar[z],$ $\max_{s \in 
		S}\{s.\cvvar[z]\}\}$\;
	\label{alg:line:updatecv}
}
\textbf{let} $(\tmpts, \tmpval)$ \textbf{be such that} 
$\exists s\in S.$ $(\tmpts, \tmpval) = (s.\Tsvar, s.\Valvar) \wedge 
\tmpts = \max\{s.\Tsvar \mid s \,{\in}\, S\}$\; 
\label{alg:line:choosets}
\If {$\tmpts > \State.\Tsvar$} {($\State.\Tsvar, \State.\Valvar) \leftarrow 
	(\tmpts, \tmpval)$} \label{alg:line:rdend} 
\label{alg:line:updatets}
$\stale \leftarrow \FALSE$\; \label{alg:line:clearstale} 
}

\smallskip\smallskip

$\fn\ \wrint(\req)$ \label{alg:line:wrint-a2+}
\SubAlgoBlock
{
	$(\idvar, Q, cv) \leftarrow (\newid(), \emptyset, [0..0])$\; \label{alg:line:x-init}
	\DoUntil {$|Q| \ge \wrqsize$} 
	{
		\label{alg:line:while}
		\PUntil{\upshape \rc\\
			\nonl \ \ $\{\WRITEACK(\idvar, \incvar_j, \tmpcv_j, j)\ |\ 
			p_j \,{\in}\, \newresps\}$ \textbf{for\ some} $\newresps$			
			\textbf{such that} $|\newresps \,{\cup}\, Q| \ge \wrqsize$ \label{alg:line:wrint-wait}}
		{	\label{alg:line:while2}
			\ForEach{$p_z \in \allreps \setminus Q$ \label{alg:line:sendwritefor}}{
			$\send\ \WRITE(\idvar, \req\newcode{, \tmpcv[z]})\ \textbf{to}\ p_z$}
			\label{alg:line:firstbcast} \label{alg:line:sendwrite}}
		\uIf{$\req \in \{\tsval(\_, \_)\}$\label{alg:line:wrif}}
		{
			$S \leftarrow \{\tmpcv'\ |\ p_i \text{~received}$\\ 
			\nonl \  \ \ \ $\WRITEACK(\idvar, \_, \tmpcv', j)\}$\; 
			\label{alg:line:rdcvts}
		}
		\Else (\tcp*[h]{$\req \in \{\cv(\_, \_)\newcode{, \precv(\_, \_)}\}$} 
		\label{alg:line:wrelse})
		{
			$S \leftarrow \rdint(\cv)$\; \label{alg:line:rdcv1}
		}
		\ForAll {$z = 1..n$\label{alg:line:cvfor}}
			{$\tmpcv[z] \leftarrow \max\{\tmpcv'[z]\ |\ \tmpcv' \in S\}$} 
			\label{alg:line:rdcv} 
		$Q \leftarrow \{p_j\; |\; i = j\; \vee$\\ 
				\nonl \ \ \  $(p_i \text{~received~}
				\WRITEACK(\idvar, \incvar_j, \_, j)$\\ 
				\nonl \ \ \ \ \ s.t. $\incvar_j \ge \tmpcv[j])\}$\;
		\label{alg:line:filter-x}
		\label{alg:line:wr-resend}
		\label{alg:line:resendwritecrashed} 
	} \label{alg:line:wrint-ret-a2plus}
}
\end{algorithm}
 \end{minipage}
 \hspace{-0.4cm}
 \begin{minipage}[t]{0.52\textwidth}
 		\begin{algorithm}[H]
 		\DontPrintSemicolon

 		$\fn\ \rdint(\req)$ \label{alg:line:rdintdef}
 		\SubAlgoBlock
 		{
 			$\idvar \leftarrow \newid()$\; \label{alg:line:id-gen-rdint}
 			\PUntil{\upshape \rc\ $\{\READACK(\idvar, x_j, j)\ |\ p_j \in Q\}$
 				\qquad\qquad\quad
 				\textbf{for\ some}\ $Q$ \textbf{such that} 
 				$|Q| \ge \rdqsize$ 
 				\label{alg:line:rcvreadack-a2+}}
 			{\label{alg:line:rdintloop}$\send\ \READ(\idvar, \req)\ \text{\bf to}\ \allreps$\; 
 			\label{alg:line:rbcast-a2+}}
 			$\return\ \{x_j\ |\ p_j \in Q\}$ \label{alg:line:endread-a2+}
 			
 		}
 		
 \smallskip\smallskip	

$\whenrc\ \READ(\idvar, \req)\ \from\ p_j$
\SubAlgoBlock
{
	$\pre\ \stale = \FALSE$\\ \label{alg:line:notstale-a2+}
	\If{$\req = \ST(\incvar)$}{$\State.\cvvar[j] \leftarrow 
	\max\{\State.\cvvar[j], \incvar\}$} \label{alg:line:rdstinc}
	$\resp \leftarrow$\ \Case {$\req$} { \label{alg:line:rproc-a2+}
		\Indp
		$\TS$: $\State.\Tsvar$\;
		$\tsval$: $(\State.\Tsvar, \State.\Valvar)$\;
		\newcode{$\precv(j)$: $\State.\precvvar[j]$\;}
		$\cv$: $\State.\cvvar$\;
		$\ST(\incvar)$:  $\State$\;
	}
	$\send\ \READACK(\idvar, \resp, i)$ \textbf{to} $p_j$ 
	\label{alg:line:readack-a2+}
}

\smallskip\smallskip
$\whenrc\ \WRITE(\idvar, \tsval(\tmpts, \tmpval)\newcode{, \_})$ \from{} $p_j$
\SubAlgoBlock
{
	$\pre\ \stale = \FALSE$\\ \label{alg:line:notstale-write}
	\If{$\tmpts > \State.\Tsvar$}
	 		{$(\State.\Tsvar, \State.\Valvar) \leftarrow (\tmpts, \tmpval)$} \label{alg:line:tsvalset-a2+}
	$\send\ \WRITEACK(\idvar, \State.\cvvar[i], \State.\cvvar, i)$ \text{\bf to} $p_j$\label{alg:line:sendwrack-tsval}
}

\smallskip\smallskip

$\whenrc\ \WRITE(\idvar, \req\newcode{, \incvar})$ \from{} $p_j$\hspace{-2pt}
\label{line:write-2-handler}
\SubAlgoBlock
{
	$\pre\ \req \in \{\cv(j, v)\newcode{, \precv(j, v)}\}$\\ 
	\label{line:write-2}
	\newcode{$\State.\cvvar[i] \leftarrow 
		\max\{\State.\cvvar[i], \incvar\}$\; \label{alg:line:incmyinc}
	\uIf{$\req = \precv(j, v)$}{
		$\State.\precvvar[j] \leftarrow \max\{\State.\precvvar[j], v\}$ \label{alg:line:precvset}
	}}
	\Else{
		$\State.\cvvar[j] \leftarrow \max\{\State.\cvvar[j], v\}$ \label{alg:line:cvset}
	}
	$\send\ \WRITEACK(\idvar, \State.\cvvar[i], \bot, i)\ \textbf{to}\ 
	p_j$\label{alg:line:sendwrack-a2+}
}
\end{algorithm}
\end{minipage}
}
 
\caption{Pseudocode of $\algname$ at process $p_i \in \allprocs$.}
\label{fig:pseudo-dyn}
\end{figure}


The pseudocode of $\algname$ appears in Figure~\ref{fig:pseudo-dyn}. It
reuses the $\rd$ and $\writ$ procedures of algorithm $\alg{}{}$
(Figure~\ref{fig:pseudo-a}, lines~\ref{alg:line:wrfn}--\ref{alg:line:rdfnret}).
The algorithm relies on read quorums of size
$\rdqsize = \ndown + 1$, and so-called \emph{crash-consistent} write quorums of
size $\wrqsize = n-\ndown$, introduced in the following. Given the mapping
between static and dynamic models established by
Proposition~\ref{prop:dyn-stat}, these quorum sizes mirror those in $\alg{2}{}$:
$k+1$ for read quorums and $n-k$ for write quorums.

Since replicas lose their memory contents during a restart, $\algname$ follows the
approach of $\alg{2}{}$: a replica relies on its $\stale$ flag to determine
whether its current state can be used to respond to $\READ$
requests. However, in contrast to $\alg{2}{}$, simply
setting $\stale=\TRUE$ upon restart so that any replica that has not crashed permanently is
active will not be live. To see why, fix arbitrary $M$ and
$\Delta$, and consider an execution $\alpha$ where 
all replicas are eventually up, every replica crashes and 
restarts at least once, and every two consecutive crashes are separated
by $M\Delta$. Clearly, $\alpha$ is $\Fd{0,n,M}$-valid. However, 
$\alg{2}{}$ will not be live in $\alpha$ as it eventually runs out 
of non-stale replicas to form read quorums.
To deal with such
scenarios, a key ingredient of $\algname$ is a novel recovery protocol executed by a
replica upon restart
(lines~\ref{alg:line:dyn-recov}--\ref{alg:line:clearstale}). This protocol
reconstructs the state of a restarted replica, and thus enables it to clear its
$\stale$ flag.

The protocol achieves this by synchronizing with other
replicas. However, doing this naively by simply querying a read quorum and
adopting the state of the most up-to-date replica would lead to a safety
violation. For example, consider the algorithm $\alg{2}{}$ instantiated with 
$\rdqsize = \ndown + 1$ and $\wrqsize = n-\ndown$. Suppose we modify the
$\onrestart$ procedure at line~\ref{line:onrestart} as follows.  After setting
$\stale=\TRUE$, a restarting replica would invoke $\rdint(\tsval)$ to retrieve
the register states from a read quorum of replicas, assign the one with the
highest timestamp to its own copy of $\State$, set $\stale=\FALSE$, and then
return. The following scenario demonstrates that this modification results in a
linearizability violation.
\begin{example}\label{examples:crash-consistency}
Let $n = 3$, $\ndown = 1$ and $\ncrash = 1$, so that $n \ge \ndown + \ncrash
+ 1$. Consider the execution of the modified algorithm in
Figure~\ref{fig:problem}, which is valid under $\Fd{1, 1, \_}$.
First, a client $p_4$ invokes $\writ(v)$
(line~\ref{alg:line:wrfn} in Figure~\ref{fig:pseudo-a}), which calls $\rdint(\TS)$
to query timestamps at a read quorum and then $\wrint(\tsval((1, 4), v))$
to store the new register state at a write quorum.
In the latter call, $p_4$ sends $m=\WRITE(\tsval((1, 4), v))$ to all replicas and waits 
until it gathers $\WRITEACK$ responses from a write quorum.
Replica $p_3$ gets $m$ first, sets
$\State = ((1, 4), v)$, and responds with a $\WRITEACK$. 
It then crashes, loses its state, and reconstructs the state by contacting 
a read quorum $\{p_1, p_2\}$ of size $q_r=\ndown+1=2$.
Since $m$ has not yet reached either replica, 
$p_3$ sets $\State$ back to its initial value $((0, 0), v_0)$.
Replica $p_2$ then receives $m$, sets $\State = ((1, 4), v)$, and replies to $p_4$
with a $\WRITEACK$. Process $p_4$ now gathers $\WRITEACK$ responses from $q_w = n-\ndown =
2$ processes, which constitutes a write quorum.
Thus, $\wrint$ and $\writ(v)$ at $p_4$ terminate. Later, 
a client $p_5$ calls $\rd$, which obtains responses from a read quorum $\{p_1, p_3\}$.
Since neither replica is aware of $p_4$'s write, 
the $\rd$ returns $v_0$, violating linearizability. 
\label{ex:nocc}
\end{example}

\begin{figure}[t]
	\centering
	\begin{tikzpicture}[scale=0.87]
		
		\tikzmath{\pclient = 1.6; \pread = 0.8; \pwriter = -1.6; \pcrash = -0.8; \prest = 0;}

		\node [color=cyan!50] at (4.5, \prest) {\pgfuseplotmark{*}};
		\node [color=cyan!50] at (4.5, \pcrash) {\pgfuseplotmark{*}};
		\node [font=\scriptsize] at (4.5, \pcrash - 0.25) {$e'$};
		\node [color=orange!50] at (2.5, \pwriter) {\pgfuseplotmark{*}};
		\node [font=\scriptsize] at (2.5, \pwriter - 0.25) {$e$};
		\node [color=orange!50] at (8.5, \pcrash) {\pgfuseplotmark{*}};
		\node [font=\scriptsize] at (8.6, \pcrash - 0.25) {$e''$};

		\node [font=\scriptsize] at (3.9, \prest / 2 + \pcrash / 2) {$\READ$};
		\node [font=\scriptsize] at (7.55, \pwriter / 2 + \pcrash / 2) {$\READACK$};
		
		\node [font=\small] at (-0.5, \pwriter) {$p_3$};
		\node [font=\small] at (-0.5, \pcrash) {$p_2$};
		\node [font=\small] at (-0.5, \prest) {$p_1$};
		
		\node [font=\scriptsize] at (0.5, \pwriter + 0.25) {$((0, 0), v_0)$};
		\node [font=\scriptsize] at (0.5, \pcrash + 0.25) {$((0, 0), v_0)$};
		\node [font=\scriptsize] at (0.5, \prest + 0.25) {$((0, 0), v_0)$};
		
		\draw [dotted] (0,\pwriter) -- (15,\pwriter);
		\draw [dotted] (0,\pcrash) -- (15,\pcrash);
		\draw [dotted] (0,\prest) -- (15,\prest);
		
		\node [font=\small] at (1.5, \pclient) {$p_4$:};
		\draw [|-|] (2, \pclient) -- (10, \pclient);
		\node [font=\footnotesize] at (6, \pclient + 0.25) {$\wrint(\tsval((1, 4), 
			v))$};
		
		\node [font=\small] at (11.25, \pread) {$p_5$:};
		\draw [|-|] (11.75, \pread) -- (15, \pread);
		\node [font=\footnotesize] at (13.4, \pread + 0.25) {$\rdint(\tsval)$};
		
		\draw [->] (2.25, \pclient) -- (2.5, \pwriter);
		\draw [->] (2.5, \pwriter) -- (2.75, \pclient);
		\node [font=\scriptsize] at (1.75, \pclient / 2 + \prest / 2) 
		{$\WRITE$};
		\node [font=\scriptsize] at (3.5, \pclient / 2 + \prest / 2 - 0.25) 
		{$\WRITEACK$};
		
		\node at (3.25, \pwriter) {$\times$};
		\node [font=\scriptsize] at (3.25, \pwriter - 0.25) {Crash};
		
		\draw [-|] (3.25, \pwriter) -- (7.5, \pwriter);
		\node [font=\footnotesize] at (5.2, \pwriter - 0.25) {$\onrestart$};
		\node [font=\scriptsize] at (7.5, \pwriter - 0.25) {$\stale = \FALSE$};
		\node [font=\scriptsize] at (7.5, \pwriter - 0.6) {$((0, 0), v_0)$};
		
		\draw [->] (3.75, \pwriter) -- (4.5, \pcrash);
		\draw [->] (4.5, \pcrash) -- (7.5, \pwriter);
		\draw [->] (3.75, \pwriter) -- (4.5, \prest);
		\draw [->] (4.5, \prest) -- (7.5, \pwriter);
		
		\draw [->] (2.25, \pclient) -- (8.5, \pcrash);
		\draw [->] (8.5, \pcrash) -- (9.5, \pclient);
		\draw [dashed,->] (2.25, \pclient) -- (8.5, \prest+0.5);
		\node [font=\scriptsize] at (9.4, \pcrash + 0.25) {$((1, 4), v)$};
		\node [font=\scriptsize] at (6.2, \pclient / 2 + \prest / 2 - 0.4) 
		{$\WRITE$};
		\node [font=\scriptsize] at (7.5, \pclient / 2 + \prest / 2 + 
		0.1) {$\WRITE$};
		\node [font=\scriptsize] at (9.75, \pclient / 2 + \pcrash / 2) 
		{$\WRITEACK$};
		
		\draw [->] (11.75, \pread) -- (13.4, \prest);
		\draw [->] (13.4, \prest) -- (15, \pread);
		\draw [->] (11.75, \pread) -- (13.4, \pwriter);
		\draw [->] (13.4, \pwriter) -- (15, \pread);
		
		\node [font=\scriptsize] at (12.2, \pread / 2 + \pwriter / 2) {$\READ$};
		\node [font=\scriptsize] at (14.4, \pwriter / 2 + \pcrash / 2) {$\READACK$};
		
	\end{tikzpicture}
	
	\caption{An execution with incorrect recovery.}
	\label{fig:problem}
\end{figure}


To address the above problem, the set of responses collected by $p_4$ from a write quorum 
of replicas must be \emph{crash-consistent}~\cite{michael2017recovering}.
Let $\rightarrow$ be a partial order on the set of replica events
such that $e \rightarrow e'$ if there exists a replica $p_i$ such that $e$
precedes $e'$ at $p_i$.
\begin{definition}\label{defn:crash-consistent}
  A set $E$ of $[\send\ \WRITEACK]$ events is \textbf{crash-consistent} if for
  for each event $e \in E$ occurring at replica $p_i$, if after this event $p_i$
  crashes, restarts, and invokes $\rdint(\ST)$, then the set $E'$ of
  $[\send\ \READACK]$ events triggered by the read satisfies
  $\neg \exists e' \in E'.\ \exists e'' \in E.\ e' \rightarrow e''$.
\end{definition}

For example, consider Figure~\ref{fig:problem}. Let $E$ be the set of
$[\send\ \WRITEACK]$ events triggered by the call to
$wq=\wrint(\tsval((1, 4), v))$ at $p_4$ (orange dots), and let $E'$ be the set
of $[\send\ \READACK]$ events triggered by the call to $\rdint(\ST)$ during the
recovery at $p_3$ (blue dots). Let
$e \in E$ be the event at $p_3$ sending $\WRITEACK$, 
$e' \in E'$ the event at $p_2$ sending $\READACK$,
and $e'' \in E$ the event at $p_2$ sending $\WRITEACK$.
Since $e' \rightarrow e''$, the set $E$ is not crash-consistent.

\subsection{Crash-Consistency with Incarnation Numbers and Crash Vectors}
\label{sec:ccq-cv}

To implement crash-consistency, similarly to~\cite{michael2017recovering,replacement}, each
replica in $\algname$ maintains two pieces of metadata -- an \emph{incarnation number}
and a \emph{crash vector}.  The former is an integer, initialized to $0$, and
the latter tracks the highest known incarnation numbers of all other
replicas. The crash vector of replica $p_i$ is stored in $\State.\cvvar$, and
for convenience, $p_i$'s own incarnation number is stored in
$\State.\cvvar[i]$. This incarnation number is updated by the recovery procedure
at line~\ref{alg:line:recoverymyincincrement}.
To explain the basic principles underlying crash-consistency checks
in $\algname$, we first assume that incarnation numbers written by each replica $p_i$
into $\State.\cvvar[i]$ monotonically increase:
\begin{restatable}[Incarnation Number Monotonicity]{propertyr}{inm}
\label{lem:monotone}
Assume that a replica $p_i$ sets $\State.\cvvar[i]$ to $\incnum$ at
line~\ref{alg:line:recoverymyincincrement} at time $t$. Then $\incnum > 0$ and
for all times $t' > t$, if $p_i$ sets $\State.\cvvar[i]$ to $\incnum'$ at
line~\ref{alg:line:recoverymyincincrement} at time $t'$, then
$\incnum' > \incnum$.
\end{restatable}
This property can be easily ensured assuming that each replica has a built-in
monotonic counter that is {\em never} rolled back, as done in previous
work~\cite{michael2017recovering}. In~\S\ref{sec:ccq-inc} we show how to ensure
this property without such an assumption. The code needed for this
is shown in blue and should be ignored when reading this section.

\subparagraph*{$\wrint(\tsval)$ implementation.}
In $\algname$, a process receiving a $\WRITE$ message piggybacks its
incarnation number and a copy of its crash vector on the $\WRITEACK$ response
(line~\ref{alg:line:sendwrack-tsval}). These are used by the implementation of
$\wrint(\tsval)$ to check whether the set of $\WRITEACK$ responses it receives
is crash-consistent.

This check is integrated within the loop executed by $\wrint$
(line~\ref{alg:line:while}). At each iteration of this loop, a $\WRITE$ message
is (periodically) broadcast to all replicas excluding those in the set $Q$,
which accumulates the replicas whose responses have been already validated as
crash-consistent. Once acknowledgments are received from a set $\newresps$ of replicas
such that $\newresps \cup Q$ is a write quorum (line~\ref{alg:line:wrint-wait}), the
function checks their crash-consistency. The replicas that sent responses that
are not crash-consistent are then purged from $\newresps\cup Q$, and the set of
remaining replicas is reassigned to $Q$ (line~\ref{alg:line:filter-x}).

To identify which replicas $p_j$ in the set $\newresps \cup Q$ sent
crash-consistent responses, the crash vectors received in the $\WRITEACK$
messages are combined into a vector $\tmpcv$ by taking their entry-wise maximum
(line~\ref{alg:line:rdcv}). Then the incarnation numbers $\incvar_j$ in the
$\WRITEACK$ messages are compared against $\tmpcv[j]$.  If
$\incvar_j <\tmpcv[j]$, then $p_j$ restarted while our invocation of $\wrint$
was gathering $\WRITEACK$ responses from other replicas, which indicates that
$p_j$'s prior response is no longer crash-consistent. As a special case, a
replica always treats its own $\WRITEACK$ responses as crash consistent (the
$i=j$ disjunct at line~\ref{alg:line:filter-x}).
If after excluding all replicas whose responses are not crash-consistent, $Q$
contains a write quorum, then $\wrint$ returns
(line~\ref{alg:line:wrint-ret-a2plus}); otherwise, it proceeds to the next
iteration of the loop.

\subparagraph*{Recovery implementation.}
The recovery procedure executed by replica $p_i$ starts by setting its $\stale$
flag to $\TRUE$ (line~\ref{alg:line:setstale-a2+}), thus preventing $p_i$ from
replying to $\READ$ and $\WRITE(\tsval)$ messages
(lines~\ref{alg:line:notstale-a2+} and~\ref{alg:line:notstale-write}). The
replica then selects a new incarnation number
(line~\ref{alg:line:recoverymyincincrement}), which is required to satisfy
Property~\ref{lem:monotone}, and stores it at a write quorum
(line~\ref{alg:line:cvwrite}; we explain the implementation of $\wrint(\cv)$
later). The replica further invokes $\rdint(\ST)$ to retrieve the states from a
read quorum, including both register states and crash vectors
(line~\ref{alg:line:rdbegin}). Finally, the replica reconstructs its state by
merging the crash vectors it received and picking the register value with the
highest timestamp, and clears the $\stale$ flag
(lines~\ref{alg:line:for-k-max}--\ref{alg:line:clearstale}).

The $\READ(\ST)$ messages sent by the $\rdint$ to retrieve the states
(line~\ref{alg:line:rdbegin}) carry the new incarnation number, and a replica
receiving such a message incorporates the new incarnation number into its crash
vector (line~\ref{alg:line:rdstinc}). Thus, the new incarnation number is
written to a read quorum of replicas, which intersects any write quorum that may
be used by a concurrent invocation of $\wrint(\tsval)$. This ensures that this
invocation only accepts a crash-consistent set of responses
(Definition~\ref{defn:crash-consistent}) and rules out the execution in
Example~\ref{examples:crash-consistency}, as we now explain.
\begin{example}
\label{ex:corr}
Assume that in Figure~\ref{fig:problem},
$p_3$'s incarnation number before the crash is $\incnum_1$ and its new
incarnation number after the recovery is $\incnum_2 > \incnum_1$. Then $p_3$ will send $\incnum_1$
in its $\WRITEACK$ response to $p_4$ (line~\ref{alg:line:sendwrack-tsval}), and 
$\incnum_2$ in its $\READ(\ST(\incnum_2))$ request to
$p_2$ (line~\ref{alg:line:rdbegin}). When $p_2$ responds to $p_3$, 
it will record $\incnum_2$ in its crash-vector entry for $p_3$ 
(line~\ref{alg:line:rdstinc}). Then $p_2$ will piggyback this 
crash vector on
its $\WRITEACK$ response to $p_4$'s $\WRITE$ request 
(line~\ref{alg:line:sendwrack-tsval}). This will cause $p_4$ to
discard $p_3$'s $\WRITEACK$ response, since it carries a smaller incarnation 
number ($\incnum_1$) than $p_3$'s entry in the crash vector received from $p_2$ ($\incnum_2$) 
(line~\ref{alg:line:filter-x}).
\end{example}

Note that in the above example, replica $p_2$ might crash and recover after
responding to $p_3$ but before responding to $p_4$. When recovering, $p_2$ will
reconstruct its crash vector to a value no lower than what it was before the
crash (line~\ref{alg:line:updatecv}), thus giving a correct response to $p_4$.
This is ensured by the fact that each recovering replica, such as $p_3$, writes
its incarnation number crash-consistently to a write quorum
(line~\ref{alg:line:cvwrite}), which intersects with a read quorum that a
replica such as $p_2$ uses to reconstruct its state during recovery
(line~\ref{alg:line:rdbegin}).

Thus, a replica $p_i$ writes its incarnation number in the restart handler in
two places -- first crash-consistently to {\em some write quorum}
(line~\ref{alg:line:cvwrite}), and then to the {\em particular read quorum} used
to reconstruct its state (line~\ref{alg:line:rdbegin}). These serve
complementary purposes: the former ensures that other replicas can reconstruct
the incarnation number of $p_i$ when they restart; the latter ensures that other
replicas can detect when the crash-consistency of their writes can be
compromised by a concurrent restart of $p_i$.

\subparagraph*{$\wrint(\cv)$ implementation.}  We now describe the $\wrint$
implementation for incarnation numbers (line~\ref{alg:line:cvwrite}), which
happens to be subtle. We could implement $\wrint(\cv)$ in the exact same way as
$\wrint(\tsval)$, by using crash vectors piggybacked on $\WRITEACK$ messages to
check for crash-consistency of a write; this was the approach taken
in~\cite{michael2017recovering}. Unfortunately, the resulting algorithm would
not be live cases where liveness must be ensured to match the lower bound of
Theorem~\ref{thm:lower-for-dynamic}.
\begin{example}
  \label{ex:live}
  Let $n = 5$, $\ndown = 1$ and $\ncrash = 3$, so that $n \ge \ndown + \ncrash + 1$. Consider
  an execution of the above version of the algorithm where $p_1$ crashes
  permanently, $p_2$ and $p_3$ crash and restart once at the beginning of the
  execution, and the remaining replicas never crash.
  This execution is valid under 
  $\Fd{1, 3, \_}$ (case~\ref{ass:stat} in Definition~\ref{ass:what}). But $p_2$ and $p_3$
  would not be able to complete the recovery in it, because their calls to 
  $\wrint(\cv)$ would wait for $q_w = n - \ndown = 4$ responses
  (line~\ref{alg:line:wrint-wait}). However, only active replicas would respond to
  $\WRITE$ messages (line~\ref{alg:line:notstale-write}), and there are only $2$
  such replicas.
\end{example}

To rule out such executions, in our algorithm a replica accepts a $\WRITE$
message for a write to $\State.\cvvar$ even if its $\stale$ flag is set
(line~\ref{line:write-2-handler}). In this case the replica cannot piggyback its crash
vector on the $\WRITEACK$ response (as in line~\ref{alg:line:sendwrack-tsval}):
the replica may have lost the crash vector upon a restart and has not yet
reconstructed it (line~\ref{alg:line:updatecv}). Hence, the replica puts a dummy
value $\bot$ in place of a crash vector (line~\ref{alg:line:sendwrack-a2+}). The
$\wrint$ function then checks crash-consistency for writes of incarnation
numbers differently from register writes: after it receives enough $\WRITEACK$
responses, it calls $\rdint(\cv)$ to read the crash vectors explicitly
(line~\ref{alg:line:rdcv1}). Replicas only reply to $\READ$ requests sent by
this function when their $\stale$ flag is cleared
(line~\ref{alg:line:notstale-a2+}), and thus they have valid crash vectors in their state.

In the absence of further crashes in Example~\ref{ex:live}, $p_2$ and
$p_3$ can complete the recovery 
as follows: they first collect $\WRITEACK$s from $q_w = n - \ndown = 4$ replicas
that are up ($p_2, p_3, p_4, p_5$) and then read crash vectors from
$q_r = \ndown + 1 = 2$ replicas that are active ($p_4, p_5$).
We can show that $\algname$ is correct assuming Incarnation Number Monotonicity.
\begin{theorem}
  For all $\ndown$, $\ncrash$, if incarnation numbers assigned at 
  line~\ref{alg:line:recoverymyincincrement} satisfy
  Property~\ref{lem:monotone}, then the algorithm in Figure~\ref{fig:pseudo-dyn}
  excluding the highlighted lines is an implementation
  of an atomic MWMR register that is
  \begin{enumerate}
  \item[\emph{(\textsf{A})}] always safe, and
  \item[\emph{(\textsf{B})}] $\Fd{\ndown, \ncrash, M}\text{-live if}\ n\ge d+c+1\ \text{and}\ M \ge 6$.
\end{enumerate}
\label{lem:dynamic-upper-ass}
\end{theorem}

We defer the proof of the theorem to \tr{\ref{app:dynamic}}{B}.
As in~\S\ref{sec:upper}, we show the safety of the algorithm by first
establishing the Real-Time Property (Property~\ref{lem:rtp}).
We next explain how to ensure Incarnation Number Monotonicity.


\subsection{Implementing Incarnation Numbers without Non-Volatile Storage}
\label{sec:ccq-inc}

Since in the presence of rollbacks the replicas cannot rely on non-volatile
storage to implement monotonically growing incarnation numbers, $\algname$ relies on a
distributed mechanism for this purpose, shown in blue in
Figure~\ref{fig:pseudo-dyn}. This mechanism extends the one used to store
incarnation numbers in $\State.\cvvar$ (line~\ref{alg:line:cvwrite}) described
in the previous section. Note that a replica cannot restore its previous
incarnation number by just reading it using $\rdint(\cv)$: if the replica
crashed before completing the $\wrint(\cv)$ at line~\ref{alg:line:cvwrite}, the
$\rdint$ is not guaranteed to restore its latest pre-crash incarnation number,
thus violating Property~\ref{lem:monotone}.

To address this problem, we adopt a two-phase approach. In addition to $\cvvar$,
the state of each replica includes a vector $\precvvar$. Before assigning its new
incarnation number to $\State.\cvvar[i]$ (line~\ref{alg:line:recoverymyincincrement}), a
replica $p_i$ first writes it to $\State.\precvvar[i]$ at a
crash-consistent write quorum of replicas
(line~\ref{alg:line:recoverymyincwrite}).  Upon a restart, a replica $p_i$
restores its previous incarnation number as a maximum of responses returned by
$\rdint(\precv(i))$ and computes the new incarnation number by incrementing the
result (lines~\ref{alg:line:recoverymyincread}--\ref{alg:line:choosenewinc}).
Then the intersection between the quorums used by $\wrint(\precv)$
(line~\ref{alg:line:recoverymyincwrite}) and $\rdint(\precv)$
(line~\ref{alg:line:recoverymyincread}) helps ensure
Property~\ref{lem:monotone}.


There is, however, a subtlety. Recall that in our algorithm a replica has to
reply to a $\WRITE$ message for a write to $\State.\cvvar$ even while it
restarting (line~\ref{line:write-2-handler}), and the same should hold for
writes to $\State.\precvvar$: this is necessary to avoid deadlocks like the one
in Example~\ref{ex:live}. When incarnation numbers are computed as described
above, a restarting replica thus needs to reply to a $\WRITE$ even before it
computed its incarnation number at
line~\ref{alg:line:recoverymyincincrement}. This poses a dilemma: which
incarnation number should the replica include into its $\WRITEACK$ message
(line~\ref{alg:line:sendwrack-tsval}) to allow other replicas to validate the
crash-consistency of their writes to $\State.\cvvar$ and $\State.\precvvar$?
To address this challenge, when a replica $p_j$ executing $\wrint(\cv)$ or
$\wrint(\precv)$ sends a $\WRITE$ message to another replica $p_i$, it
piggybacks the maximum incarnation number of $p_i$ known to it on this message
(line~\ref{alg:line:sendwrite}). A replica $p_i$ receiving a $\WRITE$ request
for $\State.\precvvar[i]$ or $\State.\cvvar[i]$ then adopts this incarnation
number (line~\ref{alg:line:incmyinc}) and uses it in its $\WRITEACK$ reply.

In \tr{\ref{app:dynamic}}{B} we prove that, even though the above
mechanism may require processes to {\em temporarily} adopt stale incarnation
numbers, it does ensure Incarnation Number Monotonicity
(Property~\ref{lem:monotone}). Thus, given Theorem~\ref{lem:dynamic-upper-ass}($\mathsf{A}$),
the overall $\algname$ algorithm is safe; we also prove its liveness in~\tr{\ref{app:dynamic}}{B}. 
Here we illustrate the operation of the incarnation number implementation on an example.

\begin{example}
\label{ex:precv-live}
Let $n = 5$, $\ndown = 1$ and $\ncrash = 3$, so that
$n \ge \ndown + \ncrash + 1$. Consider an execution where $p_1$ crashes
permanently, $p_2$ and $p_3$ are eventually up, and the remaining replicas never
crash. This execution is valid under $\Fd{1, 3, \_}$ (case~\ref{ass:stat} in
Definition~\ref{ass:what}). At the beginning of the execution, $p_2$ and $p_3$
crash, restart, and become active again, so that $p_4$ and $p_5$ have
$\State.\cvvar[2] = \State.\cvvar[3] = 1$. Suppose now that $p_2$ and $p_3$
crash and restart again. 

To recover, each of $p_2$ and $p_3$ determines its previous incarnation number
$1$ (line~\ref{alg:line:recoverymyincread}) and computes a new incarnation number
$2$, which it tries to store crash-consistently in $\State.\precvvar$ at a write
quorum (line~\ref{alg:line:recoverymyincwrite}). Since $p_2$ does not have any
information about the incarnation number of $p_3$, its first $\WRITE$ message to
$p_3$ carries $0$ as the incarnation number (line~\ref{alg:line:sendwrite}),
which $p_3$ includes into its $\WRITEACK$ reply. The $\rdint(\cv)$ that $p_2$
invokes after this (line~\ref{alg:line:rdcv1}) fetches $1$ for
$\State.\cvvar[3]$ from a read quorum $\{p_4, p_5\}$, so the $\WRITEACK$ from
$p_3$ is discarded (line~\ref{alg:line:filter-x}). The same happens at $p_2$
with the $\WRITEACK$ it receives from $p_3$.

Each of $p_2$ and $p_3$ then includes $1$ into the next $\WRITE$ message it
sends to the other replica (line~\ref{alg:line:sendwrite}), which the latter
adopts as its incarnation number and includes into its $\WRITEACK$ reply
(line~\ref{alg:line:incmyinc}). Once $p_2$ receives the $\WRITEACK$ from $p_3$,
it calls $\rdint(\cv)$ once more (line~\ref{alg:line:rdcv1}) and validates the
$\WRITEACK$ as crash-consistent, finishing the execution of
$\wrint(\precv(2, 2))$. Replica $p_2$ then sets $\State.\cvvar[2]$ to its new
incarnation number $2$ (line~\ref{alg:line:recoverymyincincrement}) and
completes the rest of the recovery, including $\wrint(\cv(2, 2))$.

Assume that now $p_3$ receives the $\WRITEACK$ from $p_2$ and calls
$\rdint(\cv)$ to validate it (line~\ref{alg:line:rdcv1}). This read yields $2$
for $\State.\cvvar[2]$, so $p_3$ discards the $\WRITEACK$ and sends another
$\WRITE$ to $p_2$ with incarnation number $2$
(line~\ref{alg:line:sendwrite}). Now $p_2$ responds with a $\WRITEACK$ carrying
its incarnation number $2$, which it also stores locally. When $p_3$ receives
this $\WRITEACK$, it performs another $\rdint(\cv)$, validates the $\WRITEACK$
as crash-consistent and completes $\wrint(\precv(3, 2))$.  It then sets
$\State.\cvvar[3]$ to its new incarnation number $2$
(line~\ref{alg:line:recoverymyincincrement}) and completes the rest of the
recovery. In the end, both $p_2$ and $p_3$ successfully recover with a higher
incarnation number, satisfying Property~\ref{lem:monotone}.
\end{example}

\section{Related Work}
\label{sec:related}

The classical crash-recovery failure model assumed that every process is
equipped with durable storage that cannot be rolled
back~\cite{hurfin1998consensus}. This assumption was lifted by Aguilera et
al.~\cite{aguilera2000failure}, who analyzed consensus solvability as a function
of the number of processes with and without durable storage.
The same model was also used by Guerraoui et
al.~\cite{DBLP:journals/talg/GuerraouiLPP08} to derive tight bounds for a
reliable register construction.  In these models the processes are assumed to
know whether they are equipped with durable storage, which they can trust to be
incorruptible. In contrast, in our $\CRR$ model a recovering process cannot
trust its non-volatile storage to be up-to-date, leading to different resilience
bounds. Furthermore, in the register implementation of Guerraoui et al. a
restarted replica is considered up-to-date once it accepts a single $\WRITE$
request. As we show in~\tr{\ref{app:rachid}}{C}, this leads to a safety
violation.

Dinis et al.~\cite{DBLP:conf/ndss/DinisD023} proposed a rollback-recovery (RR)
model to capture rollback attacks in TEE-based systems. This model is weaker
than $\CRR$ as it disallows permanent process crashes. Furthermore, the register
implementation of Dinis et al. uses ordinary quorums for writes, without
crash-consistency checks. As we show in~\tr{\ref{app:rr}}{D}, this leads to a
safety violation similar to the one in
Example~\ref{examples:crash-consistency}. The implementation also does not
guarantee wait-freedom for reads. In contrast, our $\CRR$ framework captures the
full spectrum of failures in TEE-based systems and enables developing correct
solutions for these environments.

The crash-recovery failure model where the processes do not have access to
incorruptible durable storage was assumed by Chandra et
al.~\cite{chandra2007paxos}, Ko{\'n}czak et al.~\cite{konczak2011jpaxos}, and Liskov
et al.~\cite{liskov2012viewstamped} in the context of their efforts to develop
practical variants of Paxos~\cite{paxos} and viewstamped
replication~\cite{oki1988viewstamped}.  The proposed solutions, however, did not
use crash-consistent quorums for storing their state, and as a result, were
shown in~\cite{michael2017recovering} to violate safety.

Jehl et al.~\cite{replacement} proposed a versioning scheme similar to crash
vectors that allows quickly replacing a failed replica in Paxos.
This technique was subsequently generalized to the notion of crash-consistent
quorums by Michael et al.~\cite{michael2017recovering}, who also demonstrated
how it can be used to implement a recoverable atomic register.
However, the safety of these versioning schemes critically depends on the
replicas' ability to track their incarnations across restarts. In turn,
supporting this capability in TEE-based systems requires hardware-based
persistent counters. These counters are implemented using flash memory,
resulting in poor write performance and
quick wear-out~\cite{203712}.  Furthermore, they are not universally supported
by TEE manufacturers, and have been recently deprecated by Intel
SGX~\cite{intel2023monotonic}.

Although the technical report version~\cite{michael2017recovering-tr} of the
work by Michael et al. suggests that such incarnation numbers can be supported
by means of a distributed mechanism, it does not provide a full
implementation. In addition, while the recoverable register implementation
of~\cite{michael2017recovering,michael2017recovering-tr} is always safe, its
liveness under the static failure models requires $n=2 \ncrash+1$, which is
strictly worse than our bound of $n \ge \ndown + \ncrash + 1$.

Persistent monotonic hardware counters have been demonstrated
in~\cite{strackx2016ariadne} to be a powerful mechanism to guard against
rollback attacks. In TEE-based systems they can be used together with sealing
(\S\ref{sec:intro}) to validate state freshness upon
restart~\cite{costan2016intel}. Unfortunately, these solutions inherit the
drawbacks of persistent counters we explained above.

ROTE~\cite{203712} guards against rollbacks using a rollback-resistant
distributed counter.
Its implementation relies on a two-round protocol where a
new value of the counter is first written to a write quorum, and then read back
from the \emph{same} write quorum to validate that it was stored reliably.
Although this solution is safe, it does not ensure liveness, as the same write
quorum cannot be guaranteed to be available two times in a row under failures
and asynchrony.
Engraft~\cite{DBLP:conf/ccs/0005DNRZ22} uses a similar protocol, so is subject
to the same liveness issue.

\bibliographystyle{plain}
\bibliography{refs}

@InProceedings{michael2017recovering,
	author =	{Michael, Ellis and Ports, Dan R. K. and Sharma, Naveen Kr. and Szekeres, Adriana},
	title =	{{Recovering Shared Objects Without Stable Storage}},
	booktitle =	{Symposium on Distributed Computing (DISC)},
	year={2017}
}

@techreport{michael2017recovering-tr,
	author = {Michael, Ellis and Ports, Dan R. K. and Sharma, Naveen Kr. and Szekeres, Adriana},
	title = {Recovering Shared Objects Without Stable Storage (Extended Version)},
	year = {2017},
	institution = {University of Washington},
	url = {https://www.microsoft.com/en-us/research/publication/recovering-shared-objects-without-stable-storage-extended-version/},
	number = {UW-CSE-TR-17-10-01},
}

@inproceedings{DBLP:conf/ndss/DinisD023,
	author       = {Baltasar Dinis and
		Peter Druschel and
		Rodrigo Rodrigues},
	title        = {{RR:} {A} Fault Model for Efficient {TEE} Replication},
	booktitle    = {Network and Distributed System Security Symposium (NDSS)},
	year         = {2023},
}

@article{DBLP:journals/talg/GuerraouiLPP08,
	author       = {Rachid Guerraoui and
		Ron R. Levy and
		Bastian Pochon and
		Jim Pugh},
	title        = {The collective memory of amnesic processes},
	journal      = {{ACM} Trans. Algorithms},
	volume       = {4},
	number       = {1},
	year         = {2008},
}

@inproceedings{DBLP:conf/ccs/0005DNRZ22,
author       = {Weili Wang and
	Sen Deng and
	Jianyu Niu and
	Michael K. Reiter and
	Yinqian Zhang},
title        = {{ENGRAFT:} {E}nclave-guarded {R}aft on {B}yzantine Faulty Nodes},
booktitle    = {Conference on Computer and Communications Security (CCS)},
year         = {2022},
}

@inproceedings{DBLP:conf/podc/AttiyaBD90,
author       = {Hagit Attiya and
	Amotz Bar{-}Noy and
	Danny Dolev},
title        = {Sharing Memory Robustly in Message-Passing Systems},
booktitle    = {Symposium Principles of Distributed Computing (PODC)},
year         = {1990},
}

@article{DBLP:journals/toplas/HerlihyW90,
	author       = {Maurice Herlihy and
	Jeannette M. Wing},
	title        = {Linearizability: {A} Correctness Condition for Concurrent Objects},
	journal      = {{ACM} Trans. Program. Lang. Syst.},
	volume       = {12},
	number       = {3},
	year         = {1990},
}

@inproceedings{rambo,
  author    = {Nancy A. Lynch and Alexander A. Shvartsman},
  title     = {{RAMBO:} {A} Reconfigurable Atomic Memory Service for Dynamic Networks},
  booktitle = {Symposium on Distributed Computing (DISC)},
  year      = {2002},
}

@inproceedings{grishas-abd,
  author = {Chockler, Gregory and Lynch, Nancy and Mitra, Sayan and Tauber, Joshua},
  title = {Proving Atomicity: An Assertional Approach},
  booktitle = {Symposium on Distributed Computing (DISC)},
  year = {2005}
}

@inproceedings {203712,
author = {Sinisa Matetic and Mansoor Ahmed and Kari Kostiainen and Aritra Dhar and David Sommer 
and Arthur Gervais and Ari Juels and Srdjan Capkun},
title = {{ROTE}: Rollback Protection for Trusted Execution},
booktitle = {USENIX Security Symposium (USENIX Security)},
year = {2017},
xpages = {1289-1306},
}

@inproceedings{abraham2021good,
author = {Abraham, Ittai and Nayak, Kartik and Ren, Ling and Xiang, Zhuolun},
title = {Good-case Latency of {B}yzantine Broadcast: A Complete Categorization},
year = {2021},
booktitle = {Symposium on Principles of Distributed Computing (PODC)},
}

@inproceedings{guo2023tetrabft,
  author       = {Qianyu Yu and
                  Giuliano Losa and
                  Xuechao Wang},
  title        = {{TetraBFT}: Reducing Latency of Unauthenticated, Responsive {BFT} Consensus},
  booktitle    = {Symposium on Principles of Distributed Computing (PODC)},
  year         = {2024},
}

@article{dolev1985bounds,
  title={Bounds on Information Exchange for {B}yzantine Agreement},
  author={Dolev, Danny and Reischuk, R{\"u}diger},
  journal={J. ACM},
  volume={32},
  number={1},
  year={1985},
}

@article{pease1980reaching,
  title={Reaching Agreement in the Presence of Faults},
  author={Pease, Marshall and Shostak, Robert and Lamport, Leslie},
  journal={J. ACM},
  volume={27},
  number={2},
  year={1980},
}

@article{lynch1986easy,
  title={Easy Impossibility Proofs for Distributed Consensus Problems},
  author={Lynch, Nancy A and Merritt, Michael},
  journal={Distributed Comput.},
  volume={1},
  number={1},
  year={1986},
}

@article{dwork1988consensus,
  title={Consensus in the Presence of Partial Synchrony},
  author={Dwork, Cynthia and Lynch, Nancy and Stockmeyer, Larry},
  journal={J. ACM},
  volume={35},
  number={2},
  xpages={288--323},
  year={1988},
  publisher={ACM}
}

@article{liu2019fastbft,
author = {Liu, Jian and Li, Wenting and Karame, Ghassan O. and Asokan, N.},
title = {Scalable {B}yzantine Consensus via Hardware-Assisted Secret Sharing},
year = {2019},
journal = {IEEE Trans. Computers},
volume = {68},
number = {1},
}

@article{veronese2013efficient,
  title={Efficient {B}yzantine Fault-Tolerance},
  author={Veronese, Giuliana Santos and Correia, Miguel and Bessani, Alysson Neves and Lung, Lau Cheuk and Ver{\'\i}ssimo, Paulo},
journal = {IEEE Trans. Computers},
  volume={62},
  number={1},
  xpages={16--30},
  year={2013},
}

@inproceedings{behl2017hybrids,
 title={Hybrids on Steroids: {SGX}-Based High Performance {BFT}},
 author={Behl, Johannes and Distler, Tobias and Kapitza, R{\"u}diger},
 booktitle={European Conference on Computer Systems (EuroSys)},
 xpages={222--237},
 year={2017},
}

@inproceedings{strackx2016ariadne,
  title={Ariadne: A Minimal Approach to State Continuity},
  author={Strackx, Raoul and Piessens, Frank},
  booktitle = {USENIX Security Symposium (USENIX Security)},
  xpages={875--892},
  year={2016},
}

@inproceedings{brandenburger2017rollback,
  author={Brandenburger, Marcus and Cachin, Christian and Lorenz, Matthias and Kapitza, Rüdiger},
  booktitle={Conference on Dependable Systems and Networks (DSN)}, 
  title={Rollback and Forking Detection for Trusted Execution Environments Using Lightweight Collective Memory}, 
  year={2017},
}

@inproceedings{chandra2007paxos,
 title={Paxos Made Live: An Engineering Perspective},
 author={Chandra, Tushar D and Griesemer, Robert and Redstone, Joshua},
 booktitle={Symposium on Principles of Distributed Computing (PODC)},
 xpages={398--407},
 year={2007},
}

@techreport{konczak2011jpaxos,
 title={{JPaxos}: State Machine Replication Based on the {Paxos} Protocol},
 author={Ko{\'n}czak, Jan and Santos, Nuno and {\.Z}urkowski, Tomasz and Wojciechowski, Pawe{\l} T. and Schiper, Andr{\'e}},
 institution={EPFL},
 number={EPFL-REPORT-167765},
 year={2011}
}

@techreport{liskov2012viewstamped,
 title={Viewstamped Replication Revisited},
 author={Liskov, Barbara and Cowling, James},
 institution={MIT CSAIL},
 number={MIT-CSAIL-TR-2012-021},
 year={2012}
}

@article{paxos,
  title={The Part-Time Parliament},
  author={Lamport, Leslie},
  journal={ACM Trans. Comput. Syst.},
  volume={16},
  number={2},
  xpages={133--169},
  year={1998},
  publisher={ACM}
}

@inproceedings{hurfin1998consensus,
 title={Consensus in Asynchronous Systems Where Processes Can Crash and Recover},
 author={Hurfin, Michel and Mostéfaoui, Achour and Raynal, Michel},
 booktitle={Symposium on Reliable Distributed Systems (SRDS)},
 year={1998},
}

@article{aguilera2000failure,
 title={Failure Detection and Consensus in the Crash-Recovery Model},
 author={Aguilera, Marcos K and Chen, Wei and Toueg, Sam},
 journal={Distributed Comput.},
 xpages = {99-125},
 volume={13},
 number={2},
 year={2000},
}

@article{costan2016intel,
 title={Intel {SGX} Explained},
 author={Costan, Victor and Devadas, Srinivas},
 journal={IACR Cryptology ePrint Archive},
 volume={2016},
 pages={86},
 year={2016}
}

@misc{intel2023monotonic,
 title={End of Life for {Intel} {SGX} Monotonic Counter Service},
 author={{Intel Corporation}},
 year={2023},
 howpublished={\url{https://www.intel.com/content/www/us/en/support/articles/000057968/software/intel-security-products.html}},
 note={Accessed: 2025-02-09}
}

@INPROCEEDINGS{replacement,
  author={Jehl, Leander and Lea, Tormod Erevik and Meling, Hein},
  booktitle={Symposium on Reliable Distributed Systems (SRDS)}, 
  title={Replacement: Decentralized Failure Handling for Replicated State Machines}, 
  year={2015},
}

@article{lamport1986interprocess,
  title={On Interprocess Communication},
  author={Lamport, Leslie},
  journal={Distributed Comput.},
  volume={1},
  number={2},
  year={1986},
}

@inproceedings{herlihy2003obstruction,
  title={Obstruction-Free Synchronization: Double-Ended Queues as an Example},
  author={Herlihy, Maurice and Luchangco, Victor and Moir, Mark},
  booktitle={International Conference on Distributed Computing Systems (ICDCS)},
  year={2003},
}

@article{wait-free,
author = {Herlihy, Maurice},
title = {Wait-free synchronization},
year = {1991},
volume = {13},
number = {1},
journal = {ACM Trans. Program. Lang. Syst.},
}

@article{jct98,
author = {Jayanti, Prasad and Chandra, Tushar Deepak and Toueg, Sam},
title = {Fault-tolerant wait-free shared objects},
year = {1998},
volume = {45},
number = {3},
journal = {J. ACM},
xpages = {451–500},
}

@article{agmt95,
author = {Afek, Yehuda and Greenberg, David S. and Merritt, Michael and Taubenfeld, Gadi},
title = {Computing with faulty shared objects},
year = {1995},
volume = {42},
number = {6},
journal = {J. ACM},
xpages = {1231–1274},
}

@inproceedings{gv06,
  author       = {Rachid Guerraoui and
                  Marko Vukolic},
  title        = {How fast can a very robust read be?},
  booktitle    = {Symposium on Principles of Distributed Computing (PODC)},
  xpages        = {248--257},
  year         = {2006},
}

@inproceedings{dgmsv11,
  author       = {Dan Dobre and
                  Rachid Guerraoui and
                  Matthias Majuntke and
                  Neeraj Suri and
                  Marko Vukolic},
  title        = {The complexity of robust atomic storage},
  booktitle    = {Symposium on Principles of Distributed Computing (PODC)},
  xpages        = {59--68},
  year         = {2011},
}

@InProceedings{mad02,
  author="Martin, Jean-Philippe and Alvisi, Lorenzo and Dahlin, Michael",
  title="Minimal {B}yzantine Storage",
  booktitle={Symposium on Distributed Computing (DISC)},
  year="2002",
  xpages="311--325"
}

@inproceedings{atlas14,
  author       = {Dhruva R. Chakrabarti and
                  Hans{-}Juergen Boehm and
                  Kumud Bhandari},
  title        = {Atlas: leveraging locks for non-volatile memory consistency},
  booktitle    = {Conference on Object Oriented Programming, Systems, Languages, and Applications (OOPSLA)},
  xpages        = {433--452},
  year         = {2014},
}

@article{ackm06,
  author       = {Ittai Abraham and
                  Gregory V. Chockler and
                  Idit Keidar and
                  Dahlia Malkhi},
  title        = {{B}yzantine disk {P}axos: optimal resilience with {B}yzantine shared memory},
  journal      = {Distributed Comput.},
  volume       = {18},
  number       = {5},
  xpages        = {387--408},
  year         = {2006},
}

@book{LynchBook,
author = {Lynch, Nancy A.},
title = {Distributed Algorithms},
year = {1996},
isbn = {9780080504704},
publisher = {Morgan Kaufmann Publishers Inc.},
address = {San Francisco, CA, USA}
}

@article{ext,
	title={{TEE} is not a Healer: Rollback-Resistant Reliable Storage (extended version)}, 
	author={Sadegh Keshavarzi and Gregory Chockler and Alexey Gotsman},
        journal   = {arXiv},
	year={2025},
	volume={2505.18648},
	archivePrefix={arXiv},
	primaryClass={cs.DC},
	url={https://arxiv.org/abs/2505.18648}, 
}

@InProceedings{hagit-churn,
  author="Attiya, Hagit
  and Chung, Hyun Chul
  and Ellen, Faith
  and Kumar, Saptaparni
  and Welch, Jennifer L.",
  title="Simulating a Shared Register in an Asynchronous System that Never Stops Changing",
  booktitle={Symposium on Distributed Computing (DISC)},
  year="2015"
}

@article{oki1988viewstamped,
  title={Viewstamped replication: A new primary copy method to support highly-available distributed systems},
  author={Oki, Brian M. and Liskov, Barbara H.},
  journal={Symposium on Principles of Distributed Computing (PODC)},
  xpages={8--17},
  year={1988},
}

@unpublished{adya,
	author = {Atul Adya},
	title = {Weak Consistency: A Generalized Theory and Optimistic Implementations for Distributed 
	Transactions},
	note = {PhD thesis, MIT},
	year = {1999}
}

\iflong

\appendix
\clearpage
\section{Proofs of Results for Static Failure Models}
\label{app:ktol}

\subsection{Lower Bounds}
\label{sec:lower}

We first prove a stronger version of the lower bound part of
Theorem~\ref{thm:main-tight}.
\begin{theorem}
  Assume that all replicas have access to non-volatile storage. Then for all
  $k$, $r$, $b$, there exists a $\FCRR{k,r,b}$-tolerant implementation of an
  obstruction-free safe SWSR register only if $n \ge 2k + \min(b, r) + 1$.
\label{thm:main-lower}
\end{theorem}
\begin{proof}
	We prove why no such implementation exists when $n = 2k + g 
\le 2k + \min(b, r) < 2k + \min(b, r) + 1$ since 
tolerating $k$ permanent replica crashes requires $n \ge 2k + 1$ for any 
implementation of a safe obstruction-free register~\cite{LynchBook}.

Let $\allreps = \{p_1, p_2, \ldots, p_n\}$ be the set of 
all replicas, and define $F_1 = \{p_1, p_2, \ldots, p_k\}$, $F_2 = \{p_{k+1}, 
p_{k + 2}, \ldots, p_{2k}\}$, and $G = \{p_{2k 
	+ 1}, 
p_{2k + 2}, \ldots, p_{2k+ g}\}$. Now assume the contrary to 
the 
theorem statement, i.e., 
there exists an algorithm $\alg{}{}$ that implements a $\FCRR{k,r,b}$-tolerant 
safe obstruction-free SWSR register with $n = 
2k + g$.

First define $\world_1'$ to be an execution where the only crash and 
invocation events happen before all other events and are as follows: replicas 
in $F_1$ crash 
permanently, and then $W = 
\writ(v)$ is invoked. Since $|F_1| 
\le k$, $\world_1'$ is $\FCRR{k,r,b}\text{-valid}$ and $\alg{}{}$ is 
obstruction-free in $\world_1'$, and thus $W$ 
completes.
Now consider execution $\world_1$ to be the finite execution obtained by 
removing 
permanent replica 
crashes from the smallest prefix of $\world_1'$ where $W$ completes. Due to 
asynchronicity, 
receive actions happening at replicas in $F_1$ can be arbitrarily delayed, so 
$\world_1$ 
is also valid 
despite no receive actions at any replica in $F_1$ or messages from $F_1$ being 
delivered.

Now define $\world_2'$ to be an execution where the only crash and 
invocation events happen before all other events and are as follows: replicas 
in $F_2$ crash permanently 
and replicas in 
$G$ crash and restart, then $R = \rd()$ is invoked. 
Since $|F_2| \le k$ and $|G| \le min(b, r) \le b$, $\world_2'$ is 
$\FCRR{k,r,b}\text{-valid}$ and $\alg{}{}$ is obstruction-free in $\world_2'$, 
and 
thus $R$ 
completes. Additionally, $\alg{}{}$ must also be safe in $\world_2'$ and no 
$\writ()$ has been invoked, so $R$ must return the initial value.
Now consider execution $\world_2$ to be the finite execution obtained by 
removing 
permanent replica 
crashes from the smallest prefix of $\world_2'$ where $R$ completes. Due to 
asynchronicity, message receive actions can be 
arbitrarily delayed, so $\world_2$ is valid despite no receive 
actions at any replica in $F_2$ or messages by $F_2$ being delivered.

Now construct execution $\world_3$ by appending actions of $\world_2$ to 
$\world_1$ and rolling back replicas in $G$ upon every restart 
(Figure~\ref{fig:lower}). If $\world_3$ is 
indeed a 
valid execution, 
then it violates safety as $W$ finishes before $R$, but $R$ returns the initial 
value instead of 
value $v$ written by $W$. $\world_1$ is valid on its own, so we only need to  see 
whether the actions 
of $\world_2$ are valid if they occur at the end of $\world_1$. First, $\world_2$ 
only includes message receive 
actions at replicas $F_1 \cup G$. Second, $F_1$ have no receive actions in 
$\world_1$, so 
they are in the 
initial state as required by $\world_2$. Third, $\world_2$ expects replicas in 
$G$ to be in the 
initial state after restarting, which is the case since replicas in $G$ are rolled 
back upon restart in $\world_3$. Finally, the only crash events are those of 
replicas in $G$ crashing and restarting with rollback, and since $|G| \le \min(r,b) 
\le r$, $\world_3$ is $\FCRR{k,r,b}\text{-valid}$. This results in 
$\world_3$ being a 
$\FCRR{k,r,b}\text{-valid}$ execution that violates the safety of $\alg{}{}$, 
thus 
contradicting 
$\alg{}{}$'s safety, so such an $\alg{}{}$ cannot exist.
\end{proof}

\begin{figure}[h]
	\centering
	\begin{tikzpicture}

		\node[ellipse,
		draw,
		minimum width=1.5cm,
		minimum height=3cm] (F11) at (-2,0) {$F_1$};
		
		\draw [dashed] (-1, -2) -- (-1, 2);
		
		\node[circle,
		draw,
		minimum width=1.5cm] (G1) at (0,0) {$G$};
		
		\node[ellipse,
		draw,
		minimum width=1.5cm,
		minimum height=3cm] (F21) at (2,0) {$F_2$};
		
		\draw [->] (3, 0) -- (4.5, 0);
		\node at (3.75, 0.25) {$G$ restart};
		
		\node[ellipse,
		draw,
		minimum width=1.5cm,
		minimum height=3cm] (F11) at (5.5,0) {$F_1$};
		
		\draw [dashed] (8.5, -2) -- (8.5, 2);
		
		\node[circle,
		draw,
		minimum width=1.5cm] (G1) at (7.5,0) {$G$};
		
		\node[ellipse,
		draw,
		minimum width=1.5cm,
		minimum height=3cm] (F21) at (9.5,0) {$F_2$};
		
		\draw [decorate, decoration = {calligraphic brace, amplitude=5pt, mirror}] 
		(-1, -2) -- (3, -2) node [midway, below=2pt] {$\writ(v)$};
		
		\draw [decorate, decoration = {calligraphic brace, amplitude=5pt, mirror}] 
		(4.5, -2) -- (8.5, -2) node [midway, below=2pt] {$\rd(): v_0$};
		
		\draw [decorate, decoration = {calligraphic brace, amplitude=5pt}] 
		(-1, 2) -- (3, 2) node [midway, above=2pt] {$\world_1'$ without crashes};
		
		\draw [decorate, decoration = {calligraphic brace, amplitude=5pt}] 
		(-3, 2.5) -- (3, 2.5) node [midway, above=2pt] {$\world_1$};
		
		\draw [decorate, decoration = {calligraphic brace, amplitude=5pt}] 
		(3, 2) -- (8.5, 2) node [midway, above=2pt] {$\world_2'$ without crashes};
		
		\draw [decorate, decoration = {calligraphic brace, amplitude=5pt}] 
		(3, 2.5) -- (10.5, 2.5) node [midway, above=2pt] {$\world_2$};
	\end{tikzpicture}
	\vspace{-10pt}
	\caption{Illustration of execution $\world_3$}
	\label{fig:lower}
\end{figure}


Using Theorem~\ref{thm:main-lower}, we next prove a stronger version of the 
lower bound part of Theorem~\ref{thm:non-vol-main}.
\begin{theorem}
\label{thm:nvlb}
Let $s$ be the number of replicas with non-volatile storage. Then for all $k$,
$r$, $b$, there exists a $\FCRR{k,r,b}$-tolerant implementation of an 
obstruction-free safe SWSR register only if
$$
(2k + r + 1 \le n < 2k + b + 1 \wedge s \ge 2k + r + 1) \vee
(n \ge 2k + b + 1).
$$
\end{theorem}
\begin{proof}
	Since Theorem~\ref{thm:main-lower} shows that an algorithm exists only if $n \ge 2k + \min(b, r) + 
	1$ or equivalently when $(n \ge 2k + b + 1) \vee (n \ge 2k + r + 1)$, we only need to show that no 
	such algorithm exists with $s < 2k + r + 1$ when $2k + r + 1 \le n = 2k + g 
	< 2k + b + 1$.
	
	Let $\allreps = \{p_1, p_2, \ldots, p_n\}$ be the set of all replicas, and define $F_1 = \{p_1, p_2, 
	\ldots, p_k\}$, $F_2 = \{p_{k+1}, p_{k + 2}, \ldots, p_{2k}\}$, and $G = \{p_{2k + 1}, p_{2k + 2}, 
	\ldots, p_{2k + g}\}$, such that the number of replicas with access to non-volatile storage in $F_1 
	\cup F_2$ is $\min(s, 2k)$. Now assume the contrary to 
	the 
	Theorem statement, i.e., 
	there exists an algorithm $\alg{}{'}$ that implements a safe obstruction-free SWSR register 
	that is 
	$\FCRR{k,r,b}$-tolerant when $2k + r + 1 \le n = 2k + g < 2k + b + 1$ and $s < 
	2k + r + 1$, or 
	equivalently, $s \le 2k + r$.
	
First define $\world_1'$ to be an execution where the only crash and 
invocation events happen before all other events and are as follows: replicas 
in $F_1$ crash 
permanently, and then $W = 
\writ(v)$ is invoked. Since $|F_1| 
\le k$, $\world_1'$ is $\FCRR{k,r,b}\text{-valid}$ and $\alg{}{'}$ is 
obstruction-free in $\world_1'$, and thus $W$ 
completes.
Now consider execution $\world_1$ to be the finite execution obtained by 
removing 
permanent replica 
crashes from the smallest prefix of $\world_1'$ where $W$ completes. Due to 
asynchronicity, 
receive actions happening at replicas in $F_1$ can be arbitrarily delayed, so 
$\world_1$ 
is also valid 
despite no receive actions at any replica in $F_1$ or messages from $F_1$ being 
delivered.

Now define $\world_2'$ to be an execution where the only crash and 
invocation events happen before all other events and are as follows: replicas 
in $F_2$ crash permanently 
and replicas in 
$G$ crash and restart, then $R = \rd()$ is invoked. 
Since $|F_2| \le k$ and $|G| \le min(b, r) \le b$, $\world_2'$ is 
$\FCRR{k,r,b}\text{-valid}$ and $\alg{}{'}$ is obstruction-free in $\world_2'$, 
and 
thus $R$ 
completes. Additionally, $\alg{}{'}$ must also be safe in $\world_2'$ and no 
$\writ()$ has been invoked, so $R$ must return the initial value.
Now consider execution $\world_2$ to be the finite execution obtained by 
removing 
permanent replica 
crashes from the smallest prefix of $\world_2'$ where $R$ completes. Due to 
asynchronicity, message receive actions can be 
arbitrarily delayed, so $\world_2$ is valid despite no receive 
actions at any replica in $F_2$ or messages by $F_2$ being delivered.
	
	Now construct execution $\world_3$ by appending actions of $\world_2$ to 
	$\world_1$ and rolling back replicas with stable storage in $G$ upon every 
	restart. If 
	$\world_3$ is indeed a valid execution, 
	then it violates safety as $W$ finishes before $R$, but $R$ returns the initial value instead of 
	value $v$ written by $W$. $\world_1$ is valid on its own, so we only need to 
	see whether actions 
	of $\world_2$ are valid if they occur at the end of $\world_1$. First, 
	$\world_2$ only includes message receive
	actions at replicas $F_1 \cup G$. Second, $F_1$ have no receive action in 
	$\world_1$, so they are in the 
	initial state as required by $\world_2$. Third, $\world_2$ expects replicas 
	in $G$ to be in the initial state after restarting, which is the case since 
	replicas in $G$ roll back upon restart if they have non-volatile 
	storage. Finally, since there are at most 
	$s - \min(s, 2k) \le r$ 
	replicas with non-volatile storage in $G$, at most $r$ replicas are rolled
	back in $\world_3$. As a result of no other crash taking place in 
	$\world_3$, it is $\FCRR{k,r,b}\text{-valid}$. This results in 
	$\world_3$ being a 
	$\FCRR{k,r,b}\text{-valid}$ execution that violates the safety of 
	$\alg{}{'}$, 
	thus contradicting $\alg{}{'}$'s 
	safety, so such an $\alg{}{'}$ cannot exist given the conditions on $n$ and $s$.
\end{proof}
\subsection{Proof of Linearizability of Algorithm $\alg{}{}$}
\label{app:safety-high}

We now prove that the algorithm $\alg{}{}$ in Figure~\ref{fig:upper-bound} is a
linearizable register implementation satisfying the conditions of the upper
bound in Theorem~\ref{thm:non-vol-main}.
\begin{theorem}
  \label{thm:main-upper}
  The algorithm $\alg{}{}$ in Figure~\ref{fig:upper-bound} is a
  $\FCRR{k,r,b}$-safe implementation of an atomic MWMR register, provided
  $$
  (2k + r + 1 \le n < 2k + b + 1)  \implies  s \ge 2k + r + 1.
  $$
\end{theorem}

To prove the theorem, to each (possibly infinite) execution $\sigma$ of algorithm 
$\alg{}{}$ we
associate:
\begin{itemize}
    \item a set $V(\sigma)$ consisting of all operations in $\sigma$, i.e., 
    $\rd$s and $\writ$s; and
    \item a relation $\RT(\sigma) \subseteq V(\sigma) \times V(\sigma)$, defined 
    as follows:
    for all $o_1,o_2\in V$, $(o_1,o_2) \in \RT(\sigma)$ if and only if $o_1$ 
    completes before $o_2$ is invoked.
\end{itemize}

We denote the $\rd$ operations in $\sigma$ by $\allrdr(\sigma)$ and the $\writ$
operations in $\sigma$ by $\allwrr(\sigma)$. We now define the dependency graph
of an execution as follows inspired by the framework introduced by 
Adya~\cite{adya}.
    \begin{definition}
    \label{defn:dep-graph}
    Let $\sigma$ be an execution. A \emph{dependency graph} of $\sigma$ is a 
    tuple $G = (V(\sigma), \vis, \RT(\sigma), \WR, \WW, \RW)$, where the 
    \emph{visibility} predicate $\vis$ over $V(\sigma)$ and the 
    relations $\WR, \WW, \RW \subseteq V(\sigma) \times V(\sigma)$ are such that:

\begin{enumerate}
	\item $\vis$ holds for all complete operations;
    \item $(i)$ if $(\wrr, \rdr) \in \WR$, then $\wrr \in \allwrr(\sigma)$, 
    $\vis(\wrr)$, $\rdr \in \allrdr(\sigma)$, $\rdr$ completes and $\valu(\wrr) = 
    \valu(\rdr)$;\\
    $(ii)$ for all $\wrr_1, \wrr_2, \rdr \in V(\sigma)$ such that $(\wrr_1, \rdr) 
    \in \WR$ and $(\wrr_2, \rdr) \in \WR$, we have $\wrr_1 = \wrr_2$;\\
    $(iii)$ if $\rdr \in \allrdr(\sigma)$ is complete and there is no $\wrr \in 
    \allwrr(\sigma)$ such that $(\wrr, \rdr) \in 
    \WR$, then $\rdr$ returns the initial value $v_0$; and\\
    $(iv)$ if $\wrr \in \allwrr(\sigma)$ is incomplete and there is no $\rdr \in 
    \allrdr(\sigma)$ such that $(\wrr, \rdr) \in \WR$, then $\neg\vis(\wrr)$;
    \item $\WW$ is a total order over $\{o \in \allwrr\ |\ \vis(o)\}$; and,
    \item $\RW=\{(\rdr,\wrr)\ |\ \exists {\wrr}'.\ ({\wrr}',\rdr)\in\WR \land 
    ({\wrr}',\wrr)\in\WW\} 
    \cup\\
    \phantom\ \ \ \ \ \ \ \{(\rdr,\wrr)\ |\ \rdr\in \allrdr(\sigma)\land
    \wrr\in 
    \allwrr(\sigma)\land \vis(\wrr) \land \neg\exists {\wrr}'.\ 
    ({\wrr}',\rdr)\in\WR\}$.
\end{enumerate}

A dependency graph $G = (V(\sigma), \vis, \RT(\sigma), \WR, \WW, \RW)$ has 
vertices $V$ and directed edges $\RT \cup \WR \cup \WW \cup \RW$.

\end{definition}

To prove 
linearizability we rely on the following theorem.
\begin{theorem}
    \label{thm:linearisability}
    An execution $\sigma$ is linearizable if
    there 
    exist $\vis$, $\WR$, $\WW$, and $\RW$ such that $G = (V(\sigma), \vis, 
    \RT(\sigma), \WR, \WW, \RW)$ is an acyclic dependency graph.
\end{theorem}

\begin{proof}
	Define $\allrdr'(\sigma)$ to be set of incomplete $\rd$s in $\sigma$, and 
	let 
	$V'(\sigma) = \{o \in V\ |\ \vis(o)\} \setminus \allrdr'(\sigma)$. Since 
	$G$ 
	is acyclic, then the subgraph $G'$ of $G$ induced by vertices $V'(\sigma)$ 
	is also acyclic. Thus, the partial order over $V'$ induced by $G'$ can be 
	extended to a total order $\prec$
	such that for every edge $(o_1, o_2) \in G'$, we have $o_1 \prec o_2$. We 
	now argue that 
	the history $\mathcal{H}$ obtained by sequentially ordering operations in 
	$V'(\sigma)$ according to $\prec$ and completing incomplete writes is a 
	linearization of $\sigma$. We first prove that $\mathcal{H}$ is well-formed:
	\begin{itemize}
		\item All complete operations are visible, so $\mathcal{H}$ contains all 
		complete operations.
		\item If $(o_1, o_2) \in 
	\RT$, then $o_1 \prec o_2$, so $\mathcal{H}$ respects the real-time order. 
	
	\item $\mathcal{H}$ is well-formed in that only finitely many operations 
	can precede any complete operation $o$: since $o$ completes in $\sigma$, 
	there exists some state $s$ in the execution where $o$ has returned. 	
	Exactly one action takes place between two consecutive states of the 
	execution, so there can only be finitely many operations that 
	start before $o$ returns. We denote the set of these operations as
	$\mathcal{S}$. Every operation in $V'(\sigma) \setminus \mathcal{S}$ starts 
	after $o$ returns in $\sigma$, so $o$ precedes it in real time. Hence, only 
	operations in $\mathcal{S}$ can precede $o$ in $\prec$ and therefore in
	$\mathcal{H}$.
	
	\item The same holds for incomplete operations. The definition of 
	$V'(\sigma)$ implies that an
	incomplete operation can only be a visible write. By the definition of 
	$\WR$, for any visible write $\wrr$, there exists a complete read $\rdr$ 
	such that $(\wrr, \rdr) \in \WR$. It follows that
	$\wrr \prec \rdr$, so $\wrr$ precedes $\rdr$ in $\mathcal{H}$. Since $\rdr$ 
	has only finitely many predecessors in $\mathcal{H}$, the same
	holds for $\wrr$.
      \end{itemize}
      
	Finally, we show that $\mathcal{H}$ abides by register semantics. That is, 
	for any read operation $\rdr$
	in $\mathcal{H}$, either there exists a latest write operation $\wrr$ that 
	precedes $\rdr$ in $\mathcal{H}$
	and $\valu(\wrr) = \valu(\rdr)$, or
	no such $\wrr$ exists and $\valu(\rdr) = v_0$.
	
	Assume by contradiction that $\wrr$ exists, but $\valu(\wrr) 
	\neq \valu(\rdr)$. There are two cases:
	\begin{itemize}
		\item If $\valu(\rdr) = v_0$, then since no $\wrr' \in 
		\allwrr(\sigma)$ has $\valu(\wrr) = v_0$, there is no $\wrr' 
		\in \allwrr(\sigma)$ such that $(\wrr', \rdr) \in \WR$. According to the 
		definition of $\RW$, we must have $(\rdr, \wrr') \in \RW$ for all 
		visible $\wrr'$. Then due to the definition of 
		$\prec$, we must have $\rdr \prec \wrr$. But this contradicts the 
		assumption that $\wrr$ 
		comes
		before $\rdr$ in $\mathcal{H}$.
		\item If $\valu(\rdr) \neq v_0$, then according to the definition of 
		$\WR$ 
		there must exist $\wrr' \in \allwrr(\sigma)$ such that $\vis(\wrr')$,
		$\valu(\wrr') = \valu(\rdr)$, and $(\wrr', \rdr) \in \WR$, which results 
		in $\wrr' \prec \rdr$. If we have $(\wrr, \wrr') \in \WW$, then $\wrr 
		\prec \wrr' \prec \rdr$, which contradicts $\wrr$ being the last 
		$\writ$ 
		that comes before $\rdr$ in $\mathcal{H}$. If $(\wrr', \wrr) \in \WW$, then according to 
		$\RW$ we must have $(\rdr, \wrr) \in \RW$, so $\rdr \prec 
		\wrr$, which again contradicts $\wrr$ preceding $\rdr$ in $\mathcal{H}$.
	\end{itemize}
	Now assume that no such $\wrr$ exists. If $\valu(\rdr) \neq v_0$, then 
	according to the definition of 
	$\WR$, there must exist $\wrr' \in \allwrr(\sigma)$ such that $\vis(\wrr')$ 
	and $(\wrr', \rdr) \in \WR$. Only finitely 
	many operations precede $\rdr$ in $\mathcal{H}$ and at least one $\writ$ 
	precedes $\rdr$, so there must exist a latest $\writ$ that precedes $\rdr$ in 
	$\mathcal{H}$. This contradicts our assumption of $\wrr$'s non-existence, 
	so we must have 
	$\valu(\rdr) = v_0$.
	
\end{proof}

We now prove that every execution of the algorithm
is linearizable. Fix one such execution $\sigma$.
Our strategy is to find witnesses for $\vis$, $\WR$, $\WW$ and $\RW$ that 
validate the conditions of Theorem~\ref{thm:linearisability}.
To this end, consider the function $\tau : 
\sigma \rightarrow
(\mathbb{N}\cup\{\infty\})\times\mathbb{N}$
that maps each operation in
$\sigma$
to a timestamp as 
follows:

\begin{itemize}
	\item if $o \in \allrdr(\sigma)$ and $o$ executes 
	line~\ref{alg:read:selectmax}, then $\tau(o)$ equals $\tmpts$ as defined at 
	line~\ref{alg:read:selectmax};
	\item if $o \in \allwrr(\sigma)$ and $o$ executes 
	line~\ref{alg:line:writeinctag}, 
	then $\tau(o)$ equals $\tmpts$ as defined at 	
	line~\ref{alg:line:writeinctag};
	\item otherwise $\tau(o)$ equals $(\infty, i)$ where $p_i$ is the process 
	that executes $o$
\end{itemize}

We then define the required witnesses as follows:
\begin{itemize}
	\item $\vis(o)$ is true if $o$ is complete or if $o \in 
	\allwrr(\sigma)$ and exists a complete read $\rdr$ such that $\tau(o) = 
	\tau(\rdr)$;
    \item $(\wrr,\rdr)\in\WR$ if and only if $\wrr \in \allwrr(\sigma)$, 
    $\vis(\wrr)$, $\rdr \in \allrdr(\sigma)$, $\rdr$ is complete, and 
    $\tau(\wrr)=\tau(\rdr)$;
    \item $(\wrr,{\wrr}')\in\WW$ if and only if $\wrr,{\wrr}' \in     
    \allwrr(\sigma)$, $\vis(\wrr) \wedge \vis({\wrr}')$, and 
    $\tau(\wrr)<\tau({\wrr}')$; and
    \item $\RW$ is derived from $\WR$ and $\WW$ as per the dependency graph 
    definition.
\end{itemize}


Our proof relies on the following proposition, whose easy proof we omit.
\begin{proposition}
  \label{abd:tags_prop}
  \
    \begin{enumerate}
        \item For every $\wrr_1,\wrr_2\in \allwrr(\sigma)$, 
        $\tau(\wrr_1)=\tau(\wrr_2)$ 
        implies $\wrr_1=\wrr_2$.\label{abd:tags_prop2}
        \item For every $\wrr\in \allwrr(\sigma)$, 
        $\tau(\wrr)>(0,0)$.\label{abd:tags_prop:3}
        
        \item For every $\rdr\in \allrdr(\sigma)$ such that $\rdr$ is complete, 
        either $\tau(\rdr)=(0,0)$ or there exists $\wrr\in \allwrr(\sigma)$ such 
        that $\vis(\wrr)$ and $\tau(\rdr)=\tau(\wrr)$.\label{abd:tags_prop:1}
        \item \label{abd:tags_prop:4} For every $\rdr\in \allrdr(\sigma)$ and $\wrr\in \allwrr(\sigma)$,
            $\tau(\rdr)=\tau(\wrr)$ implies $\opval(\rdr)=\opval(\wrr)$. 
    \end{enumerate}
\end{proposition}

This proposition suffices to show that the instances of $\vis$, $\WR$, $\WW$, 
and $\RW$ as defined above satisfy the conditions stated in Definition~\ref{defn:dep-graph}. Our 
proof also relies on the following auxiliary lemma:

\begin{lemma}
  \label{abd:tau-non-inc}
  \
    \begin{enumerate}
        \item For all $\rdr,\wrr\in V(\sigma)$, if $(\rdr,\wrr)\in\RW$ then 
        $\tau(\rdr)<\tau(\wrr)$.
        \item For all $o_1,o_2\in V(\sigma)$, if $(o_1,o_2)\in\RT$, then 
        $\tau(o_1)\leq\tau(o_2)$.
            Moreover, if $o_2$ is a $\writ$, then 
            $\tau(o_1)<\tau(o_2)$.\label{abd:tau-non-inc2}
    \end{enumerate}
\end{lemma}

\begin{proof}
    \begin{enumerate}
    \item Let $\rdr,\wrr\in V(\sigma)$ be such that $(\rdr,\wrr)\in\RW$. There 
    are two cases.
        \begin{itemize}
            \item Suppose that for some
                ${\wrr}'$ we have $({\wrr}',r)\in\WR$ and $({\wrr}',\wrr)\in\WW$. 
                The
                definition of $\WR$ implies that $\tau(\rdr)=\tau({\wrr}')$, and 
                the definition of
                $\WW$ implies that $\tau({\wrr}')<\tau(\wrr)$. Then 
                $\tau(\rdr)<\tau(\wrr)$.
            \item Suppose now that $\neg \exists {\wrr}'.\ ({\wrr}',\rdr)\in\WR$.
                We show that $\tau(\rdr)=(0,0)$.
                Indeed, if $\tau(\rdr)\neq(0,0)$, then by 
                Proposition~\ref{abd:tags_prop}.\ref{abd:tags_prop:1}, there 
                exists $\wrr\in W(\sigma)$
                such that $\tau(\rdr)=\tau(\wrr)$. But then $(\wrr,\rdr)\in\WR$, 
                contradicting the
                assumption that there is no such write.
                At the same time, 
                Proposition~\ref{abd:tags_prop}.\ref{abd:tags_prop:3} implies that
                $\tau(\wrr)>(0,0)$. Then $\tau(\rdr)<\tau(\wrr)$.
        \end{itemize}
      \item Let $o_1$, $o_2 \in V(\sigma)$ be such that $(o_1, o_2) \in
        \RT(\sigma)$. Operation $o_1$ must complete to be able to precede any
        other operation. As a result, if $o_1$ is a \writ, it writes $\tau(o_1)$
        to $\State.\Tsvar$ using $\wrint$ at
        line~\ref{alg:line:recwritewriteint}, and if $o_1$ is a \rd, it does the
        same at line~\ref{alg:line:readwriteback}.  If $o_2$ is a \writ, it
        reads a set of $\State.\Tsvar$, the highest of which we denote as
        $\tmpts'$, from a quorum using $\rdint$ at
        line~\ref{alg:line:writereadtag}, and if $o_2$ is a \rd, it does this at
        line~\ref{alg:line:recreadreadint}. Since $(o_1, o_2) \in \RT$,
        operation $o_2$ starts reading $\State.\Tsvar$ after $o_1$ finishes
        writing $\tau(o_1)$ to $\State.\Tsvar$, and by the Real-Time Property
        (Property~\ref{lem:rtp}) we must have $\tmpts' \ge \tau(o_1)$. If $o_2$
        is a \rd, $\tmpts'$ is used as $\tau(o_2)$
        (line~\ref{alg:read:selectmax}) and we have
        $\tau(o_2) = \tmpts' \ge \tau(o_1)$. If $o_2$ is a \writ, we have
        $\tau(o_2) > \tmpts'$ according to lines~\ref{alg:line:writeselectmax}
        and~\ref{alg:line:writeinctag}, and thus
        $\tau(o_2) > \tmpts' \ge \tau(o_1)$. This proves the statement.
\end{enumerate}
\end{proof}

\begin{lemma}
    \label{thm:dep-graph-acyc}
    $G=(V(\sigma), \vis, \RT(\sigma), \WR, \WW, \RW)$ is an acyclic dependency 
    graph.
\end{lemma}
\begin{proof}
    From Proposition~\ref{abd:tags_prop} and the
    definitions of $\vis$, $\WR$, $\WW$ and $\RW$ it easily follows
    that $G$ is a dependency graph, so we only need to show that $G$ is acyclic.
    By contradiction, assume that the graph contains a cycle $o_1, \dots, o_n = 
    o_1$.
    Then $n > 1$. By Lemma~\ref{abd:tau-non-inc} and the definitions of $\tau$, 
    $\WW$, and $\WR$, we must have $\tau(o_1) \leq \dots \leq \tau(o_n) = 
    \tau(o_1)$,
    so that $\tau(o_1) = \dots = \tau(o_n)$. Furthermore,
    if $(o,o')$ is an edge of $G$ and $o'$ is a write, then $\tau(o)<\tau(o')$.
    Hence, all the operations in the cycle must be reads, and thus, all the edges
    in the cycle come from $\RT(\sigma)$. Then there exist reads $\rdr_1,\rdr_2$
    in the cycle such that $\rdr_1$ completes before $\rdr_2$ is invoked and 
    $\rdr_2$
    completes before $\rdr_1$ is invoked, which is a contradiction.
\end{proof}

\begin{proof}[Proof of Theorem~\ref{thm:main-upper}.]
  Follows from Lemmas~\ref{thm:linearisability} and~\ref{thm:dep-graph-acyc}.
\end{proof}

\subsection{Results for Unknown $r$ and $b$}
\label{app:unk-tight}

\begin{corollary}
	Suppose that either $r$ or $b$ is unknown, and let $r \triangleq k$ if $r$ is 
	unknown, and 
	$b\triangleq n$ if $b$ is unknown. 
	Then it holds:
	there exists a $\FCRR{k,r,b}$-tolerant implementation of a wait-free atomic MWMR 
	register if and only if $n \ge 2k + \min(b, r) + 1$.
	\label{corol:main-tight-corol}
\end{corollary}
\begin{proof}
	Below we prove the corollary for the case 
	when $r$ is unknown and $b$ is known. The proof for the 
	remaining two cases ($r$ known and $b$ unknown, and both 
	$r$ and $b$ are unknown) is similar.
	
	\emph{Only if}:
	By Theorem~\ref{thm:main-tight}, there exists an algorithm 
	$A(k, r, b)$ that is atomic and wait-free when $n \ge 2k + \min(b, r) + 1$, at most $k$ replicas are 
	$\CRR$-faulty, at most $r$ 
	replicas experience rollback failures, and at most than $b$ replicas 
	are benign. Define	$A'(k, b)$ to be 
	$A(k, k, b)$.
	By definition, $A'(k, b) = A(k, k, b)$ is atomic and wait-free when $n \ge 2k + \min(b, k) 
	+ 1$, at most $k$ replicas are $\CRR$-faulty, at most $k$ replicas experience rollback failures, 
	and at most than $b$ replicas are benign. Thus, $A'(k, b)$ is also atomic and wait-free when at 
	at most $r$ replicas experience rollback failures for any $r 
	\le k$. It follows that $A'(k, b)$ is a $\FCRR{k,r,b}$-tolerant implementation of a wait-free atomic 
	MWMR 
	register for $\CRR$
	if $n \ge 2k + \min(b, k) + 1$ and $r$ is unknown. 
	
	\emph{If}: Assume that there exists an algorithm $A(k, b)$ that 	
	implements a $\FCRR{k,r,b}$-tolerant atomic wait-free MWMR register with 
	$n \le 2k+\min(b, k)$ replicas and unknown $r$. Then $A(k, b)$ is 
	$\FCRR{k,r,b}$-tolerant with $n = 3k$ and any value of $r \le k$,
	and in particular, when $r=k=b$. 
	However, in this case, Theorem~\ref{thm:main-tight} implies that $n \ge 
	2k + 
	\min(b,r)+1=3k+1$. A contradiction. This means that $A(k,b)$ can only be 
	$\FCRR{k,r,b}$-tolerant with $n \ge 2k + \min(b, k) + 1$. 
\end{proof}

 \begin{restatable}{corollaryr}{unknownlower}
	Suppose that either $r$ or $b$ is unknown, and let $r \triangleq k$ if $r$ is unknown, and
	$b\triangleq n$ if $b$ is unknown. Then for all algorithms $A$
	implementing an obstruction-free safe SWSR register, if $A$ is $\FCRR{k,r,b}$-tolerant, then
	$n \ge 2k + \min(b, r) + 1$.
	\label{corol:main-lower-unknown}
\end{restatable}

\begin{proof}
	Here we prove the corollary for the case when $r$ is unknown and $b$ is known. The proof for 
	the remaining two cases ($r$ known and $b$ unknown, and both $r$ and $b$ are unknown) is 	
	similar.
	
	Assume that there exists an algorithm $A(k,b)$ that implements a $\FCRR{k,r,b}$-tolerant safe 	
	obstruction-free MWMR register with $n \le 2k+\min(b, k)$ replicas and unknown $r$. 
	Then $A(k,b)$ is $\FCRR{k,r,b}$-tolerant with $n = 3k$ and any value of $r \le k$, and in 
	particular, when $r=k=b$. However, in this case, Theorem~\ref{thm:main-lower} implies that $n 
	\ge 2k + \min(b,r)+1=3k+1$. A contradiction. This means that $A(k, b)$ can only be 
	$\FCRR{k,r,b}$-tolerant with $n \ge 2k + \min(b, k) + 1$. 
\end{proof}
\section{Proofs of Results for Dynamic Failure Models}
\label{app:dynamic}

\subsection{Lower Bound}
\label{app:dynamic-lb}

\begin{lemma}
	\label{lem:prereddyn}
	For any $\ncrash$, $M$, $\ndown'$ and $\ndown''$ such 	
	that $\ndown' \ge \ncrash$ and $\ndown'' \ge \ncrash$, we have 	
	$\Fd{\ndown', \ncrash, M} \preceq \Fd{\ndown'', \ncrash, M}$
\end{lemma}

\begin{proof}
	Since $\ndown' \ge \ncrash$, then any $\Fd{\ndown', \ncrash, M}\text{-valid}$ 	
	execution $\alpha$ must belong to case 2 in Definition~\ref{ass:what}. As a	
	result, there exist $n - \ncrash$ replicas that never crash in $\alpha$. 	
	This means that there are at least $n - \ndown'' \le n - \ncrash$ replicas 	
	in $\alpha$ that are eventually up because $\ndown'' \ge c$. As a 	
	consequence of this and the existence of $n - c$ replicas that never crash 	
	in $\alpha$, $\alpha$ is $\Fd{\ndown'', \ncrash, M}\text{-valid}$ according 
	to case 2 in Definition~\ref{ass:what} and the statement is proven.
\end{proof}

\begin{lemma}
	\label{lem:reddyn}
	For all $\ndown$, $\ncrash$, and $M$, if $\ndown > \ncrash$, then $\Fd{\ndown, \ncrash, M} = 
	\Fd{\ncrash, 
	\ncrash, M}$.
\end{lemma}

\begin{proof}
	Since $\ndown \ge \ncrash$ and $\ncrash \ge \ncrash$, we can use Lemma~\ref{lem:prereddyn} 
	and conclude that $\Fd{\ndown, \ncrash, M} \preceq \Fd{\ncrash, \ncrash, M}$. Using $\ncrash \ge 
	\ncrash$, $\ndown \ge \ncrash$, and Lemma~\ref{lem:prereddyn} we can conclude that 
	$\Fd{\ncrash, \ncrash, M} \preceq \Fd{\ndown, \ncrash, M}$. Finally, $\Fd{\ndown, \ncrash, M} 
	\preceq \Fd{\ncrash, \ncrash, M}$ and $\Fd{\ncrash, \ncrash, M} \preceq \Fd{\ndown, \ncrash, M}$ 
	result in $\Fd{\ndown, \ncrash, M} = \Fd{\ncrash, \ncrash, M}$.
\end{proof}

\begin{proof}[Proof of Theorem~\ref{thm:lower-for-dynamic}]
	Without loss of generality, assume that $\ndown \le \ncrash$. This is 
	allowed because if $\ndown > \ncrash$, then according to 
	Lemma~\ref{lem:reddyn} we can replace $\Fd{\ndown, \ncrash, M}$ with 
	$\Fd{\ncrash, \ncrash, M}$.
	
	Suppose $n \le \ndown + \ncrash$ and that there exists an algorithm 
	$\alg{}{}$ 
	implementing an 
	$\Fd{\ndown,\ncrash,M}$-tolerant safe obstruction-free SWSR register. Let 
	$\allreps = \{p_1, 
	p_2, 
	\ldots, p_n\}$ be the set of 
	all replicas, and define $F = \{p_i \in \allreps\ |\ 1 \le i \le \ndown\}$, 
	and $G = \{p_i \in \allreps\ |\ \ndown + 1 \le i \le \ndown + \ncrash\}$. 
	Since $n \le 
	\ndown + 
	\ncrash$, we have $\allreps = F \cup G$.

	First define $\world_1'$ to be an infinite execution where
	replicas in $F$ crash 
	permanently, then $W = 
	\writ(v)$ is invoked, and no other crash or invocation takes place.
	Since 
	$|F| \le \ndown \le \ncrash$, there are $n - \ndown \ge n - \ncrash$ 
	replicas in $\world_1'$ that never crash and $\world_1'$ is $\Fd{\ndown, 
	\ncrash, M}\text{-valid}$ according to case 2 in Definition~\ref{ass:what}.
	Due to $\alg{}{}$ being obstruction-free in $\Fd{\ndown, \ncrash, M}$, 
	$W$ must complete. Now consider execution $\world_1''$ to be the finite 
	execution 
	obtained by removing 
	permanent replica 
	crashes from the smallest prefix of $\world_1'$ where $W$ completes. Then obtain execution 
	$\world_1$ by speeding up $\world_1''$ such that the last action takes 
	place no later than time 
	$\frac{\Delta}{3}$. This is 
	possible since there is no lower bound on how 
	long actions take to 
	complete. Since $\world_1$ ends before time $\Delta$, it is valid 
	despite no receive events taking place at any replica in $F$ or any message 
	from $F$ being delivered to other replicas after the invocation of $W$.
	
	Now define $\world_2'$ to be an extension of the longest prefix in 
	$\world_1$ that does not contain the invocation of $W$. Require $\world_2'$ 
	be an execution where replicas in $G$ crash and restart, then $R = \rd()$ 
	is invoked, and no other crash or invocation takes place.
	All replicas are eventually up in $\world_2'$ and there are $n - |G| \ge n 
	- c$ 
	replicas that never crash, so $\world_2'$ is $\Fd{\ndown, \ncrash, 
	M}\text{-valid}$ by case 2 of Definition~\ref{ass:what}.
	$\alg{}{}$ is obstruction-free in $\Fd{\ndown, \ncrash, M}$, so $R$ must
	complete. Additionally, no 
	$\writ()$ has 
	been invoked and $\alg{}{}$ is safe in $\Fd{\ndown, \ncrash, M}$, so $R$ 
	must return the initial value.
	Now consider the finite execution $\world_2''$ to be the the smallest 
	prefix of 
	$\world_2'$ where $R$ completes. Then obtain execution 
	$\world_2$ by speeding up $\world_2''$ such that the last action takes 
	place no later than 
	$\frac{\Delta}{3}$, which is possible since there is no lower bound on how long actions take to 
	complete.
	
	Now construct the finite execution $\world_3'$ by appending actions of 
	$\world_2$ to 
	$\world_1$ and rolling back all replicas upon restart. Let $\world_3$ be an 
	infinite extension of $\world_3'$ without any other crashes or invocations 
	than those in $\world_3'$. If $\world_3$ is 
	indeed a valid execution, 
	then it violates safety as $W$ finishes before $R$, but $R$ returns the 
	initial value instead of 
	value $v$ written by $W$. $\world_1$ is valid on its own, so we only need 
	to 
	see whether the actions 
	of $\world_2$ are valid if they occur at the end of $\world_1$.
	First, $F$ have no receive event in 
	$\world_1$ after the invocation, so they are in the 
	state expected by $\world_2$. Second, $\world_2$ expects 
	replicas in $G$ to be in the 
	initial state after restarting, which is the case since replicas in $G$ 
	restart and all replicas roll back upon restart in $\world_3$. Third, all 
	replicas are 
	eventually 
	up and because $|G| \le \ncrash$ there are 
	$n - \ncrash$ replicas that never crash, so $\world_3$ is $\Fd{\ndown, 
	\ncrash, M}\text{-valid}$ by case 2 of Definition~\ref{ass:what}. As such, 
	the timing of crashes does not have to be restricted and
	$\alg{}{}$ must be safe in $\world_3$. Finally, 
	since $\world_3'$ ends no later than time $\frac{2\Delta}{3} < \Delta$, all 
	message
	receive events can be delayed until after $\world_3'$ if 
	needed. This 
	results in $\world_3$ being a 
	$\Fd{\ndown, \ncrash, M}\text{-valid}$ execution that violates the safety 
	of $\alg{}{}$, thus contradicting 
	$\alg{}{}$'s safety, so such an $\alg{}{}$ cannot exist.
\end{proof}

\subsection{Proof of Safety for $\algname$ Assuming Monotone Incarnation Numbers}
\label{app:safety-black}

In this section we prove the safety of the algorithm in
Figure~\ref{fig:pseudo-dyn} excluding the code in blue, which assumes that the
incarnation numbers assigned at line~\ref{alg:line:recoverymyincincrement}
satisfy Incarnation Number Monotonicity (Property~\ref{lem:monotone},
\S\ref{sec:ccq-cv}).

Define $\tsset \subseteq \mathbb{N} \times \mathbb{N}$ to be the set of all
timestamps, and $\incset \subseteq \mathbb{N}$ to be the set of all incarnation
numbers.
\begin{definition}
	Consider a completed call $\wrq$ to $\wrint$ that writes a value
	$x \in \tsset \cup \incset$ to $\State.\var$. We use the following notation:
	\begin{itemize}
		\item $\rdcv[\wrq]$: the last time when $\wrq$ starts executing
		line~\ref{alg:line:rdcv1} if $x \in \incset$,
		\item $\finish[\wrq]$: the time at which $\wrq$ executes
		line~\ref{alg:line:wrint-ret-a2plus} and returns.
		\item $Q[\wrq]$: the value of the variable $Q$ at time $\finish[\wrq]$.
		\item $\acks[\wrq]$: the set of $\WRITEACK(\idvar, \incvar_j, \tmpcv_j, j)$
		messages received in $\wrq$ with the highest $\incvar_j$ for each replica
		$p_j \in Q[\wrq]$
		\item $\resptime[\wrq, p_j]$: the time at which $p_j \in Q[\wrq]$ sends its
		message in $\acks[\wrq]$.
	\end{itemize}
	We say that $\wrq$ is \textbf{durable} iff: for every $p_j \in Q[\wrq]$,
        at all times $t \ge \resptime[\wrq, p_j]$, if $p_j$ is active, then
        $p_j$ has $\State.\var \ge x$.
\end{definition}

\begin{lemma}
\label{prop:read}
Consider a completed and durable call $\wrq$ to $\wrint$ that writes
$x \in \tsset \cup \incset$ to $\State.\var$. Any call $\rdq$ to $\rdint$
invoked after $\finish[\wrq]$ to read $\State.\var$ returns a set containing a
value $\ge x$ for $\State.\var$.
\end{lemma}
\begin{proof}
  Consider one such $\rdq$ and let $\req$ be the request such that
  $\rdq=\rdint(\req)$. By the code of $\rdint$, $\rdq$ first sends
  $\READ(\idvar, \req)$ messages to all replicas where $\idvar$ is the
  identifier associated with $\req$, and then waits to receive $\READACK$
  responses from some set $Q_r$ of replicas such that
  $|Q_r| = q_r \ge \ndown+1$. On the other hand, by the code of $\wrint$, we get
  $|Q[\wrq]| = q_w \ge n - \ndown$. Thus, there exists a replica
  $p_i \in Q[\wrq] \cap Q_r$. In other words, there exists a replica
  $p_i \in Q[\wrq]$ whose $\READACK(\idvar, \_, i)$ message was received by
  $\rdq$.

  Since the identifier $\idvar$ is assigned uniquely by the $\newid$ function,
  $p_i$ must have sent $\READACK(\idvar,\_,i)$ after receiving
  $\READ(\idvar, \req)$. By our assumption, $\rdq$ starts after $\finish[\wrq]$,
  so $p_i$ responds to $\rdq$ after $\finish[\wrq] > \resptime[\wrq, p_i]$.
  Additionally, only active replicas respond to reads
  (line~\ref{alg:line:notstale-a2+}) and $\wrq$ is durable, so $p_i$ must have
  $\State.\var \ge x$ when it responds to $\rdq$. Then $\rdq$ receives a
  $\READACK$ message from $p_i$ with a value $\ge x$ for $\State.\var$, which
  $\rdq$ includes into its response.
\end{proof}

\begin{lemma}
  \label{prop:inc}
  If replica $p_i$ has value $\incnum$ for $\State.\cvvar[i]$ at time $t$, and
  $p_i$ assigns $\incnum'$ to $\State.\cvvar[i]$ at
  line~\ref{alg:line:recoverymyincincrement} at time $t' > t$, then
  $\incnum' > \incnum$.
\end{lemma}
\begin{proof}
  Replica $p_i$ must assign $\incnum$ to $\State.\cvvar[i]$ at
  line~\ref{alg:line:recoverymyincincrement} at some time $t'' \le t <
  t'$. Thus, according to Incarnation Number Monotonicity, we must have
  $\incnum < \incnum'$.
\end{proof}

\begin{lemma}
  Consider a completed call $\wrq$ to $\wrint$ by a process $p_j$ that writes
  some $x \in \incset$, and assume that all completed calls ${\wrq}'$ that write
  some $x' \in \incset$ with $\finish[{\wrq}'] \le \rdcv[\wrq]$ are durable. If
  a replica $p_i \in Q[\wrq] \setminus \{p_j\}$ crashes at time
  $\tcrash > \resptime[\wrq, p_i]$ and starts executing
  line~\ref{alg:line:rdbegin} at time $\trstate > \tcrash$, then
  $\trstate > \rdcv[\wrq]$.
\label{claim:read-cv-read-state}
\end{lemma}

\begin{proof}
  By contradiction, assume that $\trstate \le \rdcv[\wrq]$, as shown in
  Figure~\ref{fig:imp}.  Process $p_i$ crashes at time $\tcrash$ and later
  starts executing line~\ref{alg:line:rdbegin} at time $\trstate$ as part of
  recovery, so $p_i$ must also execute
  lines~\ref{alg:line:recoverymyincincrement} and~\ref{alg:line:cvwrite}. Let
  $g$ be the value of $\nextinc$ when line~\ref{alg:line:recoverymyincincrement} is executed, 
  $\Winc$ be the call to $\wrint$ at
  line~\ref{alg:line:cvwrite}, and $\twinc$ be the time when it starts. Then
$$
\tcrash < \twinc < \finish[\Winc] < \trstate \le \rdcv[\wrq],
$$
so $\Winc$ must be durable by our assumption. Let $\rdq$ be the read at
line~\ref{alg:line:rdcv1} executed by $\wrq$ at time
$\rdcv[\wrq] > \finish[\Winc]$. Then by Lemma~\ref{prop:read}, one of the
vectors returned by $\rdq$ must have a value $\ge g$ in the entry for $p_i$.
Let $g_0$ be the incarnation number in $p_i$'s message in $\acks[\wrq]$, which
$p_i$ sends at $\resptime[\wrq, p_i] < \twinc \le \rdcv[\wrq]$
(line~\ref{alg:line:sendwrack-a2+}).  By Lemma~\ref{prop:inc}, we must
have $g > g_0$. Thus, one of the vectors returned by the read $\rdq$ must have a
value $>g_0$ in the entry for $p_i$, i.e., a value greater than the incarnation
number in $p_i$'s message in $\acks[\wrq]$.  Then after completing $\rdq$ at
line~\ref{alg:line:rdcv1} the write $\wrq$ discards $p_i$'s response at
line~\ref{alg:line:filter-x}. But this contradicts the assumption of
$p_i \in Q[\wrq]$.
\end{proof}

\begin{figure}[t]
  \hspace{-6pt}
	\scalebox{0.85}{
	\begin{tikzpicture}
		
		\tikzmath{\pminy = -1; \pjy = -1;}
		
		\node at (-0.5, 0) {$p_j$};
		\node at (-0.5, \pminy) {$p_i$};
		
		\node at (7, 0.25) {$\wrq$};
		\node [font=\footnotesize] at (11.75, 0.25) {$\rdcv[\wrq]$};
		\node [font=\footnotesize] at (15, 0.25) {$\finish[\wrq]$};
		\node at (12.8, -0.6) {$\rdq$};
		
		\draw [|-|] (0,0) -- (11.75,0);
		\draw [|-|] (11.75,0) -- (14,0);
		\draw [|-|] (14,0) -- (15,0);
		\draw [decorate, decoration = {calligraphic brace, mirror, 
		amplitude=5pt}] 
		(11.75,-0.2) --  
		(14,-0.2);
		
		
		
		
		\draw [dotted] (0,\pminy) -- (11.75,\pminy);
		
		\draw [->] (0.5, 0) -- (1.25, \pminy);
		\draw [->] (1.25, \pminy) -- (2, 0);
		\node [font=\footnotesize] at (0.3, \pminy / 2) {$\WRITE$};
		\node [font=\footnotesize] at (2.6, \pminy / 2) {$\WRITEACK$};
		
		\node [font=\footnotesize] at (1.5, \pminy - 0.25) {$\resptime[\wrq, 
		p_i]$};
		\node at (3.3, \pminy) {$\times$};
		
		\draw [|-|] (4.7, \pminy) -- (7.6, \pminy);
		\draw [|-|] (7.6, \pminy) -- (9.9, \pminy);
		\draw [|-] (11.5, \pminy) -- (11.75, \pminy);
		\draw [dashed] (11.75, \pminy) -- (12.5, \pminy);
		\draw [dotted] (11.75, \pminy) -- (11.75, 0);
		\draw [decorate, decoration = {calligraphic brace, amplitude=5pt}] 
		(4.7,\pminy + 0.2) -- (9.9,\pminy + 0.2);
		
		\node at (3.6, \pminy - 0.25) {$\tcrash$};
		\node at (5.2, \pminy - 0.25) {$\twinc$};
		\node at (7.5, \pminy + 0.6) {$\Winc$};
		\node at (12, \pminy - 0.25) {$\trstate$};
		
		\node [font=\footnotesize] at (7.7, \pminy - 0.25) {$\rdcv[\Winc]$};
		\node [font=\footnotesize] at (10.1, \pminy - 0.25) {$\finish[\Winc]$};
		
	\end{tikzpicture}
	
}
	
	\caption{An illustration of the proof of 
	Lemma~\ref{claim:read-cv-read-state}.}
	\label{fig:imp}
\end{figure}

\begin{lemma}
  Consider a completed call $\wrq$ to $\wrint$ by some process $p_j$ that writes
  some $x \in \incset$ such that all completed calls ${\wrq}'$ that write some
  $x' \in \incset$ with $\finish[{\wrq}'] \le \rdcv[\wrq]$ are durable. Then the
  set of responses $\acks[\wrq]$ is crash-consistent.
\label{lem:write-cv-cc}
\end{lemma}
\begin{proof}
  By Lemma~\ref{claim:read-cv-read-state} we know that if a replica
  $p_i \in Q[\wrq] \setminus \{p_j\}$ crashes at time
  $\tcrash > \resptime[\wrq, p_i]$ and starts executing
  line~\ref{alg:line:rdbegin} at time $\trstate > \tcrash$, then
  $\trstate > \rdcv[\wrq]$. We also know that $p_j$ completes $\wrq$, so if
  $p_j \in Q[\wrq]$ and it crashes at time $\tcrash > \resptime[\wrq, p_j]$, we
  must have $\tcrash > \finish[\wrq] > \rdcv[\wrq]$ and thus $p_j$ starts
  executing line~\ref{alg:line:rdbegin} at time $\trstate > \tcrash > \rdcv[\wrq]$. By the
  definition of $\rdcv[\wrq]$, we must have
  $\resptime[\wrq, p_i] \le \rdcv[\wrq]$ for every $p_i \in Q[\wrq]$. Hence,
  every $p_i \in Q[\wrq]$ responds to any $p_k \in Q[\wrq]$'s $\rdint(\ST(\_))$ only after
  $\resptime[\wrq, p_i]$, so the set of responses $\acks[\wrq]$ is
  crash-consistent.
\end{proof}

\begin{lemma}
\label{lemma:recoverywritesafe}
All completed calls to $\wrint(\req)$ that write some $x \in \incset$ are durable.
\end{lemma}
\begin{proof}
  Assume the contrary: there is at least one completed call to $\wrint$ that
  writes some $x \in \incset$ and is not durable.  Define $\mathcal{W}$ to be
  the set of all completed calls to $\wrint$ that write some $x \in \incset$ and
  are not durable. Let $\Wmin \in \mathcal{W}$ be such that
  $$
  \forall \wrq.\ \wrq \in \mathcal{W} \implies \rdcv[\Wmin] \le \rdcv[\wrq].
  $$
  Thus, any completed call $\wrq$ that writes some $x \in \incset$ such that
  $\rdcv[\wrq] < \rdcv[\Wmin]$ must be durable.  Suppose $\Wmin$ writes a value
  $g \in \incset$ to $\State.\var$. Since $\Wmin$ is not durable, there must
  exist some replica $p_i \in Q[\Wmin]$ and time $t \ge \resptime[\Wmin, p_i]$
  such that $p_i$ is active at $t$, but has $\State.\var < g$. Let $\tmin$ be
  the earliest such $t$, and $\pmin$ be the corresponding replica $p_i$.

Since state values do not decrease in our algorithm while a replica is up, and
$\pmin$ has $\State.\var \ge g$ at time $\resptime[\Wmin, \pmin]$, the replica
$\pmin$ must crash at some time between $\resptime[\Wmin, \pmin]$ and
$\tmin$. Let $\tlc$ be the latest time when $\pmin$ crashes before $\tmin$.
Since $\pmin$ is active at time $\tmin$, it must complete its recovery code
between $\tlc$ and $\tmin$. Let $\tlrstate$ and $\tstale$ be the times when
$\pmin$ starts executing lines~\ref{alg:line:rdbegin}
and~\ref{alg:line:clearstale}, respectively, during this recovery. Then
$\tlc < \tlrstate < \tstale \le \tmin$.

Note that for any completed call $\wrq'$ that writes some $x \in \incset$ with
$\finish[\wrq'] \le \rdcv[\Wmin]$, we have
$\rdcv[\wrq'] < \finish[\wrq'] \le \rdcv[\Wmin]$, so by the definition of
$\Wmin$, $\wrq'$ is durable. Thus, by applying Lemma~\ref{lem:write-cv-cc} with
$\wrq = \Wmin$, we establish that $\acks[\Wmin]$ is crash-consistent.

Let $\rdq$ be the $\rdint$ that replica $\pmin$ invokes at time $\tlrstate$ to
retrieve the state (line~\ref{alg:line:rdbegin}). The write quorum $Q[\Wmin]$
and the read quorum used by $\rdq$ must intersect in at least one replica $p'$.
Since $\acks[\Wmin]$ is crash-consistent, this replica must responds to $\rdq$
at some time $\tresp > \resptime[\Wmin, p']$ (see
Figure~\ref{fig:rebuild}). Since only active replicas repond to reads
(line~\ref{alg:line:notstale-a2+}), $p'$ must be active at $\tresp$. But then
$p'$ must have $\State.\var \ge g$ at time $\tresp$: otherwise, at this time,
$p' \in Q[\Wmin]$ is active and has $\State.\var < g$, and furthermore
$\resptime[\Wmin, p'] < \tresp < \tstale < \tmin$, which contradicts the
definition of $\tmin$. Thus, $p'$ responds to $\rdq$ with a $\READACK$ message
containing a value $\ge g$ for $\State.\var$. Then by
lines~\ref{alg:line:for-k-max}-\ref{alg:line:clearstale}, $\pmin$ has
$\State.\var \ge g$ at time $\tstale$. Since $\pmin$ does not crash between
$\tlc$ and $\tmin$, it will also have $\State.\var \ge g$ at time $\tmin$. But
this contradicts the definition of $\pmin$. This contradiction shows that all
$\wrint$s that write some $x \in I$ are durable.
\end{proof}

\begin{figure}[t]
	\centering
	
	\begin{tikzpicture}[scale=0.87]
		
		\node at (-0.5, 0) {$\Wmin$};
		\node at (-0.5, -1) {$p'$};
		\node at (-0.5, -2) {$\pmin$};
		
		\node [font=\footnotesize] at (3.6, 0.25) {$\rdcv[\Wmin]$};
		
		\node [font=\footnotesize] at (2, -1.25) {$\resptime[\Wmin, p']$};
		\node at (7, -0.75) {$\tresp$};
		
		\node at (12.8, -2.25) {$\tstale$};
		\node at (14.2, -2.25) {$\tmin$};
		\node at (7, -2.25) {$\tlrstate$};
		\node at (8.75, -2.25) {$RQ$};
		
		\draw [-|] (0.5,0) -- (3.5,0);
		\draw [dashed] (3.5,0) -- (4.5,0);
		
		\draw [dotted] (0,-1) -- (14,-1);
	
		\draw [dotted] (0, -2) -- (12.25, -2);
		\draw [|-|] (6, -2) -- (11.5, -2);
		\draw [|-|, dotted] (12.25, -2) -- (14, -2);
		
		\draw [->] (6.5, -2) -- (6.75, -1);
		\draw [->] (6.75, -1) -- (8, -2);
		
		\draw [->] (1, 0) -- (1.75, -1);
		\draw [->] (1.75, -1) -- (2.5, 0);
		
	\end{tikzpicture}
	\caption{Visualization of $\pmin$ rebuilding its state}
	\label{fig:rebuild}
\end{figure}

\begin{lemma}
\label{lemma:incnumpersist}
Consider a completed call $\rdq$ to $\rdint(\ST(\State.\cvvar[i]))$ at
line~\ref{alg:line:rdbegin} by a replica $p_i$ with
$\State.\cvvar[i] = \incnum$. Any replica $p_j$ that receives a $\READ$ message
from $\rdq$ at time $\tresp$ will have $\State.\cvvar[i] \ge \incnum$ at all
times $t \ge \tresp$ when $p_j$ is active.
\end{lemma}
\begin{proof}
  Assume the contrary: there exists some time $t' > \tresp$ such that $p_j$ is
  active at $t'$, but has $\State.\cvvar[i] < \incnum$. Define $\tmin$ to be the
  earliest such $t'$. Since state values do not decrease in our algorithm while a
  replica is up, and $p_j$ has $\State.\cvvar[i] = \incnum$ at time $\tresp$
  (line~\ref{alg:line:rdstinc}), replica $p_j$ must crash at some time after
  $\tresp$ and before $\tmin$. Let $\tlc$ be the latest time when $p_j$ crashes
  before $\tmin$. Since $p_j$ is active at time $\tmin$, it must complete its
  recovery code between $\tlc$ and $\tmin$. Let $\tlrstate$ and $\tstale$ be the
  times when $p_j$ starts executing lines~\ref{alg:line:rdbegin}
  and~\ref{alg:line:clearstale}, respectively, during this recovery. Then
  $\tresp < \tlc < \tlrstate < \tstale \le \tmin$.
	
  Replica $p_i$ must finish the call $\wrq$ to $\wrint(\cv(i, g))$ at line~\ref{alg:line:cvwrite} before 
  starting $\rdq$, so
  $\finish[\wrq] < \tresp < \tlrstate$. The call $\rdq'$ to $\rdint$ at
  line~\ref{alg:line:rdbegin} at time $\tlrstate$ by $p_j$ starts after
  $\finish[\wrq]$. Since $\wrq$ is durable by Lemma~\ref{lemma:recoverywritesafe}, $\rdq'$ must 
  return a value $\ge \incnum$ for $\State.\cvvar[i]$ by Lemma~\ref{prop:read}. Then $p_j$ must 
  have $\State.\cvvar[i] \ge \incnum$ at
  time $\tstale$. Since $p_j$ does not crash after $\tlc$ and before $\tmin$,
  this replica must have $\State.\cvvar[i] \ge \incnum$ at time $\tmin$, which
  contradicts the definition of $\tmin$. This contradiction shows the required.
\end{proof}

\begin{lemma}
\label{lemma:registerwritecc}
For all completed calls $\wrq$ to $\wrint(\req)$ that write some $x \in \tsset$,
the set of responses $\acks[\wrq]$ is crash-consistent.
\end{lemma}
\begin{proof}
  Consider any replica $p_i \in Q[\wrq]$ that invokes
  $\rdq = \rdint(\ST(\State.\cvvar[i]))$ at time
  $\trstate > \resptime[\wrq, p_i]$. Since $p_i$ is active at time
  $\resptime[\wrq, p_i]$ and inactive at time $\trstate$, then $p_i$ must crash
  at time $\tcrash$ such that $\resptime[\wrq, p_i] < \tcrash <
  \trstate$. Suppose $p_i$ has $\State.\cvvar[i] = \incnum_1$ at time
  $\resptime[\wrq, p_i]$ and has $\State.\cvvar[i] = \incnum_2$ at time
  $\trstate$. Since $\trstate > \resptime[\wrq, p_i]$ we must have $\incnum_1 < 
  \incnum_2$ according to Lemma~\ref{prop:inc}.
	
  Now assume that $\acks[\wrq]$ is not crash-consistent, and thus, there exists
  some replica $p_j \in Q[\wrq]$ that responds to $\rdq$ at a time
  $\tresp < \resptime[\wrq, p_j]$. By Lemma~\ref{lemma:incnumpersist} for
  $\incnum = \incnum_2$, replica $p_j$ must have
  $\State.\cvvar[i] \ge \incnum_2$ at time $\resptime[\wrq, p_j]$. Then,
  according to lines~\ref{alg:line:wrint-wait}--\ref{alg:line:filter-x}, $\wrq$
  must have $\tmpcv[i] \ge \incnum_2$ when executing
  line~\ref{alg:line:filter-x} for the last time. But $p_i$'s response in
  $\acks[\wrq]$ sent at time $\resptime[\wrq, p_i]$ has incarnation number
  $\incnum_1 < \incnum_2$. Hence, $\wrq$ must discard this response,
  contradicting the definition of $\resptime[\wrq, p_i]$. This contradiction
  shows the required.
\end{proof}

\begin{lemma}
\label{lemma:registerwritesafe}
All completed calls to $\wrint(\req)$ that write some $x \in \tsset$ are
durable.
\end{lemma}
\begin{proof}
  Assume the contrary: there exists a call $\wrq$ to $\wrint$ that writes some
  $x \in \tsset$ to $\State.\var$, a replica $p_i \in Q[\wrq]$, and a time
  $t \ge \resptime[\wrq, p_i]$, such that $p_i$ is active at time $t$, but has
  $\State.\var < x$. Define $\tmin$ to be the earliest such $t$.  Since state
  values do not decrease in our algorithm while a replica is up, and $p_i$ has
  $\State.\var \ge x$ at time $\resptime[\wrq, p_i]$, replica $p_i$ must crash
  at some time after $\resptime[\wrq, p_i]$ and before $\tmin$.  Let $\tlc$ be
  the latest time when $p_i$ crashes before $t$. Since $p_i$ is active at time
  $\tmin$, it must complete its recovery code between $\tlc$ and $t$.  Let
  $\tlrstate$ and $\tstale$ be the times when $p_j$ starts executing
  lines~\ref{alg:line:rdbegin} and~\ref{alg:line:clearstale}, respectively,
  during this recovery.  Then
  $\resptime[\wrq, p_i] < \tlc < \tlrstate < \tstale \le \tmin$.
	
  The call $\rdq$ to $\rdint$ at line~\ref{alg:line:rdbegin} at time $\tlrstate$
  by $p_i$ must contact $\rdqsize$ replicas out of all replicas except $p_i$,
  since $p_i$ is not active before $\rdq$ finishes executing. Since read and
  write quorums intersect, there must exist some replica $p_j \in Q[\wrq]$ that
  responds to $\rdq$, such that $p_j \neq p_i$. Since only active replicas
  repond to reads, $p_j$ must be active when responding to $\rdq$ at time
  $\tresp$, and we must have $\tresp < \tstale \le \tmin$. By
  Lemma~\ref{lemma:registerwritecc} $\tresp > \resptime[\wrq, p_j]$, so $p_j$
  has $\State.\var \ge x$ at time $\tresp$, since $p_j$ is active at this time
  and $\tresp < \tmin$. As a result, $p_j$'s response to $\rdq$ has
  $\State.\var \ge x$, so replica $p_i$ will also have $\State.\var \ge x$ at
  time $\tstale$. Since $p_i$ does not crash after $\tstale$ and before $\tmin$,
  this replica must have $\State.\var \ge x$ at time $\tmin$, which contradicts
  our definition of $\tmin$. The contradiction shows the required.
\end{proof}

We can now prove that \algname{} is safe assuming Incarnation 
Number Monotonicity (Property~\ref{lem:monotone}).
\begin{proof}[Proof of Theorem~\ref{lem:dynamic-upper-ass}($\mathsf{A}$)]
  Lemmas~\ref{prop:read} and~\ref{lemma:registerwritesafe} imply the Real-Time
  Order Property (Property~\ref{lem:rtp}). Then the rest of the linearizability
  proof follows the one given in \S\ref{app:safety-high}.
\end{proof}

\subsection{Proof of Safety for the Full $\algname$ Algorithm}
\label{app:safety-full}

We now prove that the complete algorithm in Figure~\ref{fig:pseudo-dyn},
including the code in blue, satisfies the safety requirements of
Theorem~\ref{thm:upper-dynamic} and complete the proof. Note that when
considering the complete algorithm including the code in blue, calls to $\wrint$
that write some $x \in \incset$ include $\wrint(\precv(\_, x))$ in addition to
$\wrint(\cv(\_, x))$ that was considered in the previous section.

\begin{lemma}
\label{lem:inc}
Given a time $T$, assume that all completed calls $\wrq$ to $\wrint$ that write
some $x \in \incset$ and have $\finish[\wrq] < T$ are durable. If a replica
$p_i$ assigns $\incnum$ to $\State.\cvvar[i]$ at
line~\ref{alg:line:recoverymyincincrement} at time $t \le T$, and $p_i$ assigns
$\incnum'$ to $\State.\cvvar[i]$ at line~\ref{alg:line:recoverymyincincrement}
at time $t' > t$, then $\incnum' > \incnum$.
\end{lemma}
\begin{proof}
  Before $\incnum \in \incset$ is assigned to $\State.\cvvar[i]$ at
  line~\ref{alg:line:recoverymyincincrement} at time $t$,
  $\wrint(\precv(i, \incnum))$ at line~\ref{alg:line:recoverymyincwrite} must
  return. Let $\wrq$ be the corresponding call to $\wrint$. Since
  $\finish[\wrq] < t \le T$ and $\incnum \in \incset$, the write $\wrq$ is durable
  by our assumption. Hence, by Lemma~\ref{prop:read}, before $p_i$ assigns
  $\incnum'$ to $\State.\cvvar[i]$ at line~\ref{alg:line:recoverymyincincrement}
  at time $t' > t$, $p_i$ reads a set at line~\ref{alg:line:recoverymyincread}
  that contains a value $\ge \incnum$. Then by lines~\ref{alg:line:choosenewinc}
  and~\ref{alg:line:recoverymyincincrement}, we must have $\incnum' > \incnum$.
\end{proof}


\begin{lemma}
\label{prop:inc_pre}
Given a time $T$, assume that all completed calls $\wrq$ to $\wrint$ that write
some $x \in \incset$ and have $\finish[\wrq] < T$ are durable. If a replica
$p_i$ has $\State.\cvvar[i] = \incnum$ at time $t \le T$, and $p_i$ assigns
$\incnum'$ to $\State.\cvvar[i]$ at line~\ref{alg:line:recoverymyincincrement}
at time $t' > t$, then $\incnum' > \incnum$.
\end{lemma}
\begin{proof}
  The only place a new value is introduced to $\State.\cvvar[i]$s is as part of
  $p_i$'s recovery at line~\ref{alg:line:recoverymyincincrement}, since the only
  other places that $\State.\cvvar[i]$ might change 
  (lines~\ref{alg:line:updatecv},~\ref{alg:line:rdstinc},~\ref{alg:line:incmyinc},
  and~\ref{alg:line:cvset}) use values already present at other replicas. This
  means that replica $p_i$ must assign $\incnum$ to $\State.\cvvar[i]$ at
  line~\ref{alg:line:recoverymyincincrement} at some time $t'' \le t \le T$.  We
  also have $t'' \le t < t'$, so according to Lemma~\ref{lem:inc}, we must have
  $\incnum < \incnum'$.
\end{proof}

We can now prove the safety part of the upper bound in
Theorem~\ref{thm:upper-dynamic}.

\begin{lemma}
	\label{lem:alg-safe}
  The algorithm in Figure~\ref{fig:pseudo-dyn} is an MWMR register that is
  atomic under $\CRR$.
\end{lemma}
\begin{proof}
  Lemma~\ref{claim:read-cv-read-state} is proven for the full algorithm by
  replacing Lemma~\ref{prop:inc} with Lemma~\ref{prop:inc_pre} for
  $T = \rdcv[\wrq]$ in its proof. Using Lemma~\ref{claim:read-cv-read-state}, we
  can prove Lemmas~\ref{lem:write-cv-cc}--\ref{lemma:incnumpersist} as 
  before. Lemmas~\ref{lemma:recoverywritesafe} and~\ref{prop:inc_pre} imply 
  Lemma~\ref{prop:inc}, which then we use alongside 
  Lemma~\ref{lemma:incnumpersist} to prove 
  Lemmas~\ref{lemma:registerwritecc}--\ref{lemma:registerwritesafe} as 
  before. Finally, Lemmas~\ref{prop:read} and~\ref{lemma:registerwritesafe} imply 
  the Real-Time Order Property (Property~\ref{lem:rtp}), and the rest of the 
  linearizability proof follows the one given in \S\ref{app:safety-high}.
\end{proof}

  Note that Lemmas~\ref{lemma:recoverywritesafe} and~\ref{lem:inc} imply 
  Incarnation Number Monotonicity (Property~\ref{lem:monotone}) for the full 
  $\algname$ algorithm.



\subsection{Proof of Liveness for $\algname$}
\label{app:liveness-black}

In the following, we will refer to the version of the $\algname$ protocol
as \emph{basic} if it does not include the blue code, and as \emph{enhanced},
otherwise. 
We will write $\State.\var_i(t)$ to denote the value of $\State.\var$ at a replica
$p_i$ at time $t$.

\begin{lemma}
	\label{lem:maxinc-never-crash}
	Let $\alpha$ be an execution of the $\algname$ protocol (basic or enhanced).
	Then for all replicas $p_i$ and times $t$, if a replica $p_i$ does not crash before $t$ in $\alpha$, then 
	for all $j$, and times $s<t$, $\State.\cvvar[i]_j(s)=0$.
\end{lemma}
\begin{proof}
	New values for $\State.\cvvar[i]$ at any replica can only be generated by $p_i$
	when it executes line~\ref{alg:line:recoverymyincincrement} of the $\onrestart$
	block. Since $p_i$ never executes $\onrestart$ before $t$ in $\alpha$, for all
	times $s<t$, the value 
	of $\State.\cvvar[i]_j(s)$ remains equal to its initial value $0$ at all replicas, 
	as required.
	\end{proof}

\begin{lemma}
	\label{lem:maxinc}
	Suppose that a replica $p_i$ executes the assignment statement in 
	line~\ref{alg:line:recoverymyincincrement} at time $t$ and remains 
	up until $t' > t$. 
	Then, for all times $s \in [t, t']$, $s' \le t'$, and replicas
	$p_j$, $\State.\cvvar[i]_j(s') \le \State.\cvvar[i]_i(s)$.
\end{lemma}

\begin{proof}
	Let $\incnum$ be the value assigned to $\State.\cvvar[i]_i$ 
	in line~\ref{alg:line:recoverymyincincrement} at time $t$. 
	Fix $s' \le t'$, and suppose that $\State.\cvvar[i]_j(s')=\incnum'$ for some
	replica $p_j$. Assume by contradiction that $\incnum' > \incnum$. 
	Then $\incnum'$ cannot be the initial value of $\State.\cvvar[i]$.
	By the structure of the code, any value (apart from the initial one) that 
	a process $p_j$ can adopt for its $\State.\cvvar[i]_j$ variable must be first 
	generated and assigned to $\State.\cvvar[i]_i$ by executing
	line~\ref{alg:line:recoverymyincincrement}. Let $t''$ be the time
	when this event occurs at $p_i$. Since $\incnum' > \incnum$, by 
	Incarnation Number Monotonicity (Property~\ref{lem:monotone}),
	$t'' > t$. 

	Since the code in line~\ref{alg:line:recoverymyincincrement}
	is a part of the $\onrestart$ block, and therefore, can only be executed when 
	$p_i$ crashes and restarts, and $p_i$ is up for the entire
	duration of $[t, t']$, we must have $t'' > t'$. Therefore,
	$p_j$ can only learn of $\incnum'$ by receiving a message from another process 
	in either one of the lines~\ref{alg:line:updatecv},~\ref{alg:line:rdstinc},~\ref{alg:line:incmyinc},
	or~\ref{alg:line:cvset}, which implies $s' \ge t'' > t'$. Since
	$s' \le t'$, this is a contradiction. 
\end{proof}

\begin{lemma}
	\label{lem:livelem-enh}
	For every execution $\alpha$ of the enhanced $\algname$ algorithm and all replicas $p_i$, if 
	$p_i$ is eventually up in $\alpha$,
	then there exists a time $T$ such that for all times $t > T$, 
	$\State.\cvvar[i]_i(t)=\State.\cvvar[i]_i(T)$.
\end{lemma}

\begin{proof}
	Assume towards a contradiction that for all times $T$, there exists 
	a time $t > T$ such that $\State.\cvvar[i]_i(t)\neq \State.\cvvar[i]_i(T)$. 
	Since $p_i$ is eventually up in $\alpha$, it can only execute the $\onrestart$
	procedure finitely many times. Since new values can only be assigned 
	to $\State.\cvvar[i]_i$ as a result of executing 
	line~\ref{alg:line:recoverymyincincrement} of $\onrestart$,
	there exist finitely many values which $\State.\cvvar[i]_i$
	can be assigned to in $\alpha$. Therefore, there must exist infinitely many times $T$ and $t$ 
	such that $T < t$, and $\State.\cvvar[i]_i(t) < \State.\cvvar[i]_i(T)$.
	However, by the structure of the code,  $\State.\cvvar[i]_i$ can never decrease while $p_i$ is up, 
	and $p_i$ is eventually up. A contradiction. 
\end{proof}





\paragraph*{Basic crash consistency results.}
For a message $m$ sent by a process (client or replica) in an execution of
$\algname$, we write $\proc(m)\in \allprocs$ to denote the sender of $m$,
and $\id(m)$ to denote the message identifier associated with $m$. For 
$m=\WRITEACK(\_, \incvar, \tmpcv, \_)$, 
we let $\inc(m)=\incvar$ and $\crv(m)=\tmpcv$, and 
for $m=\READACK(\_, x, \_)$ such that $x$ is a crash vector,
we let $\crv(m)=x$.

For an invocation $a\in \{\wrint(\req), \rdint(\req)\}$ 
in an execution of $\algname$, we let $\proc(a)\in \allprocs$  
denote the process (client or replica) invoking $a$ and 
$\reqq(a)=\req$. For an invocation $\rdq=\rdint(\_)$, we let $\id(\rdq)$
denote the message identifier generated at line~\ref{alg:line:id-gen-rdint}. 
For an invocation $\wrq=\wrint(\req)$, we let $\id(\wrq)$ denote the 
message identifier generated in line~\ref{alg:line:x-init},  and
$\wacks(\wrq, t)$ the set of $\WRITEACK$ messages $m$ with $\id(m)=\id(\wrq)$
received by $\proc(\wrq)$ by time $t$. 
Additionally, if $\req \in \{\precv(\_, \_), \cv(\_, \_)\}$, we let 
$\racks(\wrq, t)$ denote the union of all sets of $\READACK$ messages
collected by operations $\rdint(\cv)$ that are invoked by $\wrq$ and complete by time $t$.
%
%
%
%
%
Below, we introduce a predicate $\CC$, which, given an invocation $\wrq=\wrint(\req)$, determines whether
the set of $\WRITEACK$ responses received by $\proc(\wrq)$
from a given replica $p_j$ is crash-consistent.
\begin{itemize}
	\item[] \raggedright $\CC(\wrq, j, t)$:
	if $\proc(\wrq) \neq p_j$, then
	$\exists m \in \wacks(\wrq, t)$ such that $\proc(m)=p_j$, and
	it holds:
	\emph{(i)} $\forall m' \in \wacks(\wrq, t).\, \crv(m')[j] \le \inc(m)$
	if $\reqq(\wrq)\in \{\tsval(\_, \_)\}$, and
	\emph{(ii)} $\forall m' \in \racks(\wrq, t).\,
	\crv(m')[j] \le \inc(m)$, otherwise.



\end{itemize}
\begin{proposition}
	\label{prop:cc}
	Consider an invocation $\wrq=\wrint(\_)$, and 
	let $Q$ be the set of replicas initialized  by $\wrq$ at line~\ref{alg:line:x-init}.
	Then, for all times $t$ at which $\wrq$ reaches line~\ref{alg:line:filter-x},
	we have: $p_j \in Q$ upon completion of line~\ref{alg:line:filter-x}
	iff $\CC(\wrq, j, t)$.
\end{proposition}

The next lemma asserts that every $\WRITEACK$ reply received from a process
that has not previously crashed is crash-consistent. 
\begin{lemma}
	\label{lem:livelem-cc-init}
	Consider a $\wrint$ invocation $\wrq$ in an execution $\alpha$
	of $\algname$ (basic or enhanced). Then, for all 
	times $t$ in $\alpha$, if there exists $m\in \wacks(\wrq, t)$
	such that $p_j=\proc(m)\neq \proc(\wrq)$ and $\proc(m)$
	never crashes at or before $t$ in $\alpha$, then $\CC(\wrq, j, t)$ holds.
\end{lemma}

\begin{proof}
	Consider $m \in \wacks(\wrq, t)$, and  let $s$ be the time at which $m$ was sent by $p_j$.
	Thus, $s \le t$ and, by the code of the $\WRITE$ message handler, 
	we also have $\inc(m) = \State.\cvvar[j]_j(s)$.
	Consider $m' \in \wacks(\wrq, t) \cup \racks(\wrq, t)$, and let $s'$
	be the time at which $m'$ is sent by $p_k=\proc(m')$.
	Then, $\crv(m')[j]=\State.\cvvar[j]_k(s')$.
	Since $p_j$ does not crash in the interval $[0, t]$, 
	by Lemma~\ref{lem:maxinc-never-crash}, $\State.\cvvar[j]_j(s)=\State.\cvvar[j]_k(s')=0$.
	Hence,
	$$
		\inc(m)=\State.\cvvar[j]_j(s)=\State.\cvvar[j]_k(s')=\crv(m')[j]=0,
	$$ 
	and the result follows.
\end{proof}

The following lemma shows that every $\WRITEACK$ reply received from a process
after it has executed the assignment in line~\ref{alg:line:recoverymyincincrement}, 
and before it crashes again, is crash-consistent.
\begin{lemma}
	\label{lem:livelem-cc-simple}
	Consider a $\wrint$ invocation $\wrq$ in an execution $\alpha$
	of $\algname$ (basic or enhanced). Suppose that a replica
	$p_j$ executes the assignment in line~\ref{alg:line:recoverymyincincrement}
	at time $t$ in $\alpha$ and remains up until $t' \ge t$.
	Then if there exists $m\in \wacks(\wrq, t')$
	sent by $p_j=\proc(m)\neq \proc(\wrq)$ in $[t, t']$, then 
	$\CC(\wrq, j, t')$ holds.
\end{lemma}

\begin{proof}
	Consider $m \in \wacks(\wrq, t')$, and let 
	$s\in [t, t']$ be the time at which 
	$m$ was sent by $p_j$. By the code of the $\WRITE$ message handler, we have
$$
	s \in [t, t'] \wedge \inc(m) = \State.\cvvar[j]_j(s).
$$
	Consider $m' \in \wacks(\wrq, t') \cup \racks(\wrq, t')$, and let $s'$
	be the time at which $m'$ is sent by $p_k=\proc(m')$.
	Thus, we have:
	$$
		s' \le t' \wedge \crv(m')[j]=\State.\cvvar_k[j](s').
	$$
	Then since $p_j$ is up in $[t, t']$, Lemma~\ref{lem:maxinc} implies
	$\State.\cvvar_k[j](s') \le \State.\cvvar_j[j](s)$, and therefore,
	$\crv(m')[j] \le \inc(m)$. Hence, $\CC(\wrq, j, t')$ holds, as required.
\end{proof}

The next lemma addresses the case of crash-consistency for the $\WRITEACK$ responses received from 
replicas that have not necessarily reached the assignment in line~\ref{alg:line:recoverymyincincrement}.
This result is essential to establish progress for the enhanced $\algname$ where a replica may restart without 
knowing its prior incarnation number.
\begin{lemma}
	\label{lem:livelem-cc}
	Consider a $\wrint$ invocation $\wrq$ in an execution $\alpha$
	of enhanced $\algname$ such that $\reqq(\wrq)\in \{\precv(\_, \_), \cv(\_, \_)\}$. 
	Suppose that $\proc(\wrq)$ sends infinitely many $\WRITE$ messages $m$ with $\id(m)=\id(\wrq)$
	to a replica $p_j\neq \proc(\wrq)$ in $\alpha$, and assume that $p_j$ is eventually up.
	Then there exists a time $t$ such that $\CC(\wrq, j, t')$ holds
	for all times $t' \ge t$.
	%
	%
\end{lemma}

\begin{proof}
	Since $p_j$ is eventually up in $\alpha$, there exists the time $T_j>0$ at which 
	$p_j$ restarts for the last time in $\alpha$.
	Since $\proc(\wrq)$ sends infinitely many $\WRITE$ 
	messages $m$ with $\id(m)=\id(\wrq)$ to $p_j$, $\proc(\wrq)$ is permanently 
	up from the time it sends first such message onwards.
	Let $T\ge T_j$ be the earliest time in $\alpha$ such that: 
	\emph{(i)} both $\proc(\wrq)$ and $p_j$ are permanently up from $T$ onwards, 
	and $\emph{(ii)}$ for all times from $T$ onwards, the value of $\State.\cvvar[j]_j$ does not 
	change. This time $T$ exists due to the fact that $\proc(\wrq)$ and $p_j$ are eventually up in 
	$\alpha$ and Lemma~\ref{lem:livelem-enh}. 

	By~\emph{(i)} and the channel reliability assumptions, after $T$,
	every message sent by either $\proc(\wrq)$ to $p_j$ or by $p_j$ to 
	$\proc(\wrq)$ is eventually received by its destination.
	Since $\proc(\wrq)$ sends infinitely 
	many $\WRITE$ messages with $\id(\wrq)$ to $p_j$, this implies that  
	there exists a time $s\ge T$ and a $\WRITE$  message $m$ with $\id(m)=\id(\wrq)$
	such that $p_j$ receives $m$ from $\proc(\wrq)$ at $s$.
	By the code of the $\WRITE$ handler, upon receiving a $\WRITE$ message carrying either 
	$\precv$ or $\cv$ request, a replica $p_j$ always replies with a $\WRITEACK$ message with the same identifier
	and the incarnation number equal to $\State.\cvvar[j]$. Thus,
	there exists a time 
	$$
	t \ge s \ge T
	$$
	at which $\proc(\wrq)$ receives a
	$\WRITEACK$ message $\hat{m}$ with $\id(\hat{m})=\id(m)=\id(\wrq)$ from $p_j$ such that 
	$$
	\hat{m} \in \wacks(\wrq, t) \wedge \inc(\hat{m}) = \State.\cvvar[j]_j(s).
	$$
	We show that 
	\begin{equation}
	\label{eq:lem-no-change}
	\forall t'\ge t.\, \forall m'\in \racks(\wrq, t').\, \inc(\hat{m}) \ge \crv(m')[j].
	\end{equation}
	Assume the contrary and let $u \ge t$ be the earliest time such that 
	$\inc(\hat{m}) < \crv(m')[j]$ for some $m'\in \racks(\wrq, u)$. 
	Then $m'$ was received by a $\rdint(\cv)$ operation that completes at $u$
	and returns a set of crash vectors $S$. Since $m'$ was used to compute
	$S$, $\crv(m')\in S$. Since $\proc(\wrq)$ sends infinitely
	many $\WRITE$ messages to $p_j$, $\proc(\wrq)$
	must reach line~\ref{alg:line:sendwrite} after $\rdq$ returns.
	Let $m''=\WRITE(\id(\wrq), \_, \tmpcv[j])$ be a $\WRITE$ message
	sent to $p_j$ at line~\ref{alg:line:sendwrite}.
	By the code in lines~\ref{alg:line:cvfor}--\ref{alg:line:rdcv},
	$\tmpcv[j] \ge \crv(m')[j]$. Since $u \ge t$,
	there exists a time 
	$$
		s' > u \ge t
	$$
	at which $m''$ is received by $p_j$.
	By the code handling $\WRITE$ messages with $\precv$ or $\cv$ requests,
	$$
		\State.\cvvar[j]_j(s') \ge  \crv(m')[j] > \inc(\hat{m}) = \State.\cvvar[j]_j(s).
	$$
	Since $s' > u \ge t \ge s \ge T$, this is a contradiction to~\emph{(ii)}. Thus,~(\ref{eq:lem-no-change}) holds,
	which given that $\reqq(\wrq) \in \{\precv(\_, \_), \cv(\_, \_)\}$, implies
	that $\CC(\wrq, j, t')$ holds for all $t' \ge t$, as required.
\end{proof}

\paragraph*{Liveness in good-case executions.}
Given an execution $\alpha$ of the $\algname$ (basic or enhanced), we let
\begin{equation}
\hat{M}_\alpha = \begin{cases}
  6,  & \text{if } \alpha \text{ is an execution of the basic \algname}, \\
  12, & \text{otherwise.}
\end{cases}
\label{eq:hat-m}
\end{equation}
Intuitively, $\hat{M}_\alpha$ represents the \emph{minimum} number of message delays required for a correct replica
to complete the $\onrestart$ procedure in the corresponding version of the $\algname$ algorithm.
In the following, we will often omit the subscript from $\hat{M}$ if it is clear from the context.
We prove that in every good-case execution where the replica failures are separated by 
at least $\hat{M}\Delta$, a restarting replica is always able
to complete its $\onrestart$ block 
before any other replica may fail.
To establish this result, we will first prove the following lemma:
\begin{lemma}
Let $\alpha$ be an execution of $\algname$ (basic or enhanced) and suppose that
$n > 2d+1$ where $d$ is the maximal number of replicas that 
are not eventually up in $\alpha$.
Suppose that any two replica crashes occurring 
in $\alpha$ are more than $\hat{M}\Delta$ apart, and at all times,
there is at most one replica, which is up, but not active. If a replica $p_i$
crashes and restarts at time $t$ in $\alpha$, then the following holds:
\begin{enumerate}
	\item[\normalfont{(i)}] \label{lem:m-dyn-rd} every $\rdint$ invoked by $p_i$ at $t'\in [t, t+\hat{M}\Delta - 2\Delta]$
	completes by $t'+2\Delta$; and

	\item[\normalfont{(ii)}] \label{lem:m-dyn-wr} every $\wrint(\cv(\_, \_))$ or $\wrint(\precv(\_, \_))$ 
	invoked by $p_i$ at $t' \in [t, t+\hat{M}\Delta - 4\Delta]$
	completes by $t'+4\Delta$.
\end{enumerate}
\label{lem:m-dyn}
\end{lemma}

\begin{proof}
Let $D(t)$ denote the number
of replicas that have crashed without restarting by time $t$ in $\alpha$.
We first show (i). By the lemma's premise, $p_i$ is the only replica which is up but
not active at $t$. Since no replica fails until $t+\hat{M}\Delta$, all replicas in the set 
$S_r = \allreps \setminus (D(t) \cup \{p_i\})$ remain active within 
the interval $[t, t+\hat{M}\Delta]$. Since 
$n \ge 2d+2$ and $|D(t)|\le d$, we have $|S_r|\ge 2d + 2 - d - 1 = d+1$. 
Thus, $S_r$ contains a read quorum, which remains available for the entire duration of
$[t, t+\hat{M}\Delta]$. By the code, every $\rdint$ invoked at $t'$ completes within 
$2\Delta$ provided a read quorum is available within $[t', t'+2\Delta]$. Therefore,
$t' \in [t, t+\hat{M}\Delta-2\Delta]$ and $\hat{M}>2$ imply that any
$\rdint$ invoked at $t'$ has a read quorum available for the duration 
of its execution, and therefore, will terminate by $t'+2\Delta$, as needed.

To prove (ii),
consider an invocation $\wrq $ of $\wrint(\cv(\_, \_))$ or $\wrint(\precv(\_, \_))$ by $p_i$ at 
$t' \in [t, t + \hat{M}\Delta - 4\Delta]$.
%
Then at $t'$, $p_i$ sends a $\WRITE$ message $m$
with $\id(m)=\id(\wrq)$ to all replicas. Let 
$S_w=\allreps \setminus D(t)$ be the set of all replicas that are up at $t$. 
Since no replicas crash in $[t, t + \hat{M}\Delta]$ and $[t', t' + \Delta] \subseteq [t, t + 
\hat{M}\Delta]$, by the channel reliability assumptions, every 
replica in $S_w$ receives $m$ by $t' + \Delta$. Additionally, since the 
$\reqq(m)$ is either $\cv$ or $\precv$, every replica 
$p_j \in S_w$ replies with a $\WRITEACK$ message $m'$ with $\id(m')=\id(m)=\id(\wrq)$
to $p_i$ upon receiving $m$. Since all replicas in 
$S_w$ are up in $[t', t' + 2\Delta]$, there exists a time $s \le t' + 2\Delta$ at which 
$p_i$ collects $\WRITEACK$ responses with the identifiers equal to $\id(\wrq)$
from a set $S_w' \subseteq S_w$ such that $|S_w'| \ge n-d$. It then invokes $\rdint$, which by \emph{(i)}, 
terminates at 
$$
s' \le s + 2\Delta \le t' + 4\Delta.
$$
We now argue that for all $p_j\in S_w'$, $\CC(\wrq, j, s')$. 
There are two cases to consider: If $p_j = p_i$, then $\CC(\wrq, j, s')$ trivially holds. 
Otherwise, let $T_j$ be the last time $p_j$ restarted before $s'$. Since $p_j$ is active at $t$, and 
no replica crashes until $t+\hat{M}\Delta$, $T_j < t$ and $p_j$ remains active for the whole 
duration of $[T_j, t+\hat{M}\Delta]$. Thus, $p_j$ executes the assignment in 
line~\ref{alg:line:recoverymyincincrement}
at $T_j < t$ and remains active until $s' \ge T_j$. Furthermore, by the argument 
above, there exist a $\WRITEACK$ message $m' \in \wacks(\wrq, s')$
with $\proc(m')=p_j$ and $\id(m')=\id(\wrq)$ that was sent by $p_j$ at time $s'' \in [t', s']$. 
Since $t' \ge t \ge T_j$, $s'' \in [T_j, s']$. Hence, by Lemma~\ref{lem:livelem-cc-simple}, 
$\CC(\wrq, j, s')$ holds. 

Since $p_i$ reaches line~\ref{alg:line:filter-x} by time $s'$ and
$\CC(\wrq, j, s')$ holds for all $p_j \in S_w'$, 
by Proposition~\ref{prop:cc}, we have that upon completion of 
line~\ref{alg:line:filter-x}, for all $p_j\in S_w'$,  $p_j\in Q$. Hence, at $s'$,
$S_w'\subseteq Q$, and therefore, $|Q|\ge n-d=q_w$, which implies that $p_i$ exits the loop
in lines~\ref{alg:line:while}--\ref{alg:line:wrint-ret-a2plus}.
Hence, $\wrq$ returns at $s'\le t'+4\Delta$, as required.
\end{proof}
By the structure of the code, each execution of $\onrestart$ consists of the following
sequences of $\rdint$ and $\wrint$ calls interleaved with local steps: 
$\wrint$ followed by $\rdint$ for the basic $\algname$, and
$\rdint$, $\wrint$, $\wrint$, and $\rdint$ for the enhanced one. 
Hence, from Lemma~\ref{lem:m-dyn} we can conclude:
\begin{corollary}
Let $\alpha$ be an execution of $\algname$ (basic or enhanced) and suppose that
$n > 2d+1$ where $d$ is the maximal number of replicas that 
are not eventually up in $\alpha$.
Suppose that any two replica crashes occurring 
in $\alpha$ are more than $\hat{M}\Delta$ apart, and at all times,
there is at most one replica, which is up, but not active. If a replica $p_i$
crashes and restarts at time $t$ in $\alpha$, then $p_i$ completes its $\onrestart$ block
by $t+\hat{M}\Delta$.
\label{cor:m-dyn}
\end{corollary}

\paragraph*{Liveness of \rdint\/ in $\Fd{\ndown,\ncrash,\hat{M}}$-valid executions.}
We first show that in every $\Fd{\ndown,\ncrash,\hat{M}}$-valid execution of $\algname$
there are always sufficiently many active replicas to ensure the existence
of an available read quorum. To establish this result, we first prove that at each point 
in time in every good-case execution of $\algname$, there is at most one replica that is
up but not active.  
\begin{lemma}
	Let $\alpha$ be an execution of $\algname$ (basic or enhanced) and suppose that
	$n > 2d+1$ where $d$ is the maximal number of replicas that 
	are not eventually up in $\alpha$.
	If any two replica crashes occurring 
	in $\alpha$ are more than $\hat{M}\Delta$ apart, then at all times, 
	there is at most one replica which is up but not active.
	\label{lem:dyn-act}
\end{lemma}
\begin{proof}
	Assume by contradiction that the result is false. Let $\alpha'$ be the minimal
	prefix of $\alpha$ in which there are exactly $2$ replicas that are up, but not active.
	Since the number of active replicas can only decrease when some replica crashes
	and restarts, the last event $e$ in $\alpha'$ is a crash and restart of some replica $p_i$. Let
	$t$ be the time associated with $e$. 

	Since there are exactly two replicas which are 
	up, but not active in $\alpha'$, there exist a time $t' \le t$ and a replica $p_j\neq p_i$
	such that it holds: \emph{(i)} $p_j$ crashed and restarted at $t'$, 
	\emph{(ii)} $p_j$ is up in $[t', t]$, and \emph{(iii)} $p_j$ has not completed
	the execution of $\onrestart$ initiated at $t'$. By the lemma's premise, 
	$0 \le t' < t - \hat{M}\Delta$. Therefore, $p_j$ does not 
	complete its $\onrestart$ block by $t' + \hat{M}\Delta$. However,
	since $p_j$ is the only replica that is up, but not active before $t$, this is
	a contradiction to Corollary~\ref{cor:m-dyn}.
\end{proof}
We now generalize the  above result  for all $\Fd{\ndown,\ncrash,\hat{M}}$-valid
executions of $\algname$ in which the number of replicas satisfies $n \ge \ndown + \ncrash + 1$.
\begin{lemma}
	Let $\alpha$ be an execution of $\algname$ (basic or enhanced) and $\ups$ be 
	the set of replicas, which are eventually up in $\alpha$. Then for all
	$\ndown$ and $\ncrash$, if $n \ge \ndown + \ncrash + 1$, and 
	$\Fd{\ndown,\ncrash,\hat{M}}(\alpha)$ holds, then at all times, 
	there are at least $n - \ncrash$ replicas in $\ups$, which are active.
	\label{lem:active-black}
\end{lemma}

\begin{proof}
	Since $\Fd{\ndown,\ncrash,\hat{M}}(\alpha)$ holds, by Definition~\ref{ass:what}, we have 
	$|U|\ge n-d$, and either the condition stated in part 1 or the one stated in part 2 of 
	Definition~\ref{ass:what} is satisfied in $\alpha$. If the latter holds, 
	then there exists $S \subseteq \ups$ such that $|S| \ge n - \ncrash$ and all replicas in 
	$S$ never crash in $\alpha$. Hence, all replicas in $S$ are permanently active
	in $\alpha$, and the lemma follows.

	Suppose that $\alpha$ satisfies the condition given in part 1 of Definition~\ref{ass:what}.
	Then, $\ncrash > \ndown$, and any two crash events occurring 
	in $\alpha$ are separated by more than $\hat{M}\Delta$. 
	Let $U_a(t)\subseteq U$ be the set of active replicas in $U$ at time $t$.
	Since $c > d$, we have $n \ge d + c + 1 > 2d+1$, and therefore, by Lemma~\ref{lem:dyn-act},
	there is at most $1$ replica in $U$ which is not active at any given time. Hence,
	$|U_a(t)| \ge n-d-1$, which, given that $c > d$, implies that
	$|U_a(t)| \ge n-c$, as needed. 
%
\end{proof}
We now show that every invocation $\rdq = \rdint(\_)$ is guaranteed 
to eventually complete in every $\Fd{\ndown,\ncrash,\hat{M}}$-valid 
execution of $\algname$, provided $n \ge \ndown + \ncrash + 1$ and 
the process invoking $\rdq$ (either a replica or a client) remains 
up from that point onward.
\begin{lemma}
	\label{lem:rdlive-black}
	Suppose that for all
	$\ndown$ and $\ncrash$, $n \ge \ndown + \ncrash + 1$, and let $\alpha$
	be an execution of $\algname$ (basic or enhanced) such that 
	$\Fd{\ndown,\ncrash,\hat{M}}(\alpha)$ holds. 
	If a process (replica or client) $p_i$ calls $\rdq=\rdint$ at time $t$, and 
	remains up from $t$ onwards, then $\rdq$ eventually returns.
\end{lemma}

\begin{proof}
	Let $U$ be the set of replicas that are eventually up in $\alpha$, and $T$
	be the time such that all replicas in $U$ are up from $T$ onwards.
	By Lemma~\ref{lem:active-black}, there exists a set $U_a \subseteq U$ of
	replicas such that $|U_a|\ge n-c$, and all replicas in $U_a$ 
	are active at time $T$. Since no replica in 
	$U$ crashes after time $T$, all replicas in $U_a$ are active from 
	$T$ onward.

	Consider an invocation $\rdq=\rdint(\_)$ in $\alpha$ and suppose that $\rdq$ 
	is still in progress at some time $t_0 \ge T$. We show that there exists a time
	$t \ge t_0$ such that $\rdq$ is guaranteed to return by $t$.
	By the code, $p_i$ 
	continues re-transmitting $\READ(\id(\rdq), \_, \_)$ messages to all replicas until 
	receiving $q_r=d+1$ $\READACK$ replies with the same identifier $\id(\rdq)$. Therefore, 
	there exists a time $t_1 \ge t_0 \ge T$ such that either $\rdq$ terminates before $t_1$ or
	$p_i$ initiates a new re-transmission round at $t_1$. Since $p_i$ and 
	all processes in $U_a$ are up from $t_1$ onwards, the channel
	reliability assumptions imply that every replica $p_j\in U_a$ eventually
	receives a message $m=\READ(\id(\rdq), \_, \_)$ sent to it by $p_i$ at $t_1$. Since $p_j$
	is active from $T\le t_1$ onwards (i.e., $\stale=\FALSE$ in line~\ref{alg:line:notstale-a2+}), 
	it responds to $m$ with $m'=\READACK(\id(\rdq), \_, j)$.
	Since both $p_j$ and $p_i$ are up from $T$ onwards, $p_i$ eventually receives $m'$.
	Hence, there exists a time $t_2\ge t_1 \ge t_0 \ge T$ such that if $\rdq$ does not
	return before $t_2$, then at $t_2$, $p_i$ receives $\READACK(\id(\rdq), \_, \_)$
	from all replicas $p_j\in U_a$ such that $|U_a|\ge n-c$. 
	Since $n \ge d + c + 1$, we have $n-c \ge d+1=q_r$, and therefore, the termination 
	condition of the retransmission loop (line~\ref{alg:line:rcvreadack-a2+}) validates at $t_2$.
	Hence, $\rdq$ is guaranteed to complete no later than at $t=t_2 \ge t_0$, as required.
\end{proof}

\paragraph*{Liveness of $\wrint$ in $\Fd{\ndown,\ncrash,\hat{M}}$-valid executions.}
We first show that in every $\Fd{\ndown,\ncrash,\hat{M}}$-valid execution 
of $\algname$, a process that invokes $\wrint$ and remains up from that 
point onwards is guaranteed to exit the inner loop 
(lines~\ref{alg:line:while2}--\ref{alg:line:sendwrite}) 
at each iteration of the outer loop 
(lines~\ref{alg:line:while}--\ref{alg:line:wrint-ret-a2plus}).
\begin{lemma}
	\label{lem:inner-wrint-live}
	Suppose that for all
	$\ndown$ and $\ncrash$, $n \ge \ndown + \ncrash + 1$. Let $\alpha$
	be an execution of $\algname$ (basic or enhanced) 
	such that $\Fd{\ndown,\ncrash,\hat{M}}(\alpha)$ holds, 
	and $U_a$ be the set of replicas that are eventually active in $\alpha$.
	Consider a process $p_i$ (client or replica) such that $p_i$ is permanently 
	up after some time $t$ in $\alpha$. Then, if $p_i$ 
	invokes $\wrq=\wrint(\req)$ after $t$, and 
	$(\req \in \{\tsval(\_,\_)\} \implies |U_a|\ge n-d)$,
	then $\wrq$ exits the inner loop in lines~\ref{alg:line:while2}--\ref{alg:line:sendwrite} 
	at every iteration of the outer loop in lines~\ref{alg:line:while}--\ref{alg:line:wrint-ret-a2plus}.
\end{lemma}
\begin{proof}
	Let $U$ denote the set of replicas that are eventually up in $\alpha$. Since
	$\Fd{\ndown,\ncrash,\hat{M}}(\alpha)$ holds, $|U|\ge n-d$. Let
	$$
	\hat{U} = \begin{cases}
		U_a, & \text{if } \req\in \{\tsval(\_, \_)\}, \\
  		U,  & \text{otherwise.} 
	\end{cases}
	$$
	Then, $|\hat{U}|\ge n-d$. Let $T$ (respectively, $T_a$) be the time such that 
	all replicas in $U$ (respectively, $U_a$) are up (respectively, active) after $T$ (respectively, $T_a$).
	Let
	$$
	\hat{T} = \begin{cases}
		T_a, & \text{if } \req\in \{\tsval(\_, \_)\}, \\
  		T,  & \text{otherwise.} 
	\end{cases}
	$$
	Let $t_0 \ge \hat{T}$ and suppose that $\wrq$ is executing inside the inner loop 
	(lines~\ref{alg:line:while2}--\ref{alg:line:sendwrite}) at $t_0$. We prove that
	there exists a time $t \ge t_0$ by which $\wrq$ is guaranteed to exit the inner loop. 
	
	By the code, $p_i$ 
	continues re-transmitting $\WRITE(\id(\wrq), \id(\wrq), \_)$ messages to all replicas until 
	receiving $q_w=n-d$ $\WRITEACK$ replies with the same identifier $\id(\wrq)$. Therefore, 
	there exists a time $t_1 \ge t_0 \ge \hat{T}$ such that  $p_i$ either exits the 
	inner loop before $t_1$ or initiates a new re-transmission round at $t_1$. Since $p_i$ and 
	all replicas in $\hat{U}$ are permanently up from $t_1$ onwards, the channel
	reliability assumptions imply that every replica $p_j\in \hat{U}$ eventually
	receives a message $m=\WRITE(\id(\wrq), \reqq(\wrq), \_)$ sent to it by $p_i$ at $t_1$. 
	If $\req \in \{\precv(\_, \_), \cv(\_, \_)\}$, then $p_j\in U$, and therefore,
	the condition in line~\ref{line:write-2} holds. Otherwise, $p_j\in U_a$ so that
	the condition in line~\ref{alg:line:notstale-a2+} is validated. 
	Therefore, once $p_j$ 
	receives $m$, it responds with $m'=\WRITEACK(\id(\wrq), \_, \_, j)$.
	Since both $p_j$ and $p_i$ are up from $\hat{T}$ onwards, $p_i$ eventually receives $m'$.
	Hence, there exists a time $t_2\ge t_1 \ge t_0 \ge \hat{T}$ such that if 
	the loop in lines~\ref{alg:line:while2}--\ref{alg:line:sendwrite}
	is still in progress at $t_2$, then at $t_2$, $p_i$ has received $\WRITEACK(\id(\wrq), \_, \_, j)$
	from all replicas $p_j\in \hat{U}$. Since $|\hat{U}|\ge n-d=q_w$, 
	the inner loop termination condition in line~\ref{alg:line:wrint-wait} is met at $t_2$. 
	Hence, $\wrq$ is guaranteed to exit the inner loop at $t=t_2\ge t_0$ at the latest, as required.
\end{proof}
We next show that if a process invoking $\wrint$ never exits the outer loop
(lines~\ref{alg:line:while}--\ref{alg:line:wrint-ret-a2plus}), then 
the set of replicas $Q$ maintained inside the loop will eventually include 
every replica that is eventually up (or active).
\begin{lemma}
	Consider an execution $\alpha$ of $\algname$ (basic and enhanced) and
	suppose that a process $p_i$ (client or replica) invokes $\wrq=\wrint(\req)$.
	Let $\tau = \tau_1 < \tau_2 < \dots$ be an infinite sequence of times such that
	for all $k \ge 1$, $p_i$ reaches line~\ref{alg:line:filter-x} at $\tau_k$ while
	executing $\wrq$. Then for all replicas $p_j$, if
	$p_j$ is eventually up and $(\req \in \{\tsval(\_, \_)\} \implies p_j \text{~is~eventually~active})$,
	then there exists $l$ such that for all $k\ge l$, $p_j\in Q$ at $\tau_k$.
	\label{lem:catch-all}
\end{lemma}

\begin{proof}
	If $p_j=p_i$, then the result trivially holds. Consider $p_j \neq p_i$
	and assume by contradiction that there exists an infinite sequence
	of times $\sigma = \sigma_1 < \sigma_2 < \dots$ such that for all $k \ge 1$, $\wrq$ reaches 
	line~\ref{alg:line:filter-x} at $\sigma_k$ and $p_j\not\in Q$ at $\sigma_k$. Since $p_i$ reaches 
	line~\ref{alg:line:filter-x} infinitely many times in $\alpha$, it never crashes after invoking $\wrq$.
	Hence, for all $k\ge 1$, $p_i$ sends $\WRITE(\id(\wrq), \req, \_)$ to $p_j$ 
	in-between $\sigma_k$ and $\sigma_{k+1}$. As a result, starting from $\sigma_1$ onwards, $p_i$ 
	sends infinitely many  
	$\WRITE(\id(\wrq), \req, \_)$ messages to $p_j$. 
	Let $W$ denote the set of these messages. By the lemma's premise,
	$p_j$ is eventually up. We make a case split based on 
	whether $p_j$ ever crashes, and whether 
	$p_j$ is eventually active.	
	\begin{itemize}
		\item $p_j$ never crashes in $\alpha$.
			Let $m\in W$.
			By the channel reliability assumptions, after $\sigma_1$, every message sent by $p_i$
			to $p_j$ and by $p_j$ to $p_i$ is eventually received by its destination.
			Furthermore, since $p_j$ never crashes, it is active at all times in $\alpha$.
			Therefore, upon receiving $m$, $p_j$ replies to $p_i$ with a 
			$\WRITEACK$ message $m'$ satisfying $\id(m') = \id(\wrq)$, 
			regardless of the payload of $m$.
			Hence, there exists $l$ such that for all $k \ge l$, $m'\in \wacks(\wrq, \sigma_k)$.
			Therefore, by Lemma~\ref{lem:livelem-cc-init}, $\CC(\wrq, j, \sigma_k)$ holds for all 
			$k \ge l$. Hence, by Proposition~\ref{prop:cc}, 
			$p_j \in Q$ at all times $\sigma_k$ such that $k \ge l$ contradicting the definition of $\sigma$.

		\item $p_j$ crashes and restarts at least once in $\alpha$, and is eventually active.
			%
			Hence, there exists a time $T_j$ such that $p_j$ 
			executes the assignment in line~\ref{alg:line:recoverymyincincrement} at $T_j$ and
			does not crash afterwards, and a time $T_j' > T_j$ such that $p_j$ becomes
			active at $T_j'$. Therefore, there exists $l$ and $m\in W$ such that $p_i$ sends
			$m$ after time $\sigma_l \ge T_j'$. 
			By the channel reliability assumptions, after $\sigma_l$, every message sent by $p_i$
			to $p_j$ and by $p_j$ to $p_i$ is eventually received by its destination.
			Furthermore, since $p_j$ is active after $\sigma_l$, upon receiving $m$, it replies to $p_i$ with a 
			$\WRITEACK$ message $m'$ satisfying $\id(m') = \id(\wrq)$, 
			regardless of the payload of $m$.			
			Hence, there exists $l' \ge l$ such that for all $k \ge l'$, $m'\in \wacks(\wrq, \sigma_k)$.
			Since $m'$ is sent by $p_j$ after $T_j' > T_j$ and $p_j$ is up from $T_j$ onwards, 
			by Lemma~\ref{lem:livelem-cc-simple}, $\CC(\wrq, j, \sigma_k)$ holds for all 
			$k \ge l'$. Hence, by Proposition~\ref{prop:cc}, 
			$p_j \in Q$ at all times $\sigma_k$ where $k \ge l'$ contradicting the definition of $\sigma$.



		\item $p_j$ crashes and restarts at least once, is not eventually active, and reaches 
			line~\ref{alg:line:recoverymyincincrement} after its last restart. 
			Hence, there exists a time $T_j$ such that $p_j$ executes the assignment in 
			line~\ref{alg:line:recoverymyincincrement} at $T_j$ and remains permanently up afterwards.
			Therefore, there exists $l$ and $m\in W$ such that $p_i$ sends
			$m$ after time $\sigma_l \ge T_j$. 
			By the channel reliability assumptions, after $\sigma_l$, every message sent by $p_i$
			to $p_j$ and by $p_j$ to $p_i$ is eventually received by its destination.
			Since $p_j$ is not eventually active, $\req \in \{\precv(\_, \_), \cv(\_, \_)\}$.
			Thus, given that $p_j$ is up after $\sigma_l$, it replies to $p_i$ with a 
			$\WRITEACK$ message $m'$ satisfying $\id(m') = \id(\wrq)$ upon receiving $m$.
			Hence, there exists $l' \ge l$ such that for all $k \ge l'$, $m'\in \wacks(\wrq, \sigma_k)$.
			Since $m'$ is sent by $p_j$ after $T_j$ and $p_j$ is up from $T_j$ onwards, 
			by Lemma~\ref{lem:livelem-cc-simple}, $\CC(\wrq, j, \sigma_k)$ holds for all 
			$k \ge l'$. Hence, by Proposition~\ref{prop:cc}, 
			$p_j \in Q$ at all times $\sigma_k$ where $k \ge l'$, contradicting the definition of $\sigma$.

		\item $p_j$ crashes and restarts at least once, is not eventually active, and never reaches 
		line~\ref{alg:line:recoverymyincincrement} after its last restart. 
		Then by the lemma's 
		premise $\req\in \{\precv(\_, \_), \cv(\_, \_)\}$. Furthermore, since in the basic
		algorithm, $p_j$ always reaches  line~\ref{alg:line:recoverymyincincrement}
		upon restart, $\alpha$ must be an execution of the enhanced $\algname$.
		Since $p_i$ sends infinitely many $\WRITE(\id(\wrq), \req, \_)$ messages to $p_j$, 
		by Lemma~\ref{lem:livelem-cc}, there exists $l$
		such that $\CC(\wrq, j, \sigma_k)$ holds for all $k \ge l$. 
		Then by Proposition~\ref{prop:cc}, $p_j\in Q$ at all times 
		$\sigma_k$ such that $k \ge l$, which is a contradiction to the definition
		of $\sigma$. 
	\end{itemize}
\end{proof}
%
We now use the above lemma to derive sufficient conditions for 
the termination of $\wrint$ invocations in $\Fd{\ndown,\ncrash,\hat{M}}$-valid executions
where $n \ge \ndown + \ncrash + 1$.
\begin{lemma}
	\label{lem:wrint-live}
	Suppose that for all
	$\ndown$ and $\ncrash$, $n \ge \ndown + \ncrash + 1$. Let $\alpha$
	be an execution of $\algname$ (basic or enhanced) 
	such that $\Fd{\ndown,\ncrash,\hat{M}}(\alpha)$ holds, 
	and $U_a$ be the set of replicas that are eventually active in $\alpha$.
	Consider a process $p_i$ (client or replica) such that $p_i$ is permanently 
	up after some time $t$ in $\alpha$. Then, if $p_i$ 
	invokes $\wrq=\wrint(\req)$ after $t$, and 
	$(\req \in \{\tsval(\_,\_)\} \implies |U_a|\ge n-d)$,
	then $\wrq$ eventually returns.
\end{lemma}

\begin{proof}
	Assume by contradiction that $\wrq$ never completes.
	By Lemma~\ref{lem:rdlive-black}, $\wrq$ never blocks 
	in the $\rdint$ calls, and by Lemma~\ref{lem:inner-wrint-live}, 
	the inner loop in lines~\ref{alg:line:while2}--\ref{alg:line:sendwrite} 
	terminates at every iteration of the outer 
	loop in lines~\ref{alg:line:while}--\ref{alg:line:wrint-ret-a2plus}.
	Therefore, the outer loop (lines~\ref{alg:line:while}--\ref{alg:line:wrint-ret-a2plus})
	never terminates, and $p_i$ reaches line~\ref{alg:line:filter-x} at every 
	iteration of the outer loop. Hence, 
	there exists an infinite sequence of times 
	$\tau=\tau_1 < \tau_2 < \dots$ such that for all $k\ge 1$, $p_i$ reaches 
	line~\ref{alg:line:filter-x} at $\tau_k$. 
	Let $U$ be the set of replicas which are eventually up in $\alpha$.
	Then since $\Fd{\ndown,\ncrash,\hat{M}}(\alpha)$ holds, $|U|\ge n-d$.
	Let
	$$
	\hat{U} = \begin{cases}
		U_a, & \text{if } \req\in \{\tsval(\_, \_)\}, \\
  		U,  & \text{otherwise.} 
	\end{cases}
	$$
	Then, $|\hat{U}|\ge n-d$.
	By Lemma~\ref{lem:catch-all},
	there exists $l$ such that for all $k \ge l$, $\hat{U} \subseteq Q$ at $\tau_k$.
	Since $p_i$ never exits the outer loop, $|Q|<q_w$ whenever $p_i$
	reaches line~\ref{alg:line:filter-x}. However, 
	$|\hat{U}|\ge n-d=q_w$, and therefore, for all $k\ge l$,
	$|Q|\ge q_w$ at $\tau_l$. A contradiction.
\end{proof}

The code of the $\onrestart$ block consists of calls to $\rdint$ and $\wrint(\req)$
where $\req \in \{\precv(\_, \_), \cv(\_, \_)\}$, interleaved with local steps.
Therefore, from Lemmas~\ref{lem:rdlive-black} and~\ref{lem:wrint-live}, we can
conclude that every replica that crashes and restarts, and 
does not crash again in a $\Fd{\ndown,\ncrash,\hat{M}}$-valid execution
of $\algname$ is eventually able to complete its $\onrestart$ block and become active.
\begin{corollary}
	Suppose that for all
	$\ndown$ and $\ncrash$, $n \ge \ndown + \ncrash + 1$, and let $\alpha$
	be an execution of $\algname$ (basic or enhanced) such that 
	$\Fd{\ndown,\ncrash,\hat{M}}(\alpha)$ holds. 
	Then for all replicas $p_i$, if $p_i$ is eventually up in $\alpha$, then $p_i$
	is eventually active. 
\label{cor:restart-active}
\end{corollary}
Since there are $n-d$ replicas that are eventually up in every $\Fd{\ndown,\ncrash,\hat{M}}$-valid execution $\alpha$
of $\algname$, Corollary~\ref{cor:restart-active} implies that a write quorum consisting entirely of active
replicas is also eventually available in $\alpha$.
\begin{corollary}
	Suppose that for all
	$\ndown$ and $\ncrash$, $n \ge \ndown + \ncrash + 1$, and let $\alpha$
	be an execution of $\algname$ (basic or enhanced) such that 
	$\Fd{\ndown,\ncrash,\hat{M}}(\alpha)$ holds. Then, there exists a set $U_a$ of replicas
	such that $|U_a|\ge n-d$ and all replicas in $U_a$ are eventually active in $\alpha$.
	\label{livecor:rec-black}
\end{corollary}

\paragraph*{Liveness of $\algname$.}
The implementations of the register $\rd$ and $\writ$ operations consist
of the calls to $\rdint$ and $\wrint(\req)$ such that $\req\in \{\tsval(\_, \_)\}$.
By Corollary~\ref{livecor:rec-black}, in every $\Fd{\ndown,\ncrash,\hat{M}}$-valid
execution of $\algname$, there is a set $U_a$ of active replicas such that 
$|U_a|\ge n-d$. Therefore, from Lemmas~\ref{lem:rdlive-black} and~\ref{lem:wrint-live},
we can conclude that every $\rd$ and $\writ$ invoked by a correct client, eventually completes
in every $\Fd{\ndown,\ncrash,\hat{M}}$-valid execution of $\algname$
\begin{corollary}
	For all $\ndown$, $\ncrash$, $\alpha$, if $n \ge \ndown + \ncrash 
	+ 1$, then the $\algname$ algorithm (basic or enhanced) implements a
	wait-free register under $\Fd{\ndown, \ncrash, \hat{M}}$.
\label{lem:black-live}
\end{corollary}
Using the corollary above we can now prove Theorems~\ref{thm:upper-dynamic} 
and~\ref{lem:dynamic-upper-ass}.
\begin{proof}[Proof of Theorem~\ref{thm:upper-dynamic}]
	By Lemma~\ref{lem:alg-safe}, the enhanced $\algname$ 
	(i.e., the algorithm in Figure~\ref{fig:pseudo-dyn}, including the highlighted lines)
	implements an atomic MWMR register in all $\Fd{\ndown, \ncrash, M}$-valid executions 
	for arbitrary values of $\ndown$, $\ncrash$, and $M$. 
	From Corollary~\ref{lem:black-live}, the enhanced $\algname$ is also
	wait-free in all $\Fd{\ndown,\ncrash,\hat{M}}$-valid executions provided
	$n \ge d + c + 1$. By~(\ref{eq:hat-m}), 
	this implies that
	the enhanced $\algname$ is also wait-free in all $\Fd{\ndown, \ncrash, M}$-valid
	executions where $n \ge d + c + 1$ and $M \ge 12$. Hence, the enhanced
	$\algname$ is an implementation of an atomic MWMR register,
	which is  $\Fd{\ndown, \ncrash, M}$-live and always safe 
	provided $n \ge d + c + 1$ and $M \ge 12$, as required.
\end{proof}

\begin{proof}[Proof of Theorem~\ref{lem:dynamic-upper-ass}]
	By Theorem~\ref{lem:dynamic-upper-ass}($\mathsf{A}$) proven in~\ref{app:safety-black}, if 
	incarnation numbers assigned on 
	line~\ref{alg:line:recoverymyincincrement} satisfy Property~\ref{lem:monotone} (Incarnation 
	Number Monotonicity), then the basic $\algname$ 
	(i.e., the algorithm in Figure~\ref{fig:pseudo-dyn}, excluding the highlighted lines)
	implements an atomic MWMR register in all $\Fd{\ndown, \ncrash, M}$-valid executions 
	for arbitrary values of $\ndown$, $\ncrash$, and $M$.
	From Corollary~\ref{lem:black-live}, the basic $\algname$ is also
	wait-free in all $\Fd{\ndown,\ncrash,\hat{M}}$-valid executions provided
	$n \ge d + c + 1$. 
	By~(\ref{eq:hat-m}), 
	this implies that
	the basic $\algname$ is also wait-free in all $\Fd{\ndown, \ncrash, M}$-valid
	executions where $n \ge d + c + 1$ and $M \ge 6$.
	Hence, the basic $\algname$  is an 
	implementation of an atomic MWMR register that is always safe and
	$\Fd{\ndown,\ncrash,M}$-live provided $n \ge d + c + 1$ and $M \ge 6$, as required.
\end{proof}

\section{Safety Violation in the Algorithm by Guerraoui et al.~\cite{DBLP:journals/talg/GuerraouiLPP08}}
\label{app:rachid}


	The model in~\cite{DBLP:journals/talg/GuerraouiLPP08} is defined as 
follows~\cite[\S 2]{DBLP:journals/talg/GuerraouiLPP08}: The set of processes $N$, 
$|N| = n$, is static and every process executes a deterministic algorithm 
assigned to it, unless it crashes. The process does not behave maliciously. If it 
crashes, the process simply stops executing any computation, unless it possibly 
recovers, in which case the process executes a recovery procedure which is part 
of the algorithm assigned to it. Note that in this case we assume that the 
process is aware that it had crashed and recovered. Every process has a volatile 
storage and some processes may also have a stable storage. If a process crashes 
and recovers, the content of its volatile storage is lost but not the content of 
its stable storage. The processes with stable storage belong to a set denoted 
$S$, $S \subseteq N$. There are a total number of $0 \le |S| = s \le n$ processes 
with stable storage. Whereas the sets $N$ and $S$ are static for all executions, 
the sets of processes that we will define now are not: They might be different 
for each execution (and are unknown in advance). These sets are defined for an 
infinite execution, that is, the sets can only be evaluated by an external 
observer of the system looking at infinite executions. The sets are defined in 
the same way as in~\cite{aguilera2000failure}. Processes that crash at least 
once, always recover after a crash, and eventually do not crash are eventually-up 
and belong to a set denoted $C$, $|C| = c$. These might crash (and recover) a 
large (but finite) number of times. A process is faulty (and belongs to a set 
denoted $F$, $|F| = f$) if there is a time after which the process crashes and 
never recovers, or crashes infinitely many times. We also consider a set of 
processes that are always-up $U$, $|U| = u$ (in that given execution). 
Considering $n = c + f + u$, a number $c + f$ of processes can crash and $c$ 
eventually recover.

The atomic read/write register protocol~\cite[\S5.1]{DBLP:journals/talg/GuerraouiLPP08} makes use 
of an Amnesia-Masking Storage (AMS) abstraction~\cite[\S4]{DBLP:journals/talg/GuerraouiLPP08} 
	that provides the $\ReadAMS$ and $\WriteAMS$ procedures that are tasked 
with sending the appropriate request message and collecting the response from the 
corresponding quorum. In addition to applying $\WriteAMS$ requests to the local 
state (which is stored on stable storage if available) and responding to request 
messages, processes must also complete the 
recovery procedure upon restart. In processes with stable storage, this only 
consists of reading the state stored on stable storage, and in other processes, 
it consists of setting the $\amnesic$ flag to $\TRUE$ and $(\tmpval, 
\tmpts)$ to $(0, 
\bot)$. Processes without stable storage do not reply to $\ReadAMS$ when 
$\amnesic = \TRUE$, and when they receive a message from $\WriteAMS$, $\amnesic$ 
is set to $\FALSE$ and the process responds to $\ReadAMS$ after applying the 
$\WriteAMS$ request to local state.
AMS must satisfy following 
property~\cite[$\S4.1$]{DBLP:journals/talg/GuerraouiLPP08}:

Consider a set of value-timestamp pairs $(\tmpval_i , \tmpts_i)$, each 
value $\tmpval_i$ being
associated with a timestamp $\tmpts_i$. If $\ReadAMS()$ successfully 
completes and returns
a set $V$ of value-timestamp pairs, then $V$ includes the value-timestamp pair
$(\tmpval_h , \tmpts_h)$, where $\tmpts_h$ is the highest timestamp 
among 
all $\WriteAMS(\tmpval_i , \tmpts_i)$ invocations
that successfully completed before the $\ReadAMS()$ invocation.

However, in the following execution, this property is violated:

Assume $n = 2a + 3b + 1$ where $n$ is the number of replicas and $a > 0$ and $b > 1$ are 
arbitrary parameters. Additionally, assume $u = 0$, $f = a + b$, $s = 2a + 2b + 1$,
and 
$S = \{p_1, p_2, \ldots, p_{2a + 2b + 1}\}$. The AMS abstraction must have $n - s 
+ f < |Q_W| \le
n - f$, $|Q_R| \le s - f,$ and $|Q_R| + |Q_W| > n$, which results in $|Q_W| =
a + 2b + 1$ and $|Q_R| = a + b + 1$ where $|Q_W|$ is the write quorum size
and $|Q_R|$ is the read quorum 
size~\cite[\S4.1]{DBLP:journals/talg/GuerraouiLPP08}.


Now consider the following events in order:
\begin{itemize}
	\item $\WriteAMS(1, 1)$ is invoked at replica $p_1$, and completes after 
	receiving responses from $p_1, p_2, \ldots, p_{a + 2b}, p_{2a + 2b + 2}$.

	\item $\WriteAMS(2, 2)$ is invoked at $p_1$ and finishes after
receiving responses from $p_1, p_2, \ldots, p_{a + 2b}, p_{2a + 2b + 3}$.

	\item Process $p_{2a + 2b + 3}$ crashes and recovers, and since it does not 
	have stable storage, it will have $(\tmpval, \tmpts) = (\bot, 0)$ 
	and $\amnesic = 
	\TRUE$.
	
	\item Process $p_{2a + 2b + 3}$ receives the message from $\WriteAMS(1, 1)$ 
	and changes its state to $(\tmpval, \tmpts) = (1, 1)$ and 
	$\amnesic = \FALSE$.

	\item $\ReadAMS()$ is invoked at $p_{a + 2b + 1}$, and finishes after 
	receiving responses from $\mathcal{P} = p_{a + 2b + 1}, p_{a + 2b + 2}, 
	\ldots, p_{2a + 3b 
	+ 1}$.
\end{itemize}

In this scenario, the only replica in $\mathcal{P}$ that has received a message 
about $\WriteAMS(2, 2)$ is $p_{2a + 2b + 3}$, which now has $(\tmpval, 
\tmpts) = (1, 1)$. 
This means that $\ReadAMS()$ at $p_{a + 2b +1}$ returns a set including $(1, 
1)$, but not $(2, 2)$. This contradicts the property above that stipulates that 
the return value must include $(2, 2)$, which has the highest timestamp among all 
calls to $\WriteAMS()$. As a result, AMS, and, in turn, the register emulation 
algorithm suffer from a safety bug.

	A correct solution for the model described above is to use the $\alg{1}{}$ algorithm logic with 
$q_w = q_r = n - f$ when $s > 2f$ and use algorithm $\alg{2}{}$ logic with $q_w = n - f$ and $q_r 
= u$ 
when $u > f$. As to why these algorithm logics are live, note that at least $n - f$ replicas 
are guaranteed to eventually stay up and $u$ replicas are guaranteed to never crash, so read and 
write replicas are available in both cases. Regarding safety, in the first case, a read quorum and a 
write quorum intersect in at least $n - 2f$ replicas and we assume $s > 2f$, so there must be at 
least one replica with trusted stable storage in the intersection of any two such quorums. This 
results in the Real-Time Property (Property~\ref{lem:rtp}) in the same way as proven 
in~\S\ref{sec:upper}. In the second case, read and write quorums intersect in at least $u - f > 0$ 
replicas, and read quorums only consist of replicas that have not crashed so far, so the Real-Time 
Property (Property~\ref{lem:rtp}) holds for the same reasons argued in~\S\ref{sec:upper}. Note 
that 
it is possible and feasible to decide which algorithm to use since $u$, $f$, and $n$ are known prior 
to the start of any execution.

  
	The problem with the protocol in~\cite{DBLP:journals/talg/GuerraouiLPP08} is that, despite 
not having a recovery protocol that breaks crash-consistency, replicas assume that whenever they 
receive a request that writes a value to the register, the writer must have finished a read recently 
and is thus aware of the latest written value. However, this assumption does not necessarily hold 
since messages can be in flight for a long time due to asynchronicity.


\section{Safety Violation in the Algorithm by Dinis et al.~\cite{DBLP:conf/ndss/DinisD023}}
\label{app:rr}

The RR model is defined as follows~\cite[\S III.A]{DBLP:conf/ndss/DinisD023}:
Nodes in the RR model are TEE instances with external
persistent state. Nodes can crash at any point in their execution
and restart at a later instant. Additionally, nodes can suffer a
rollback failure upon a restart, where their externally stored,
persistent state is valid but stale. After each restart, nodes flag
their state as suspicious when replying to requests, signaling
that they have restarted and thus may have been subject to a
rollback. The integrity properties of TEEs imply that a node
can reliably determine when it has executed its initialization
code and thus restarted. A node stops indicating the suspicious
flag once it has ascertained that its state is fresh, by adopting the most recent state among a 
read quorum.
The integrity of TEEs implies
that, other than possibly holding stale persistent state after
restarting, nodes correctly execute the expected code.

To limit the power of the adversary, two independent bounds are introduced: $M_R$ and $F$; 
$M_R$ is the number of processes that can suffer rollback, and $F$ is the number of 
processes that can become unavailable simultaneously due to network partitions, power outages, or 
reboots~\cite[\S III.C]{DBLP:conf/ndss/DinisD023}.

In order to implement an atomic read/write register in the RR model, the ABD 
algorithm~\cite{DBLP:conf/podc/AttiyaBD90} is used, but with different quorum sizes~\cite[\S 
IV.B]{DBLP:conf/ndss/DinisD023}. Given $M_R$ and $F$,
{
	the total number of replicas 
required is $N 
= \max(M_R, F) + F + 1$
}, and the quorum sizes are defined as follows: 
$W=\max(M_R, F) + 1$ and $R= F + 
\min(s, M_R) + 1$, where $W$ is the write quorum size, $R$ is 
the read quorum size~\cite[\S III.D]{DBLP:conf/ndss/DinisD023}, and $s$ is the number of responses 
received from processes with a suspicious state, i.e., processes that have crashed and restarted and 
suspect a rollback attack~\cite[\S III.A and III.C]{DBLP:conf/ndss/DinisD023}. 
	Processes 
can clear their suspicious flag and bring their state up to date if necessary by contacting a read 
quorum 
and choosing the most recent state as ordered by the timestamp~\cite[\S 
IV.A]{DBLP:conf/ndss/DinisD023}.


However, the proposed solution suffers from a safety 
violation scenario
	similar to the one visualized in 
	Fig~\ref{fig:problem}
	due to not ensuring crash-consistent quorums. 
Consider the RR model with $M_R=F=2$. Then we must have $N = 5$, $W = 3$, and $R = 3 + 
\min(s, 2)$.
The following scenario, depicted in~\autoref{fig:rr}, shows a safety breach in the protocol 
provided for this model:
\begin{itemize}
\item Suppose that the initial value of the register is $0$ and a process $p_1$ initiates a write of 
value $1$.
\item $p_1$ contacts a read quorum to determine the latest value $t$ of the timestamp,
  and then initiates access to a write quorum to store $(t', 1)$ where $t' > t$.

\item At this point, $p_1$ sends $(t', 1)$ to $p_2$, $p_2$ stores $(t', 1)$ on stable
storage, and sends an ack back to $p_1$.

\item Process $p_1$ receives the ack from $p_2$, and then gets delayed.

\item Meanwhile, $p_2$ crashes and recovers with its state rolled back to $(t, 0)$.

\item $p_2$ contacts $\{p_3, p_4, p_5\}$, which is a read quorum 
since none of these three processes have crashed and recovered (i.e., $s=0$). Since these 
processes are unaware of the write in progress, $p_2$ assumes that its state 
$(t, 0)$ is up-to-date and clears the suspicious flag.

\item Process $p_1$ resumes, and
stores $(t', 1)$ at $p_1$ and $p_3$. Since $\{p_1, p_2, p_3\}$ is a write
quorum, $p_1$ completes the write and returns success.

\item At this point, 
no process is suspicious, and therefore, $\{p_2, p_4, p_5\}$ is both a write
and a read quorum. However, all processes in this set are unaware
of a preceding complete write, and will therefore return $0$ if queried by a read 
operation, violating safety.
\end{itemize}

\begin{figure}[H]
	\centering
	
	\begin{tikzpicture}[scale=0.88]
		
		\tikzmath{\pwriter = 0; \pcrash = -1; \prest = -2; \pfour = -3; \pfive = -4;}
		
		\node at (-0.5, \pwriter) {$p_1$};
		\node at (-0.5, \pcrash) {$p_2$};
		\node at (-0.5, \prest) {$p_3$};
		\node at (-0.5, \pfour) {$p_4$};
		\node at (-0.5, \pfive) {$p_5$};
		
		\draw [dotted] (0,\pcrash) -- (15,\pcrash);
		\draw [dotted] (0,\prest) -- (15,\prest);
		\draw [dotted] (0,\pfour) -- (15,\pfour);
		\draw [dotted] (0,\pfive) -- (15,\pfive);
		
		\draw [|-|] (0, \pwriter) -- (15, \pwriter);
		\node at (7.5, \pwriter + 0.25) {$\writ(1)$};
		
		\draw [->] (1.5, \pwriter) -- (1.75, \pcrash);
		\draw [->] (1.75, \pcrash) -- (2, \pwriter);
		\node [font=\scriptsize] at (0.6, \pwriter / 2 + \pcrash / 2) {$\READREPLICA$};
		\node [font=\scriptsize] at (2.75, \pwriter / 2 + \pcrash / 2) {$\REPLY(t, 0)$};
		
		\draw [->] (1.5, \pwriter) -- (1.75, \prest);
		\draw [->] (1.75, \prest) -- (2, \pwriter);
		\node [font=\scriptsize] at (0.7, \pcrash / 2 + \prest / 2) {$\READREPLICA$};
		\node [font=\scriptsize] at (2.6, \pcrash / 2 + \prest / 2) {$\REPLY(t, 0)$};
		
		\draw [->] (6.5, \pwriter) -- (6.75, \pcrash);
		\draw [->] (6.75, \pcrash) -- (7, \pwriter);
		\node [font=\scriptsize] at (5.2, \pwriter / 2 + \pcrash / 2) {$\WRITEREPLICA(t', 1)$};
		\node [font=\scriptsize] at (7.5, \pwriter / 2 + \pcrash / 2) {$\SUCCESS$};
		
		\node at (8, \pcrash) {$\times$};
		
		
		\draw [->] (8.25, \pcrash) -- (8.5, \prest);
		\draw [->] (8.5, \prest) -- (8.75, \pcrash);
		\draw [->] (8.25, \pcrash) -- (8.5, \pfour);
		\draw [->] (8.5, \pfour) -- (8.75, \pcrash);
		\draw [->] (8.25, \pcrash) -- (8.5, \pfive);
		\draw [->] (8.5, \pfive) -- (8.75, \pcrash);
		\node [font=\scriptsize] at (7.3, \pcrash / 2 + \prest / 2) {$\READREPLICA$};
		\node [font=\scriptsize] at (9.5, \pcrash / 2 + \prest / 2) {$\REPLY(t, 0)$};
		\node [font=\scriptsize] at (7.4, \prest / 2 + \pfour / 2) {$\READREPLICA$};
		\node [font=\scriptsize] at (9.4, \prest / 2 + \pfour / 2) {$\REPLY(t, 0)$};
		\node [font=\scriptsize] at (7.5, \pfour / 2 + \pfive / 2) {$\READREPLICA$};
		\node [font=\scriptsize] at (9.4, \pfive / 2 + \pfour / 2) {$\REPLY(t, 0)$};
		
		\draw [->] (13.4, \pwriter) -- (13.6, \prest);
		\draw [->] (13.6, \prest) -- (13.8, \pwriter);
		\node [font=\scriptsize] at (12.1, \prest / 2 + \pcrash / 2) {$\WRITEREPLICA(t', 1)$};
		\node [font=\scriptsize] at (14.3, \prest / 2 + \pcrash / 2) {$\SUCCESS$};
		
	\end{tikzpicture}
	
	\caption{An illustration of safety violation
          in~\cite{DBLP:conf/ndss/DinisD023},
          	similar to Fig~\ref{fig:problem}
        }
	\label{fig:rr}
\end{figure}




\fi

\end{document}